\documentclass[12pt,reqno]{amsart}
\usepackage{amsthm}
\usepackage{amsmath}
\usepackage{amsfonts}
\usepackage{amssymb}
\usepackage{bbm,dsfont}
\usepackage{graphicx}
\usepackage{subcaption}
\usepackage{mathtools}
\usepackage{mathbbol}
\usepackage{hyperref}
\usepackage{enumerate}
\usepackage{multirow}
\usepackage{comment}
\usepackage{cleveref}
\usepackage[all]{xy}
\usepackage{tikz}
\usepackage{array}
\usepackage{makecell}
\usepackage[table]{xcolor}
\usepackage{appendix}
\usepackage[a4paper,margin=2.8cm]{geometry}

\newtheorem{thm}{Theorem}[section]
\newtheorem{cor}[thm]{Corollary}
\newtheorem{con}[thm]{Conjecture}
\newtheorem{lem}[thm]{Lemma}
\Crefname{lem}{Lemma}{lemmas}
\newtheorem{prop}[thm]{Proposition}
\theoremstyle{definition}
\newtheorem{defn}[thm]{Definition}
\newtheorem{rem}[thm]{Remark}
\newtheorem{eg}[thm]{Example}

\newcommand{\tr}{\text{tr}}
\newcommand{\ket}[1]{| #1 \rangle}
\newcommand{\bra}[1]{\langle #1|}
\newcommand{\ketbra}[1]{| #1 \rangle \langle #1|}
\newcommand{\ip}[2]{\langle #1|#2 \rangle}
\newcommand{\bea}{\begin{eqnarray}}
\newcommand{\eea}{\end{eqnarray}}
\newcommand{\id}{\mathrm{I}}
\newcommand{\real}{\mathds R}
\newcommand{\complex}{\mathds C}
\newcommand{\naturale}{\mathds N}
\newcommand{\one}{\mathds{1}} 
\newcommand{\Pb}{P_b}

\newcommand{\RomanNumeralCaps}[1]
    {\MakeUppercase{\romannumeral #1}}
\newcommand{\euler}{e}
\newcommand{\plusmod}{\oplus_m}
\newcommand{\polyone}{\mathds{1}_{\mathcal{P}(m)}}
\newcommand{\polyspace}{\mathcal{P}(m)}
\DeclareMathOperator{\complexi}{i}
\numberwithin{equation}{section}

\makeatletter
\patchcmd{\@maketitle}
  {\ifx\@empty\@dedicatory}
  {\par
   {\centering
   \small Faculty of Information Technology, University of Jyv\"askyl\"a, Finland\par}
   \ifx\@empty\@dedicatory}
  {}{}
\makeatother

\title[Nonclassical traits in multi-copy state discrimination]{Nonclassical traits \\ in multi-copy state discrimination}
\author[Achenbach, Lepp\"aj\"arvi, Lee, and Heinosaari]{Tim Achenbach, Leevi Lepp\"aj\"arvi, Hanwool Lee, and Teiko Heinosaari}
\address{Faculty of Information Technology, University of Jyv\"askyl\"a, Finland}

\begin{document}

\maketitle
\makeatletter \let\@setaddresses\relax \makeatother

\begin{abstract}
    Quantum state discrimination is a fundamental information processing task that serves as a key component in many applications while also carrying foundational significance.
    In this work, we consider minimum error discrimination of multi-copy states, where instead of preparing a single system we assume that multiple instances of the same state are prepared. Now the discrimination allows for measurements from multiple parties with different measurement strategies varying from global measurement strategy to ones restricted to different forms of local operations and classical communication strategies. By comparing the average success probabilities in quantum and classical cases, we find a qubit strategy that outperforms all the bit strategies. On the other hand, we show that the classical measurement strategy does not give benefit in qubit over bit.
    However, we find that there are other (qu)bit-like operational theories which can outperform the best qubit strategies even with a classical measurement strategy and we are able to identify instances of different theories where different measurement strategies are optimal. In this way, we are able to find instances of nonlocality without entanglement as well as provide general bounds for bit-like operational theories.
\end{abstract}

{\small
\tableofcontents
}

\newpage

\section{Introduction}

Quantum state discrimination \cite{Chefles01112000,barnett2008quantumstatediscrimination} is a central task in quantum information theory, finding motivation and application in many areas of the field \cite{Bae_2015}.
To be concise, the task in state discrimination is to identify a state drawn from an known ensemble with some probability. The goal is to achieve the maximum probability of success via correctly identifying (discriminating) the state, i.e., finding the label that is associated with it. 
In this way, information can be encoded and decoded using states.
This is a non-trivial task in general, since quantum states cannot be perfectly distinguished unless they are orthogonal.
State discrimination is interesting also from a foundational perspective as the performance and limits of state discrimination are shaped by the underlying theory. 
This means that one can try to single out quantum theory just by characterizing how good or bad states can be discriminated \cite{arai2024derivationstandardquantumtheory}.

Traditionally, state discrimination has been studied in a setting where only one system is sent to the receiver, who is then tasked to discriminate the state.
It would make sense, though, that chances of success are increased if multiple instances (copies) of the state in question are prepared and sent. This version of state discrimination is known as \textit{multi-copy state discrimination}.
The idea of providing multiple copies of a state to fulfill a task is not new and has been studied in different scenarios \cite{Acin01,roy2026quantumstateexclusioncopies} as well as in state discrimination \cite{Audenaert07,Higgins_2011}.
The question is not only about how many copies are sent. More importantly, one should ask: What are we going to do with them? Multiple copies pave the way to the complex field of measurement strategies.
It is natural to consider scenarios in which each copy is sent to a different party. What, then, is the best they can do under a given set of constraints?
In broad terms, their strategies can be divided into two classes. Spacelike separated parties measure their states individually and communicate with the others until someone makes the final guess. This is known under the terminology \textit{local operations and classical communication} (LOCC).
Or, they could stick together and perform one global, one-shot operation on the whole composite system. Hence, this is called \textit{global} operation.
We will discuss certain classes of strategies later. A more comprehensive study was done in \cite{Chitambar_2014}. The hope is that in the latter case, one could make use of correlations in the composite system and measurement setting. However, it is more difficult to implement global operations in practice as they require quantum memory and joint quantum operations.
It is thus interesting to find gaps between classes of strategies, which allow for evaluating tradeoffs, but also have foundational implications. This is because in classical theory, there is no difference between local and global operations as there is only one choice of measurement.

In this work, we investigate minimum‑error discrimination of multi-copy states with the following focus: What are the conditions to outperform classical theory in multi-copy state discrimination and what are the conditions for perfect discrimination in a given setting?
The first step involves finding bounds to classical theory, which will be the quantities to which we compare other theories. The first one will be quantum theory, which has to compete against classical theory in a bit vs. qubit comparison.
We extend our investigations to other operational theories, specifically looking into local strategies motivated by the question: Is it the local state space or the structure of the global measurements that allow for advantage over using classical system for the task?
Our results may later be used as benchmarks, for instance, by testing if a quantum protocol exceeds the classical bound.

\section{Summary of main results}
\noindent
\textbf{Setting} It is well-known that in the single copy case $k=1$ there is no quantum over classical advantage. In antidiscrimination \cite{roy2026quantumstateexclusioncopies} and channel discrimination \cite{Acin01} super-activation effects through multiple copies were found. Contrary to those, however,  no finite number of copies can achieve perfect discrimination, whereas infinitely many copies yield perfect knowledge of the state and hence perfect discrimination. It makes sense nevertheless, that multiple copies should increase the success probability and may yield a qubit over bit advantage. Moreover, it allows the comparison of different classes of strategies.

Suppose there is a set of $n$ 2-dimensional pure states $\{\varrho_i\}_{i=1}^n$, which is known to Alice and Bob. Alice randomly prepares $k$ instances of a state $\varrho_i$ and sends them to Bob, who is then asked to figure out the label $i$. Instead of just one party Bob, it is natural to think of Alice to send one state to $k$ parties respectively.
Now, they have to identify the state together by either following local operations and classical communication (LOCC) or collective (global) strategies.
They succeed if their final guess $j$ equals $i$. The goal is to maximize the average probability of success
\begin{equation}
    P(n,k)=\frac{1}{n} \sum_{i=1}^n \tr[\varrho_i^{\otimes k} M_i] \, ,
\end{equation}
where $M=\{M_i\}_{i=1}^n$ is a measurement that depends on the strategy pursued. This quantity will be our figure of merit. We also refer to the $k$ preparations of a state as copies. This task, as described, is named multi-copy state discrimination.

Among other things, we have the following three main results.

\noindent
\textbf{Result A: Quantum advantage}
We show that the classical (bit) success probability $P_b(n,k)$ is upper bounded by a function $\frac{g(k)}{n}$ and we give a quantum strategy of the form $\frac{f(k)}{n}$.
Even though we do not know the optimal qubit strategy $P_q(n,k)$ we get the following hierarchy of success probabilities:
\begin{equation}
    P_b(n.k) \leq \frac{g(k)}{n} < \frac{f(k)}{n} \leq P_q(n,k)
\end{equation}
Thus, showing a quantum advantage for any number of copies $1 < k < n$. We note that a quantum advantage was simultaneously found in \cite{Quintino_2026}.

\noindent
\textbf{Result B: Vanishing quantum advantage}
The above quantum advantage can be obtained using geometrically uniform states and the square root measurement, that is, a global measurement. Is the application of a global strategy the reason for the quantum advantage as there is no difference between LOCC and collective strategies in classical theory?
If we further restrict the strategy to be locally fixed (LF), i.e.\ every party performs the same measurement independently, then quantum theory is just as good as classical theory.
If one can successfully en- and decode maximally one bit in a state of some operational theory, we call that theory bit-like.
This idea is captured by the notion of information storability \cite{Heinosaari_2024}.
We show that all bit-like operational theories (ObT), e.g.\ quantum and classical theory, can achieve success probability at most $\frac{8}{9}$ in this setting given $k=2$ copies and $n=3$ states. By providing an explicit example it is shown that this bound is tight and it holds
\begin{equation}
    P_b^{LF}(3,2) = P_q^{LF}(3,2) = \frac{5}{6} <P_{ObT}^{LF}(3,2) \leq \frac{8}{9} < 1 \, .
\end{equation}
As a guiding example, we compare the performance of classical and quantum theory in the task of discriminating two copies of a trine state for various strategies. We furthermore derive necessary and sufficient conditions for perfect discrimination of $k=2$ copies and $n=3$ states in the LF setting for any operational theory.

\noindent
\textbf{Result C: Nonlocality without entanglement in Polygon theories}
In search for the root of advantages for multi-copy state discrimination, we consider Polygon theories.
Although a pure toy-theory, a Polygon is interesting because the information storability $\lambda_{max}$ of it is known or can be calculated. That is, all even polygons are bit-like and odd polygons can provide super-information storability, i.e. $\lambda_{max}$ surpasses the  operational dimension $d=2$ \cite{Heinosaari_2024}.
A Polygon theory is an operational theory, where the pure states $\{s_i\}_{i=1}^m$ are the vertices of a regular $m$-Polygon, where $m \geq 4$ is to be considered. A pure effect $e_i$ can be understood as projecting onto an edge of the polygon and thus, in a sense, 'picking a side'.

Here we are focusing on perfect discrimination of $k=2$ copies randomly drawn from an ensemble of $n=3$ states.
It turns out that Square and Hexagon are the only Polygons that allow for perfect discrimination using LOCC strategies. The key observation is that they contain at least 3 states that are pairwise perfectly distinguishable.
All other polygons lack such a 3-tuple, where every pair of states is pairwise distinguishable.
Interestingly, by solving optimization programs we find for polygons up to $m=15$ that separable strategies sometimes suffice for perfect discrimination in this setting, where there seems to be no clear pattern. For example, no LOCC protocol with one round of communication can lead to perfect discrimination in a Pentagon.
Though, here we found a perfect separable strategy, and since no entanglement is involved, this is an instance of nonlocality without entanglement.

\section{Multi-copy state discrimination}

Throughout the paper, we consider an ensemble of $n$ states, $S=\{\varrho_i\}_{i=1}^n$, described by density matrices, i.e., positive operators with trace equal to 1, that are prepared with equal prior probabilities, $\frac{1}{n}$. 
In the task of minimum-error discrimination \cite{barnett2008quantumstatediscrimination}, the aim is to minimize the average error probability, which is the same as maximizing the average success probability. 
The average success probability of discriminating ensemble $S$ by a measurement $M$ is given by
\begin{equation} \label{eq:average-succ-prob}
    P=\frac{1}{n}\sum_{i=1}^n \tr[\varrho_i M_i]\, ,
\end{equation}
where $M$ denotes a measurement described by a positive operator-valued measure (POVM), that is, a collection of positive operators $M_i \geq 0$, called effects, which sum up to the unit effect, $\sum_{i=1}^n M_i = \one$. \Cref{fig:multi-copy-task} illustrates the task of state discrimination. The upper bound for the success probability can be found in a straightforward way. Let $d$ denote the dimension of the Hilbert space. Then, 
\bea
P=\frac{1}{n}\sum_{i=1}^n \tr[\varrho_i M_i] \leq  \frac{1}{n}\sum_{i=1}^n \tr[\one_d \, M_i]=\frac{d}{n} \,,
\eea
which follows from $\varrho_i \leq \one_d$ for all $i \in [n] := \{1, \ldots,n\}$. Since this holds for any choice of ensemble and measurement, it gives the upper bound for the optimal success probability $P(n)$ of discriminating $n$ states, also known as the \textit{basic decoding theorem} \cite{Schumacher_Westmoreland_2010}.

In this work, we consider a slight modification of the state discrimination task. 
Suppose that we are provided more, say $k \in \naturale$ number of copies of the states. When $k$ number of identical copies of $\varrho_i$ are prepared, the ensemble is described as $S=\{\varrho_i^{\otimes k}\}_{i=1}^n$.
Now, the average success probability reads:
\begin{equation} \label{eq:average-multi-succ-rpob}
    P=\frac{1}{n} \sum_{i=1}^n \tr[\varrho_i^{\otimes k} \Tilde M_i] \, .
\end{equation}
The structure of the effects $\Tilde M_i$ depends on the measurement strategy that we use. \Cref{eq:average-multi-succ-rpob} will be our figure of merit and the quantity we use to compare the performance of different strategies and theories. If we optimize \Cref{eq:average-multi-succ-rpob} over all sets of $n$ states and all $n$-outcome global measurements we get the optimal average success probability $P(n,k)$ of discriminating $k$ copies of $n$ states.

\begin{figure}[!t]
    \centering
    \includegraphics[width=0.7\linewidth]{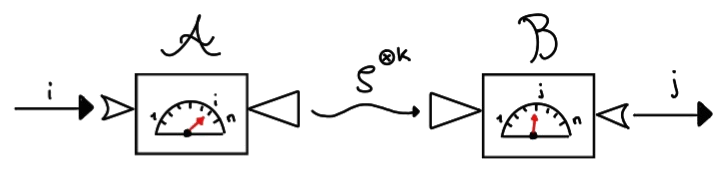}
    \caption{Alice is randomly asked to prepare a state with label $i$. She sends $k$ instances of that state to Bob, to whom the label is unknown. He tries to figure out the label, measuring the states and finally outputting his guess $j$. If $j=i$, they succeed. The measurement strategies can vary from Bob making a collective global measurement on all copies to making individual local measurements on the different copies.}
    \label{fig:multi-copy-task}
\end{figure}

Intuitively, it makes sense that the chances of successfully discriminating ensemble increase with growing $k$. 
The first question is that can we achieve perfect discrimination by simply adding copies? Let us first make a simple general observation: whenever $n>d$, where $d$ is the dimension of the underlying Hilbert space, we have that $P(n,k)<1$ for all $k \in \naturale$. This is simply because probability 1 can be reached only in the case when the $n$ $k$-copy states are orthogonal. Since now $P(n,1) \leq \frac{d}{n}<1$, the local states cannot be orthogonal and thus simply adding more copies cannot make them orthogonal either. This means that \emph{(both) in quantum (and classical) theory whenever $n>d$ we cannot perfectly distinguish the states with any finite number of copies}. 

However, it is known that in the asymptotic case, i.e., for infinitely many copies, perfect discrimination is possible. This becomes clear simply from the observation, that the overlap of non-identical states decreases exponentially with the number of copies and in this case the states become orthogonal at the limit. 
In particular, it was shown by \cite{Barnum:2000cwz} that for a given set of states $\{ \varrho_i\}_{i \in [n]}$, the average success probability can be lower bounded using the fidelity between states. The authors in \cite{Harrow_2012} took a step further and reconsidered this bound for multiple copies:
\begin{equation}
    P(n,k) \geq 1 - \frac{1}{n} \sum_{i \neq j} \sqrt{F(\varrho_i^{\otimes k}, \varrho_j^{\otimes k})} \, , \quad i,j \in [n] \, .
\end{equation}
The factor $\frac{1}{n}$ stems from the uniform prior distribution and $F(\varrho, \sigma) = \tr(\sqrt{\sqrt{\varrho} \sigma\sqrt{\varrho}})^2$ is the fidelity between two states $\varrho$ and $\sigma$. If $\varrho_i \neq \varrho_j$, the fidelity is strictly below 1 and it follows
\begin{equation}
    F(\varrho_i^{\otimes k}, \varrho_j^{\otimes k}) = F(\varrho_i, \varrho_j)^k \overset{k \rightarrow \infty}{\longrightarrow} 0 \, ,
\end{equation}
which means that the states become distinguishable. 
As shown in \cite{Harrow_2012}, any precision $\epsilon >0$ can be obtained to yield $P(n,\hat k) \geq 1-\epsilon$ for a sufficiently large $\hat k$.

In other tasks like state exclusion (also called  antidiscrimination) and channel discrimination, it was shown that already finitely many copies are enough to fulfill the task perfectly \cite{Acin01,roy2026quantumstateexclusioncopies}. Such super-activation effects are very interesting, but not always guaranteed. For example, saturation effects might occur in entanglement transformation \cite{PhysRevA.65.052315}.
The discrimination task of multi-copy states has richer picture as different measurement classes can be considered \cite{Chitambar_2014}. In some cases, global operations can outperform local operations. Since the states contain no entanglement, this phenomenon is called \textit{nonlocality without entanglement (NLWE)}\cite{Bennet_NLWE}. 

In addition to exploring features like NLWE, in this work we are interested in comparing different quantum strategies to optimal classical strategies. In particular we will focus on the simplest quantum and classical systems, i.e., the two-level systems of qubit and bit. We start our investigation by deriving bounds and exact forms of optimal success probability when we use a bit as the information carrier.

\section{Optimal bit strategies}\label{sec:classical}
 By classical theory we mean a restriction of quantum theory such that one is only allowed to use commuting states and measurements with commuting effects. Then bit states can be written as  $\varrho_i = p_i \ketbra{0}+(1-p_i)\ketbra{1}$ for some $p_i \in [0,1]$ for all $i \in [n]$. The effects of the discriminating $k$-bit measurement $M$ on the $k$ copies of states can be written in the form 
\begin{align}
     M_j = \sum_{i_1=0}^1 \cdots \sum_{i_k=0}^1 \nu_{j|i_1, \ldots, i_k}\ketbra{i_1 \cdots i_k}
\end{align}
 for all $j \in [n]$ for some $ \nu_{j|i_1, \ldots, i_k} \geq0$ such that $\sum_{j=1}^n  \nu_{j|i_1, \ldots, i_k} = 1$ for all $i_1, \ldots, i_k \in \{0,1\}$. Thus, the $k$-bit measurement consists of $k$ basis measurements $\{\ketbra{0}, \ketbra{1}\}$ on each bit and a joint classical stochastic postprocessing of the measurement outcomes described by $\nu$. For the optimal success probability $\Pb$ in bit we can show the following upper bound:

\begin{thm}\label{prop:classical-n-k-bound}
    For bit system the optimal success probability $\Pb(n,k)$ for discriminating $n \geq 3$ states of $k$ copies with equal a priori probability has the following upper bound
    \begin{align}\label{eq:classical-n-k-ub}
        \Pb(n,k) \leq \frac{1}{n} \left[ 2 + \sum_{j=1}^{k-1} {k \choose j} \left(\frac{j}{k}\right)^j \left(1-\frac{j}{k}\right)^{k-j} \right] \, .
    \end{align}
    This bound can be achieved always when $n > k$ and requires the use of the states $\{\varrho_i = p_i \ketbra{0}+(1-p_i)\ketbra{1}\}_{i=1}^n$ with $p_1 = 0$, $p_{k+1} = 1$ and $p_j = \frac{j-1}{k}$ for all $j \in \{2, \ldots, k\}$ and the rest of the $n-k-1$ states can be chosen arbitrarily.
\end{thm}
\begin{proof}
    Let us consider an optimal strategy involving $n$ states $\{\varrho_j\}_{j=1}^n$ and an $n$-outcome measurement $M$ on the $k$ copies of states consisting of effects $\{M_j\}_{j=1}^n$ as described above. Now the average success probability reads as
    \begin{align}
        \Pb(n,k) &= \frac{1}{n} \sum_{i=1}^n \tr[M_i \varrho_i^{\otimes k}]  \\
        &= \frac{1}{n} \sum_{i=1}^n \sum_{i_1=0}^1 \cdots \sum_{i_k=0}^1\nu_{i|i_1, \ldots, i_k}  \ip{i_1, \ldots, i_k}{(p_i \ketbra{0} + (1-p_i) \ketbra{1})^{\otimes k}| i_1, \ldots, i_k }   \\
        &= \frac{1}{n} \sum_{i=1}^n \sum_{i_1=0}^1 \cdots \sum_{i_k=0}^1\nu_{i|i_1, \ldots, i_k} \prod_{j=1}^k \left( p_i \delta_{i_j, 0} +(1-p_i)\delta_{i_j, 1}  \right)  \\
        &\leq \frac{1}{n} \sum_{i=1}^n \sum_{i_1=0}^1 \cdots \sum_{i_k=0}^1\nu_{i|i_1, \ldots, i_k} \max_{i'} \prod_{j=1}^k \left( p_{i'} \delta_{i_j, 0} +(1-p_{i'})\delta_{i_j, 1}  \right)  \\
        &= \frac{1}{n} \sum_{i_1=0}^1 \cdots \sum_{i_k=0}^1 \max_{i'} \prod_{j=1}^k \left( p_{i'} \delta_{i_j, 0} +(1-p_{i'})\delta_{i_j, 1}  \right)   
    \end{align}
    \begin{align}
        &= \frac{1}{n} \sum_{j=0}^k {k \choose j} \max_{i} p_{i}^j (1-p_{i})^{k-j}  \\
        &= \frac{1}{n} \left(\max_{i} (1-p_{i})^k  + \sum_{j=1}^{k-1} {k \choose j} \max_{i} p_{i}^j (1-p_{i})^{k-j} + \max_{i} p^k_{i} \right) \, . \label{eq:classical-n-k-ub0}
    \end{align}
    Clearly the first and the last terms are maximized when $p_{i'} =0$ and $p_{i''} = 1$ for some $i', i'' \in [n]$ respectively. For the rest of the terms we can look the function $f_{j,k}: [0,1] \to [0,1]$ defined as $f_{j,k}(x) = x^j (1-x)^{k-j}$ for all $x \in [0,1]$ for all $j \in [k-1]$. We note that $f_{j,k}(0) = f_{j,k}(1) = 0$ and that $f_{j,k}(x)> 0$ for all $x \in (0,1)$ for all $j \in  [k-1]$ so that the function is not maximized at $x=0$ or $x=1$. We now have that 
    \begin{align}
        \frac{d f_{j,k}(x)}{dx} = [j(1-x) - (k-j)x] x^{j-1} (1-x)^{k-j-1}
    \end{align}
    which equals zero on the interval $(0,1)$ if and only if $j(1-x) - (k-j)x=0$ which gives us that the maximum of $f_{j,k}$ is achieved at $x=\frac{j}{k}$. Thus, we finally have that 
    \begin{align}
        \Pb(n,k) &\leq \frac{1}{n} \left(\max_{i} (1-p_{i})^m  + \sum_{j=1}^{k-1} {k \choose j} \max_{i} p_{i}^j (1-p_{i})^{k-j} + \max_{i} p^k_{i} \right) \\
        &\leq \frac{1}{n} \left[ 2 + \sum_{j=1}^{k-1} {k \choose j} \left(\frac{j}{k}\right)^j \left(1-\frac{j}{k}\right)^{k-j} \right] \, .
    \end{align}
    Now, if $n > k$ we can set $p_1 = 0$, $p_{k+1} = 1$ and $p_j = \frac{j-1}{k}$ for all $j \in \{2, \ldots, k\}$ and for the measurement we can set  $\nu_{j|i_1, \ldots, i_k} = 1$ if there are exactly $j-1$ indices in $\{i_1, \ldots, i_k\}$ which are zero and otherwise $\nu_{j|i_1, \ldots, i_k} = 0$. One can readily confirm that with these choices the upper bound is achieved. We note that to reach this bound, the choice of these $k+1$ states and the measurement is unique (up to bijective relabeling) but that we can choose the remaining $n-k-1$ states arbitrarily.
\end{proof}

We postpone plotting the bound given in \Cref{prop:classical-n-k-bound} until the next section where we compare it to some qubit discrimination strategies. For the future comparison we also present the case of $k=2$ copies as a separate corollary:

\begin{cor} \label{cor:class-optimum-n3-k2}
        For bit system the optimal success probability for discriminating $n \geq 3$ states of $k=2$ copies with equal a priori probability is 
        \begin{align}
            \Pb(n,2) = \frac{5}{2n} \, .
        \end{align}
        The optimal strategy requires the use of the states $\varrho_1=\ketbra{1},\varrho_2=\frac{1}{2}(\ketbra{0}+\ketbra{1}), \varrho_3=\ketbra{0}$ (the remaining states can be arbitrary) and a measurement $M$ with effects $M_1 = \ketbra{11}$, $M_2 = \ketbra{01}+\ketbra{10}$, $M_3 = \ketbra{00}$ and $M_j =0$ for $j>3$.
\end{cor}

On the other hand, if we fix our number of states to be $n=3$ but allow arbitrary number of copies, we can show the following.

\begin{prop}\label{prop:classical-3-k-exact}
    For bit system the optimal success probability for discriminating $n=3$ states of $k$ copies with equal a priori probability is 
    \begin{align}\label{eq:classical-3-k-exact}
            \Pb(3,k) = 1 - \frac{1}{3 \cdot 2^{k-1}} \, .
    \end{align}
    The optimal strategy requires the use of the states $\varrho_1=\ketbra{1},\varrho_2=\frac{1}{2}(\ketbra{0}+\ketbra{1}), \varrho_3=\ketbra{0}$ and a measurement $M$ with effects $M_1 = \ketbra{1^{\otimes k}}$, $M_3 = \ketbra{0^{\otimes k}}$ and $M_2 = \one_{2^k} - M_1 - M_3$.
\end{prop}
\begin{proof}
    As was shown in the proof of   \Cref{prop:classical-n-k-bound}, by using the optimal states $\{\varrho_i = p_i \ketbra{0}+(1-p_i)\ketbra{1}\}_{i=1}^3$ we have the following upper bound 
    \begin{align}\label{eq:m-3-2-ub}
        \Pb(3,k) \leq \frac{1}{3} \left(\max_{i} (1-p_{i})^k  + \sum_{j=1}^{k-1} {k \choose j} \max_{i} p_{i}^k (1-p_{i})^{k-j} + \max_{i} p^k_{i} \right) \, .
    \end{align}
Unlike in the case when $n> k$ we cannot choose all the $p_i$'s which maximize all the terms in the sum separately since we have more terms in the sum than we have states. Instead we have to limit to three $p_i$'s which give us the maximum for the whole expression and then show that we can achieve that. 

First we see that (without the normalizing factor) the first and the last term in the above expression are clearly less or equal than one. In fact they are one only when two of the states are $\ketbra{1}$ and $\ketbra{0}$ respectively. On the other hand, using the bound from   \Cref{prop:classical-n-k-bound} we see that for all $k \in [m-1]$
\begin{align}
     {k \choose j} \max_{i} p_{i}^j (1-p_{i})^{k-j} \leq  {k \choose j} \left(\frac{j}{k}\right)^j \left(1-\frac{j}{k}\right)^{k-j} 
    < \sum_{i=0}^k {k \choose i} \left(\frac{j}{k}\right)^i \left(1-\frac{j}{k}\right)^{k-i} 
    =1 \, .
\end{align}
Thus, none of the other terms in the sum in   \Cref{eq:m-3-2-ub} (without the normalizing factor) can attain the value one. Hence, the largest value in the upper bound in   \Cref{eq:m-3-2-ub} is attained when (up to relabeling) we take $p_1 = 0$ and $p_3 = 1$ and we choose $p_2$ such that the term $\sum_{j=1}^{k-1} {k \choose j} p_{2}^j (1-p_{2})^{k-j}$ is maximized. 
In order to do that let us consider the function $g_m: [0,1] \to [0,1]$ defined as $g_k(x) = \sum_{j=1}^{k-1} {k \choose j} x^j (1-x)^{k-j}$ for all $x \in [0,1]$. 
First, we note that $g_k(0) = g_k(1)=0<g_k(x)$ for all $x \in (0,1)$. 
From the binomial formula we know that 
\begin{align}
    g_k(x) &=  \sum_{j=1}^{k-1} {k \choose j} x^j (1-x)^{k-j} = 1-x^k-(1-x)^k
\end{align}
for all $x \in (0,1)$. 
Now we see that
\begin{align}
    \frac{d g_k(x)}{dx} &=  -kx^{k-1}+k(1-x)^{k-1}
\end{align}
which is zero only when $x=\frac{1}{2}$ at which point $g_k$ attains its maximum value on $[0,1]$. Thus, if we set $p_2=\frac{1}{2}$, we get that
\begin{align}
        \Pb(3,k) &\leq \frac{1}{3} \left(\max_{i} (1-p_{i})^k  + \sum_{j=1}^{k-1} {k \choose j} \max_{i} p_{i}^k (1-p_{i})^{k-j} + \max_{i} p^k_{i} \right) \\
        & \leq \frac{1}{3} \left[ 1 + g_k\left( \frac{1}{2} \right) + 1 \right] = \frac{1}{3} \left[ 2 + 1-2 \left( \frac{1}{2} \right)^m  \right] \\
        &= 1 - \frac{1}{3 \cdot 2^{k-1}} \, .
\end{align}
It can be readily confirmed that this bound is achieved when we take $p_1=0$, $p_2 = \frac{1}{2}$, $p_3=1$ for the states (as previously noted) and when we set $\nu_{1|1, \ldots, 1} = \nu_{3| 0, \ldots, 0} = \nu_{2|i_1, \ldots, i_k}=1$ for all $i_1, \ldots, i_k \in \{0,1\}$ such that $i_j \neq i_{j'}$ for at least some $j \neq j'$. For all other indices we set $\nu_{i|i_1, \ldots, i_k}=0$. This optimal strategy is unique.
\end{proof}

We also postpone plotting the optimal success probability given in \Cref{prop:classical-3-k-exact} until we make the comparison with qubit.

\begin{rem}
    Both in   \Cref{prop:classical-n-k-bound} and in   \Cref{prop:classical-3-k-exact} the optimal strategies use states which are distributed on the bit in equal intervals. However, this is not the case in general. Namely, we can saturate the upper bound of   \Cref{eq:classical-n-k-ub0}  with different strategy when $3< n\leq k$. Let us demonstrate this in the simplest case when $n=k=4$. Now in   \Cref{eq:classical-n-k-ub0} there are in total five terms which we should aim to maximize with four states. As is observed in the proof of   \Cref{prop:classical-3-k-exact} only two of the five terms can reach one (without the normalization). These are the first and the last term which reach one when we take $p_1=0$ and $p_4=1$. Now the remaining two states should be chosen such that together they maximize the remaining three terms in the sum. Thus, the only options are that we choose $p_2$ and $p_3$ such that
    \begin{align}
        p_2 = \arg \max_{0<x<1} \left[ {4 \choose 1} x^1 (1-x)^{3} +{4 \choose 2} x^2 (1-x)^{2} \right], \  p_3 = \arg \max_{0<x<1}  {4 \choose 3} x^3 (1-x)^{1},
    \end{align}
    or 
    \begin{align}
        p_2 = \arg \max_{0<x<1} \left[ {4 \choose 1} x^1 (1-x)^{3} +{4 \choose 3} x^3 (1-x)^{1} \right], \  p_3 = \arg \max_{0<x<1}  {4 \choose 2} x^2 (1-x)^{2},
    \end{align}
    or
    \begin{align}
        p_2 = \arg \max_{0<x<1} \left[ {4 \choose 2} x^2 (1-x)^{2} +{4 \choose 3} x^3 (1-x)^{1} \right], \  p_3 = \arg \max_{0<x<1}  {4 \choose 1} x^1 (1-x)^{3} \, .
    \end{align}
    It is straightforward to show that the overall maximum for   \Cref{eq:classical-n-k-ub0} is achieved when we choose the first option (which is equivalent to the last option) and in that case we take $p_2= \frac{1}{2}(\sqrt{3}-1) \approx 0.366$ and $p_3 = \frac{3}{4}$. As in the proof of   \Cref{prop:classical-n-k-bound} this bound can be achieved by a measurement $M$ for which we set  $\nu_{j|i_1, \ldots, i_4} = 1$ if there are exactly $j-1$ indices in $\{i_1, \ldots, i_4\}$ which are zero and otherwise $\nu_{j|i_1, \ldots, i_4} = 0$. Using this optimal strategy we find that 
    \begin{align}
        \Pb(4,4) = \frac{1}{4}\left( 3 \sqrt{3} -\frac{133}{64}\right) \approx 0.780 \, .
    \end{align}
    We note that we could use the same method to saturate the bound given by   \Cref{eq:classical-n-k-ub0} in other cases of $3< n\leq k$ as well. However, in this work we will focus on comparing qubit strategies to the optimal bit strategies given in   \Cref{prop:classical-n-k-bound} and   \Cref{prop:classical-3-k-exact}.
\end{rem}

\section{Quantum benefits in multi-copy discrimination}
\label{sec:quantum-part}

In discrimination task for single-copy states, the achievable upper bound for the success probability is $\frac{d}{n}$. The same bound holds for both quantum and classical theories, and therefore quantum theory does not offer any benefits.
Their similarity in this maximal success probability sense persists in unambiguous discrimination and antidiscrimination tasks \cite{heinosaari2024can}. It is therefore not at all evident that a qubit would offer and advantage over a bit.
Interestingly, we find that in multi-copy settings, quantum theory can yield higher success probability than classical theory.

In the following, we derive the upper bound for the success probability in quantum theory which can go beyond the classical limit. We then show an instance in which this bound is attained. 
\begin{prop} \label{prop:pure-states-bound}
    When the ensemble consists of $k$ copies of pure qubit states, the success probability is upper bounded by
    \bea
    P_{q,pure}(n,k) \leq\frac{k+1}{n}
    \eea
\end{prop}
\begin{proof}
    Let us denote $\Pi_{sym}$ as the projection onto the symmetric subspace. For a pure state $\varrho_i= \ketbra{\psi_i}$, it holds that $\ketbra{\psi_i}^{\otimes k} \leq \Pi_{sym}$, and therefore
    \begin{align}
            P_{q,pure}(n,k)&=\max_{\{\ket{\psi_i}\}, \{M_i\}} \frac{1}{n}\sum_i \tr[\ketbra{\psi_i}^{\otimes k} M_i] \\
            &\leq \frac{1}{n} \sum_i \tr[\Pi_{sym} M_i]=\frac{1}{n}\tr[\Pi_{sym}]=\frac{k+1}{n}.
    \end{align}
\end{proof}
This upper bound is not tight in general. In particular, the bound is trivial for $k+1\geq n$. 
Next we present an instance where this bound is achieved. Consider $k=2$ copies of pure states that form a tetrahedron in Bloch sphere, defined as 
\begin{equation}
\begin{aligned}
   |\psi_1\rangle &=|0\rangle, \quad  |\psi_2\rangle=\frac{1}{\sqrt{3}}|0\rangle+ \sqrt{\frac{2}{3}}|1\rangle, \\
|\psi_3\rangle&=\frac{1}{\sqrt{3}}|0\rangle+ \sqrt{\frac{2}{3}}e^{\frac{2\pi i}{3}}|1\rangle, \quad
|\psi_4\rangle=\frac{1}{\sqrt{3}}|0\rangle+ \sqrt{\frac{2}{3}}e^{\frac{-2\pi i}{3}}|1\rangle. 
\end{aligned}
\end{equation}
By using the \textit{pretty good measurement (PGM)} \cite{Hausladen01121994}, defined formally later in \Cref{eq:PGM-1} and \Cref{eq:PGM-k}, we get the average success probability $P=\frac{3}{4}$, which is the upper bound in \Cref{prop:pure-states-bound} in the case of $n=4$ and $k=2$.

In general, it is hard to find the optimal quantum strategy for a given setup in state discrimination. Our approach is to start with an educated guess.
Intuitively, one would think that the optimal strategy contains states, which exhibit some symmetries that can be exploited in state discrimination tasks. A well-known example are \emph{geometrically-uniform (GU) states }\cite{Zhou_2025}. For GU states it is known that that the optimal discriminating measurement is the PGM.

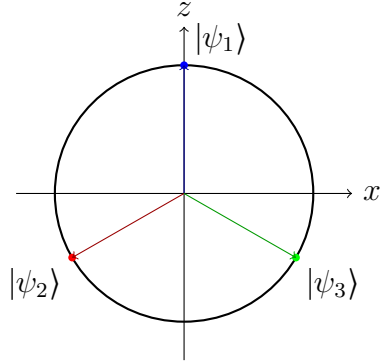
\begin{figure}
    \centering
        \begin{tikzpicture}[scale=1.7]
            \draw[thick] (0,0) circle (1);
            
            \draw[->] (-1.3,0) -- (1.3,0) node[right] {$x$};
            \draw[->] (0,-1.3) -- (0,1.3) node[above] {$z$};
            
            \foreach \ang/\name/\col/\labelpos in {
            90/1/blue/right,
            210/2/red/below left,
            330/3/green/below right}
            {
            \coordinate (P\name) at ({cos(\ang)},{sin(\ang)});
            \fill[\col] (P\name) circle (0.03);
            \draw[->, \col!60!black] (0,0) -- (P\name);
            }
            
            \node[above right] at (P1) {$\ket{\psi_1}$};
            \node[below left] at (P2) {$\ket{\psi_2}$};
            \node[below right] at (P3) {$\ket{\psi_3}$};
            
        \end{tikzpicture}
    \caption{The trine states are the simplest non-trivial case of GU-states. All three states share a $120^\circ$ angle between them. The figure shows a projection of the Bloch sphere onto the $z-x$-plane.}
    \label{fig:trine-states}
\end{figure}

We use the optimality of PGM for GU states \cite{JOUR, Zhou_2025,Optimal_detec_sym_states} as a case study for quantum multi-copy state discrimination.
In general, GU-states can be defined as states generated by group actions and generator states \cite{Zhou_2025,Optimal_detec_sym_states}.
\begin{defn}
    Let $\mathcal{G}$ be a finite group on a complex space $\mathds{C}^d$ with $n$ unitary group elements $g_i \in \mathcal{G}$. An ensemble of GU states is the set
    \begin{equation}
        \left\{ \frac{1}{n} ,\varrho_i\right\}_{i=1}^n \, ,\quad \varrho_i = g_i \varrho_0 g_i^\dagger
    \end{equation}
    for some generator state $\varrho_0 \in \mathds{C}^d$. Since we are only considering uniform prior distributions, we keep the factor $\frac{1}{n}$ implicit for most of the time.
\end{defn}
\begin{eg}
The trine ensemble consists three states, that lie uniformly distributed in the Bloch plane representation \Cref{fig:trine-states} and are defined as
\begin{equation} \label{eq:trine-states}
    \ket{\psi_1}=\ket{0}, \quad \ket{\psi_2}=\frac{1}{2}\ket{0}+\frac{\sqrt{3}}{2}\ket{1}, \quad  \ket{\psi_3}=\frac{1}{2}\ket{0}-\frac{\sqrt{3}}{2}\ket{1}.
\end{equation}
Given two copies of trine states, the ensemble is known as double trine \cite{eric_optimal_detec}.
\end{eg}

For states $\{\varrho_i\}_{i=1}^n$ the PGM is $M$ with
\begin{equation}\label{eq:PGM-1}
    M_i = \frac{1}{n} \sqrt{\varrho^{-1}} \varrho_i \sqrt{\varrho^{-1}} \, , \quad \textrm{where} \quad \varrho=\frac{1}{n} \sum_{i=1}^n \varrho_i \, .
\end{equation}
Here $\varrho^{-1}$ denotes the Moore-Penrose inverse of $\varrho$. 
We note that $k$ copies of GU states are still GU states so that the PGM is still the optimal discriminating measurement. If one is tasked with discriminating $k$-copies of $n$ GU-states $\{\varrho_i^{\otimes k} \}_{i=1}^n$, the PGM straightforwardly reads:
\begin{equation}\label{eq:PGM-k}
     M_i = \frac{1}{n} \sqrt{(\varrho^{(k)})^{-1}} \varrho_i^{\otimes k} \sqrt{(\varrho^{(k)})^{-1}} \text{ and } \varrho^{(k)}=\frac{1}{n} \sum_{i=1}^n \varrho_i^{\otimes k} \, .
\end{equation}
The average success probability is thus calculated as:
\begin{equation} \label{eq:normal-PGM-success-Prob}
    P_{GU}(n,k) = \frac{1}{n} \sum_{i=1}^n \tr \left[\varrho_i^{\otimes k} \sqrt{(\varrho^{(k)})^{-1}} \varrho_i^{\otimes k} \sqrt{(\varrho^{(k)})^{-1}} \right]
\end{equation}
In general, this expression can get quite complicated and difficult to compute. However, for pure states the exact expressions and bounds have been obtained in the literature by considering the Gram matrix of the state ensemble instead.

Gram matrices proved to be useful for the study of state discrimination and state exclusion \cite{Zhou_2025, Montanaro_2007, Johnston_2025}. They allow for a characterization of a pure state ensemble solely based on the inner products of the contained states. Given a set of $d$-dimensional  pure states $\{\ket{\varphi_i}\}_{i=1}^n$ one defines a state matrix $A = (\ket{\tilde{\varphi}_1}, \ldots,\ket{\tilde{\varphi}_n})$, where $\ket{\tilde{\varphi}_i} := \frac{1}{\sqrt{n}}\ket{\varphi_i}$ for all $ i \in [n]$, and the Gram matrix of the ensemble is defined as $G=A^\dagger A$.
Thus, the Gram matrix contains all possible inner products of the subnormalized states $\{\ket{\tilde{\psi}_i}\}_{i=1}^n$. For $k$-copies of pure states we have that:
\begin{equation}
    G^{(k)}=(G_{ij}^{(k)})_{ij}, \text{ where } G_{ij}^{(k)} = \frac{1}{n} \langle \varphi_i|\varphi_j\rangle^k \quad \forall i,j \in [n]\, .
\end{equation}
It is clear that $G^{(k)}$ is symmetric, positive semidefinite, $\tr[G^{(k)}]=1$ and $G^{(k)} \leq \one$ for all $k \in \naturale$. The latter fact can be seen by diagonalizing $G^{(k)}$ and showing that $\one - G^{(k)} \geq 0$ since the states $\ket{\varphi_i}$ are renormalized by $\frac{1}{\sqrt{n}}$.
In general, the Gram matrix approach for GU states leads to the following results \cite{Zhou_2025, Montanaro_2007}.
\begin{equation} \label{eq:success-prob-for-GU-formula}
    P_{GU}(n,k)=\frac{1}{n} [\tr(\sqrt{\varrho^{(k)}}) ]^2 = \sum_{i=1}^n (\sqrt{G^{(k)}})_{ii}^2
\end{equation}
Moreover, the average state $\varrho^{(k)}$ and the Gram matrix $G^{(k)}$ share the same nonzero eigenvalues \cite{Jozsa_2000}.

We will now consider a specific class of qubit GU states which are generated by rotational symmetries similarly to the trine states. 

\begin{defn}
    For given $n \in \naturale$, we define specific class of GU-states, called the \emph{cyclic geometrically-uniform (CGU) states }, as the set of the following pure, qubit states generated by the cyclic rotation groups on the $z-x$-plane of the Bloch sphere:
\begin{equation} \label{eq:GU-states}
    \{ \ket{\psi_i} \}_{i=1}^n, \quad \mathrm{where} \, \ket{\psi_i} = \cos{\left(\tfrac{\pi i}{n}\right)} \ket 0 + \sin{\left(\tfrac{\pi i}{n}\right)} \ket 1 \, .
\end{equation}
\end{defn}
We note that since CGU states are a particular subset of GU states, for the optimal success probabilities we have that $P_{GU}(n,k) \geq P_{CGU}(n,k)$ for all $n, k$.
Now, in the case of $k$ copies of $n$ CGU-states \Cref{eq:GU-states}, the Gram matrix evaluates to
\begin{equation}
    G_{ij}^{(k)} = \frac{1}{n} \cos{\left(\frac{\pi}{n}(i-j)\right)}^k \, .
\end{equation}
Although the Gram matrix and \Cref{eq:success-prob-for-GU-formula} give us means to calculate the optimal success probability, it is hard to find a closed analytic form for general $n$ and $k$ even for our chosen set of CGU states. 
However, for CGU states we find that the success probability is inversely proportional to $n$ for all $k$, that is, $P_{CGU}(n,k) \propto \frac{1}{n}$.

\begin{table}[!t]
    \centering
    \begin{tabular}{| c || c |} 
        \multicolumn{2}{c}{} \\ \hline
        $k$ & $n \cdot P_{CGU}(n,k)=f(k)$\\
        \hline
        2 & $\frac{3}{2}+\sqrt{2} \approx 2.91$ \\
        3 & $2+\sqrt{3} \approx 3.73$ \\
        4 & $\frac{3(7+2\sqrt{6})}{8} \approx 4.46$ \\
        5 & $\frac{8+5\sqrt{2} +\sqrt{5}+\sqrt{10} }{4} \approx 5.12$ \\
        6 & $\frac{27+10\sqrt{3}+2\sqrt{5}+2\sqrt{6}+6\sqrt{10}+2\sqrt{15}+2\sqrt{30}}{16} \approx 5.71$\\
        7 & $\frac{(\sqrt{2}+\sqrt{14}+\sqrt{42}+\sqrt{70})^2}{64} \approx 6.25$\\ \hline
    \end{tabular}
    \caption{Values of $f(k)$ for $k=2, \ldots,7$. The success probability can be obtained by dividing with $n$ for which $n > k$ must hold.}
    \label{table:values-f(k)}
\end{table}

\begin{prop}\label{prop:GU-prop-to-1-n}
    The optimal success probability of discriminating $k$ copies of $n>k$ CGU states is of the form $P_{CGU}(n,k)=\frac{f(k)}{n}$, where $f(k)$ is a number only depending on $k$.
\end{prop}
We postpone the proof to \Cref{appendix:GU-f(k)}.
Computing $f(k)$ for given $k$ can be done in principle, but is in general hard to do because it requires diagonalizing a $2^n \times 2^n$ matrix. 
We list the values of $f(k)$ for $k=2,\ldots,7$ in \Cref{table:values-f(k)}.

\Cref{prop:GU-prop-to-1-n} leads to an interesting observation: Since the classical upper bound $\Pb(n,k):=\frac{g(k)}{n}$ from \Cref{prop:classical-n-k-bound} and $P_{CGU}(n,k)$ are of similar form, a comparison of classical and quantum theory seems to be very inviting.
A comparison of $f(k)$ and $g(k)$ in \Cref{fig:f-vs-g} further suggests a quantum over classical advantage.
However, for a more in-depth comparison, a simpler expression than $f(k)$ is needed. We therefore look for a lower bound that is in between $f(k)$ and $g(k)$. Obtaining this lower bound also gives a lower bound on $P_{GU}(n,k)$ as well.

\begin{prop}\label{prop:lower-GU-bound}
    The optimal success probability of discriminating $k$ copies of $n>k$ GU states is lower bounded by 
    \bea
    P_{GU}(n,k) \geq P_{CGU}(n,k) \geq \frac{2^{2k}}{n} \binom{2k}{k}^{-1} =: \frac{h(k)}{n}
    \eea
\end{prop}
\begin{proof}
    From \cite{Montanaro_2007} we take that the success probability of discriminating random, i.e.\ uniformly distributed, states is lower bounded by
    \begin{equation}
        P_{CGU}(n,1) = \sum_{i=1}^n (\sqrt{G^{(1)}})_{ii}^2 \geq \sum_{i=1}^n \frac{(G^{(1)}_{ii})^3}{\sum_{j=1}^n |G^{(1)}_{ij}|^2} = \frac{1}{n} \sum_{i=1}^n \frac{1}{\sum_{j=1}^n \lvert \ip{\psi_i}{\psi_j} \rvert^2}.
    \end{equation}
    In our case the denominator evaluates to:
    \begin{equation}
        \sum_{j=1}^n \lvert \ip{\psi_i}{\psi_j} \rvert^2 = \frac{n}{2} + \frac{1}{2} \sum_{j=1}^n \cos \left( \frac{2 \pi}{n}(i-j) \right)
    \end{equation}
    One realizes that the latter sum is the real part of $\sum_{j=1}^n \exp[{\frac{2\pi i}{n}(\Tilde i - j)}]$, which can be shown to equal 0 using the periodicity of the complex exponential function if $n > 1$.
    Imaginary and real part are hence zero and we are thus left with $P_{GU}(n,1) \geq \frac{2}{n}$, which is twice as good as guessing. Now, let us see what happens if we add more copies.
    Here, calculations from the proof of \Cref{prop:GU-prop-to-1-n} can help us out to show that
    \begin{align}
        \sum_{j=1}^n \lvert \ip{\psi_i}{\psi_j}^{\otimes k}\rvert^2 &= \sum_{j=1}^n \cos \left( \frac{2 \pi}{n}(i-j) \right)^{2k} \\
        &=\sum_{j=1}^n \left[\frac{1}{2^{2k}} {2k \choose k} + \frac{2}{2^{2k}} \sum_{s=0}^k {2k \choose s} \cos \left( \frac{2 \pi}{n}(k-s) \right)\right] \\
        &= \frac{n}{2^{2k}} {2k \choose k},
    \end{align}
    from which it then follows that
    \begin{align}
        P_{CGU}(n,k) \geq \frac{1}{n} \sum_{i=1}^n \frac{2^{2k}}{n} {2k \choose k}^{-1} = \frac{2^{2k}}{n} {2k \choose k}^{-1} \, .
    \end{align}
    Finally, we define $h(k) \coloneqq 2^{2k} {2k \choose k}^{-1}$.
\end{proof}

\begin{figure}[!t]
    \centering
    \includegraphics[width=.8\textwidth]{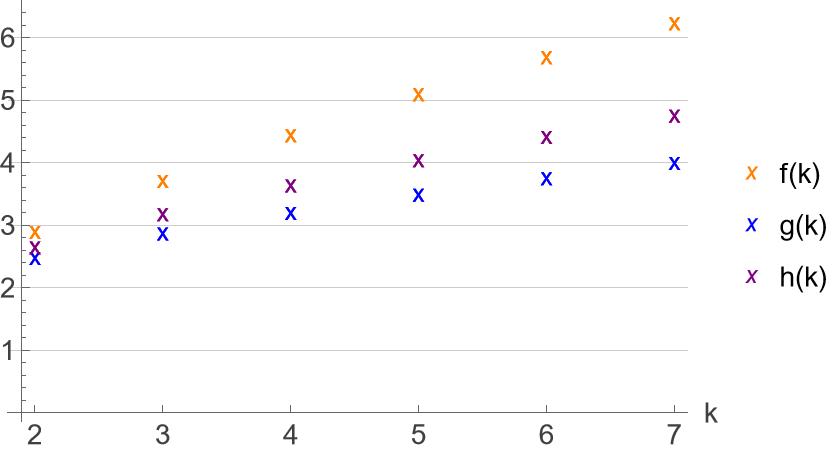}
    \caption{The gap between $f(k)$ and $g(k)$ increases with the number of copies. Since both functions are monotonically increasing, this suggests a quantum advantage. }
    \label{fig:f-vs-g}
\end{figure}

This lower bound will do the job: It enables us to show analytically a gap between discriminating states in classical bit theory and GU qubit states.
\begin{thm}\label{thm:quantum-advantage-by-GU}
    For all $n >k$ there exists $k$ copies from an ensemble of $n$ qubit states such that they have a higher average success probability in minimal error discrimination than $k$ copies from any ensemble of $n$ bit states. In particular, 
    \begin{align}
        P_q(n,k) \geq P_{GU}(n,k) \geq P_{CGU}(n,k) > \Pb(n,k)
    \end{align} for all $n>k$.
\end{thm}
\begin{proof}
 From \Cref{prop:classical-n-k-bound} we know that $P_{C}(n,k) \leq \frac{g(k)}{n}$ and \Cref{prop:GU-prop-to-1-n} and \Cref{prop:lower-GU-bound} together state that $P_{GU}(n,k)\geq P_{CGU}(n,k) = \frac{f(k)}{n} \geq \frac{h(k)}{n}$. For the function $g(k)$ related to $P_{C}(n,k)$ we can show a further upper bound:
 \begin{align}
     g(k) < 2+ 0.771\sqrt{\pi k} =: l(k)
 \end{align}
 for all $k \in \naturale$ (see \Cref{lem:super-class-bound} in the Appendix for the proof). Furthermore, in \Cref{lem:advantage-threshold} we show that $h(k)>l(k)$ for all $k\geq 25$. Thus, for $k\geq 25$ we get the following chain of inequalities:
 \begin{align}
     P_{GU}(n,k)\geq P_{CGU}(n,k) = \frac{f(k)}{n} \geq \frac{h(k)}{n} > \frac{l(k)}{n} > \frac{g(k)}{n} \geq \Pb(n,k).
 \end{align}
 For $k<25$ one can easily check by calculating explicitly that $g(k)<h(k)$ from which the statement follows. A plot of this can be found in \Cref{fig:h_vs_g}.
 \end{proof}
 
We thus have a quantum advantage even though we do not know the optimal quantum strategy. Only considering 3 CGU states, we retrieve the well known trine state ensemble \Cref{eq:trine-states}. Since $n=3$ is relatively small here, we are able to give the exact success probability for all $k$.

\begin{prop}\label{prop:GU-3-k-exact}
    The optimal success probability for discriminating $k$ copies of the trine states is given as:
    \begin{equation}
        P_{CGU}(3,k) =\frac{5+2^{1-k}+4 \sqrt{1+2^{-k}}\sqrt{2^{-k}(-2+2^k)}}{9}
    \end{equation}
\end{prop}
\begin{proof}
    This is done by straightforward computation.
    We calculate the Gram matrix for $k$-copies of trine states \Cref{eq:trine-states}:
    \begin{equation}
        G^{(k)} = \frac{1}{3}\begin{pmatrix}
            1 & -2^{-k} & -2^{-k} \\
            -2^{-k} &1 &-2^{-k} \\
            -2^{-k} & -2^{-k} & 1
        \end{pmatrix}
    \end{equation}
    We find the eigenvalues of $G^{(k)}$
    \begin{equation}
        \lambda_1=\frac{1}{3} 2^{-k} \left(2^k-2\right), \, \lambda_2=\frac{1}{3} 2^{-k} \left(2^k+1\right), \, \lambda_3 =  \frac{1}{3} 2^{-k} \left(2^k+1\right)\
    \end{equation}
    and corresponding eigenvectors
    \begin{equation}
        v_1 = (1,1,1)^\intercal, \, v_2 = (-1,0,1)^\intercal, \, v_3 = (-1,1,0)^\intercal \, .
    \end{equation}
    Let $\Sigma$ be a diagonal matrix containing the eigenvalues $\lambda_i$ and set $B = (v_1 \, v_2 \, v_3)$. Then use \Cref{eq:success-prob-for-GU-formula} to find:
    \begin{equation}
        P_{CGU}(3,k) = \sum_{i=1}^3 (\sqrt{G^{(k)}})_{ii}^2 = \sum_{i=1}^n (B \sqrt{\Sigma} B^*)_{ii}^2 = \frac{5+2^{1-k}+4 \sqrt{1+2^{-k}}\sqrt{2^{-k}(-2+2^k)}}{9}
    \end{equation}
\end{proof}

\begin{figure}[!t]
    \centering
    \includegraphics[width=0.9\linewidth]{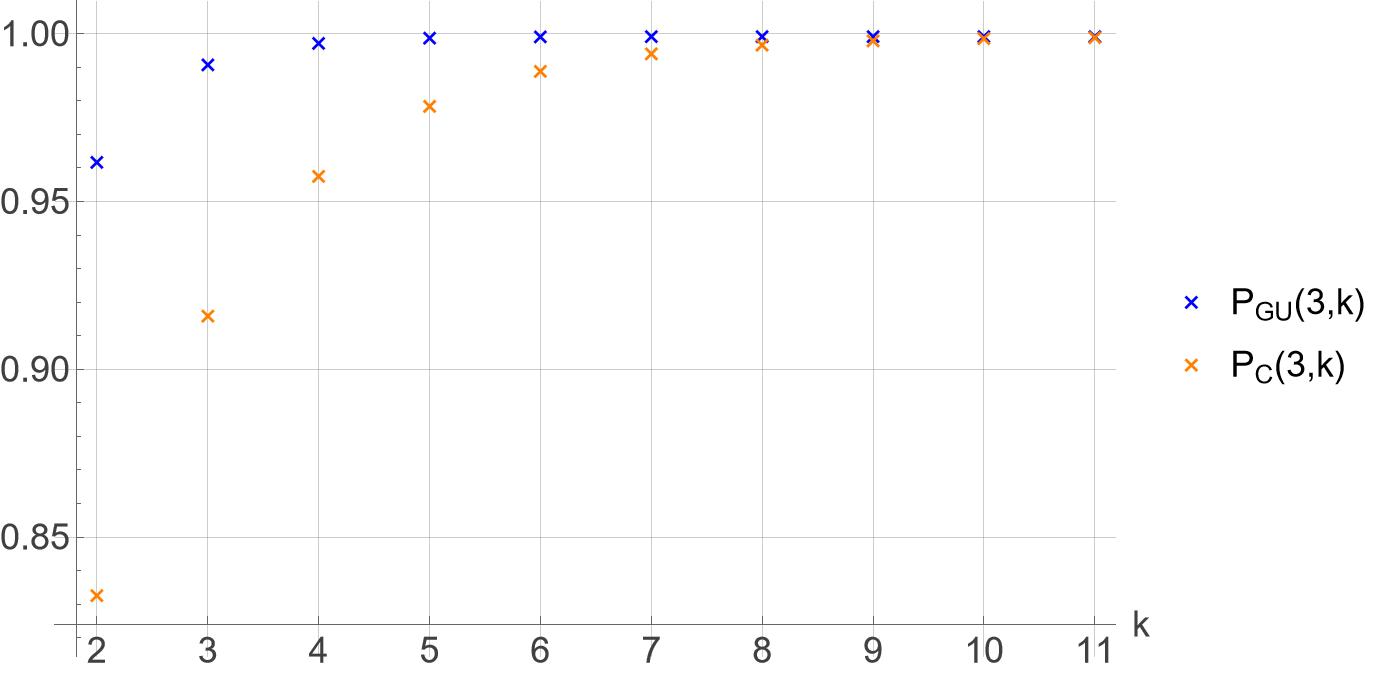}
    \caption{The success probability for multi-copy discrimination of trine states scales exponentially in both, quantum as well as classical, theories. However, it takes about $k=11$ copies for classical theory to be practically 1, whereas $k=5$ copies are sufficient for quantum theory to achieve effectively success probability 1.}
    \label{fig:Qtrine-vs-Ctrine}
\end{figure}
 We note that, as expected, the success probability asymptotically approaches 1. However, we see that it approaches 1 faster than the optimal success probability for bit as given in \Cref{prop:classical-3-k-exact}.
 \begin{cor}\label{cor:trine-advantage}
  For all $k>1$ there exists $k$ copies from an ensemble of $3$ qubit states such that they have a higher average success probability in minimal error discrimination than $k$ copies from any ensemble of $3$ bit states. In particular, 
     \begin{align}
         P_q(n,k) \geq  P_{GU}(3,k) \geq P_{CGU}(3,k) > P_{C}(3,k) 
     \end{align}
     for all $k>1$.
 \end{cor}
 \begin{proof}
     By using the expressions from \Cref{prop:classical-3-k-exact} and \Cref{prop:GU-3-k-exact} it can be straightforwardly seen that
     \begin{align}
          &  &  P_{C}(3,k)   & <  P_{CGU}(3,k)   \\
          &\Leftrightarrow & \quad 1-\frac{1}{3 \cdot 2^{k-1}} &< \frac{1}{9}\left( 5+2^{1-k}+4 \sqrt{1+2^{-k}}\sqrt{2^{-k}(-2+2^k)} \right)  \\
          &\Leftrightarrow & \quad  1-2^{-k+1} &< \sqrt{1+2^{-k}} \sqrt{1-2^{-k+1}} \\
          &\Leftrightarrow & \quad 1-2^{-k+1} &< 1+2^{-k+1} \\
          &\Leftrightarrow & \quad 0 &< 2^{-k+2}
     \end{align}
     which clearly holds for all $k>1$.
 \end{proof}
 The plot in \Cref{fig:Qtrine-vs-Ctrine} underlines this result.

\section{Comparing different measurement strategies for double trine}
\label{sec:measurement-settings}

We have seen that there are benefits of using qubits over bits in multi-copy state discrimination. 
Our next task is try to identify from what properties of qubit(s) do the benefits come from. 
In particular, we are interested whether the benefits arise strictly from the local properties of the single qubit system or from the global properties of the overall multiqubit system. 
In the multi-copy case we are restricted to product states so it is clear that entanglement of states is not the resource for the advantages. In order to do a proper comparison between qubit and bit and identify the source of the advantages, we examine some different measurement strategies. 

Let us briefly review different measurement operations that are investigated in this work. The general operation on the states is done by bringing them together and making a global operation on them. We denote $GLOBAL$ as a set of all such operations. An interesting measurement class is separable measurement, where each effect is represented by a separable operator. We denote a set of separable measurements as  $SEP$. On the other hand, \emph{local operations and classical communication (LOCC)} represents a general measurement strategy that is done locally; individual parties make measurements on their local subsystems via sequential measurements described by quantum instruments \cite{Davies:1970cf} and exchange classical information between measurement rounds. It is known that every $LOCC$ measurement is a separable measurement, but the converse does not hold \cite{Chitambar_2014}. 

To see where quantum benefit comes from, we put more restrictions on local 
operations, upto where the measurement strategy is identical to the one used in classical theory. For the discrimination of multipartite systems, one restricted operation of $LOCC$ is \textit{local adaptive strategy}, where each party makes exactly one local measurement based on the outcomes of the measurements of the previous parties. We denote such the set of such measurements as $AD$.
If in addition the communication is restricted to \textit{1 bit}, we write $AD_1$. \textit{Non-adaptive strategies} $NAD$ on the other hand, require every party to perform a fixed measurement on their system. Here, the measurement of each party can be chosen individually, i.e.\ they do not need to be identical across all parties. But once they are set, they are fixed. In classical theory, the measurement is represented as a single projective measurement followed by a postprocessing. This can be viewed as a local measurement strategy where each individual party makes identical measurements. We call this measurement strategy as \textit{local fixed measurement} $FIX$. In short, in general it holds that
\begin{align}
    FIX \subset NAD \subset AD \subset LOCC \subset SEP \subset GLOBAL \, . 
\end{align}

A particularly interesting ensemble is the double trine ensemble, $\{\ket{\psi_i}^{\otimes 2}\}_{i=1}^3$, where $\{\ket{\psi_i}\}_{i=1}^3$ forms a trine ensemble in \Cref{eq:trine-states}. 
We will now compare the different measurement strategies for the double trine ensemble. As the trine state (given by \Cref{eq:trine-states}) are GU-states, as we have seen in \Cref{cor:trine-advantage} they always give an advantage over any $n=3$ bit states for any $k>1$. Furthermore, there is strong numerical evidence suggesting that in fact the double trine states are the optimal states in qubit for $n=3$ and $k=2$. However, proving this analytically is still an open question.
\subsection{GLOBAL and SEP} Since the trine states are GU-states, the optimal global measurement for the double trine states is the PGM given by \Cref{eq:PGM-k} for $k=2$. As is seen from \Cref{table:values-f(k)} and \Cref{cor:trine-advantage}, we have that
\bea
P_{CGU}(3,2)=\frac{1}{2}+\frac{\sqrt{2}}{3} \approx 0.97
\eea
It was shown in \cite{eric_optimal_detec} that the optimal measurement can be chosen separable by mixing a singlet state and the PGM. Thus, for double trine we have that $P^{SEP}_{CGU}(3,2) = P^{GLOBAL}_{CGU}(3,2)$. Furthermore, in \cite{eric_optimal_detec} it was shown that the global optimal success probability cannot be attained by LOCC measurements. In the final stages of writing this manuscript we furthermore learned that a proof for optimality of the trine ensemble was reported in \cite{cha_2026_20586963}.

\subsection{Adaptive measurement strategy}
For $k=2$ the adaptive measurement strategy correspond to one-way LOCC: the first party, Alice, performs a measurement on her copy after which she sends her outcome to the second party, Bob, who can adapt his measurement based on Alice's outcome. Let us consider the following adaptive strategy for double trine: first Alice excludes one of the trine states using the antidiscriminating measurement
$M_i=\frac{2}{3}\ketbra{\psi_i^\perp}$,
where $\{\ket{\psi_i^\perp}\}_{i=1}^3$ are orthogonal to the trine states $\{\ket{\psi_i}\}_{i=1}^3$ defined in \Cref{eq:trine-states}. Next Alice sends her measurement result to Bob. This requires $\log 3$ bits of communication. Since Alice's measurement always excludes one of the states, 
Bob can make an optimal measurement for the remaining two states. Since Bob discriminates two states with prior probability $\frac{1}{2}$, the overall success probability is the Helstrom bound for two states and is given by $\frac{1}{2}+\frac{\sqrt{3}}{4}\approx 0.93$. It was shown in \cite{eric_optimal_detec} that this measurement strategy is actually optimal for one-way LOCC so that in fact
\bea
P^{AD}_{CGU}(3,2)=\frac{1}{2}+\frac{\sqrt{3}}{4}\approx 0.93 \, .
\eea
We note that this is still higher than the optimal success probability for bit $\Pb(3,2)=\frac{5}{6}$.

\subsection{Restricted adaptive measurement strategies}
The most restricted form of local adaptive strategy is when only one bit of classical communication is allowed. We show that even in this scenario, the double trine ensemble yields higher success probability than the classical protocol. 

The protocol starts with Alice making the $\sigma_x$ measurement. Let us denote $q_i(a)$ as the updated prior probability of $\varrho_i$ based on the measurement outcome of Alice, $a$. Alice's measurement outputs outcomes $a\in \{+,-\}$ from the measurement $\{\ketbra{+}, \ketbra{-}\}$. The updated priors are 
\bea
q_1(+)=\frac{1}{3}, q_2(+)=\frac{1}{6}(2+\sqrt{3}), q_3(+)=\frac{1}{6}(2-\sqrt{3}) 
\eea
and
\bea
q_1(-)=\frac{1}{3}, q_2(-)=\frac{1}{6}(2-\sqrt{3}), q_3(-)=\frac{1}{6}(2+\sqrt{3}) .
\eea
Bob then chooses the optimal measurement based on the updated priors. By using semidefinite programming (SDP), we find that the success probability $P$ using this strategy is now
\bea
P^{AD_1}_{CGU}(3,2) \geq P  \approx 0.8976
\eea
which is higher than the optimal success probability in classical theory, $\frac{5}{6}$. 
We remark that numerically it can be shown that the optimal success probability with one-bit of classical communication lies in $0.8976 \leq P^{AD_1}_{CGU}(3,2) \leq 0.9050$ \cite{Dutra_2026}. Furthermore, in \cite{Dutra_2026} it was also shown that for the non-adaptive strategy we have $0.8079 \leq P^{NAD}_{CGU}(3,2) \leq 0.8509$. Thus, for $NAD$ strategies it is no longer clear if we can beat the optimal bit success probability. 

Next we will take a deeper focus on the fixed measurement strategy in qubit as it is the most straightforward way to compare the power of qubit to bit.

\section{Local fixed measurements }
\label{sec:fixed-measurement-strategy}
In this section we want to examine if the quantum advantage persists if we restrict the setting to classical strategies.
So let us now focus on the fixed measurement strategy first for general $n,k \in \naturale$ in a $d$-dimensional quantum (or classical) theory. For any set of states $\{\varrho_i\}_{i=1}^n$ we define the fixed measurement strategy to consist of each of the $k$ parties applying the \emph{local optimal} discriminating  $n$-outcome measurement $M$ on their state after which they can jointly classically postprocess the measurement outcomes to make their guess. Thus, the local measurement $M$ satisfies
\begin{align}
     \sum_{i=1}^n \tr[\varrho_i M_i] = \max_{N} \sum_{i=1}^n \tr[\varrho_i N_i] \, .
\end{align}
In the fixed measurement strategy the overall discriminating measurement $M^{(k)}$ is of the form $M^{(k)}_{i_1, \ldots, i_k} = M_{i_1} \otimes \cdots \otimes M_{i_k}$ for all $(i_1, \ldots, i_k) \in [n]^k$.

Since the average success probability is a convex function of the measurement and since the set of stochastic postprocessings (with fixed size input and output) is a convex set, the success probability can be maximized by extreme postprocessings. Extreme postprocessings are simply the deterministic postprocessings and they can be characterized by some functions $f: [a] \to [b]$, where $a$ is the size of the input and $b$ is the size of the output. 
In the fixed measurement case with $k$ copies we are now postprocessing the measurement outcomes $(i_1, \ldots,i_k) \in [n]^k$ to some outcome $i \in [n]$. Thus, the optimal success probability reads as 
\begin{align}\label{eq:fix-measurement}
    P^{FIX}(d,n,k) &= \max_{\varrho, f} \frac{1}{n} \sum_{i=1}^n \tr\left[ \varrho_i^{\otimes k} \sum_{(i_1, \ldots, i_k) \in f^{-1}(i)} M_{i_1} \otimes \cdots \otimes M_{i_k}\right] \\
    &= \max_{\varrho, f} \ \frac{1}{n}\sum_{i_1=1}^n \cdots \sum_{i_k=1}^n \tr[\rho_{f(i_1, \ldots, i_k)}^{\otimes k} M_{i_1} \otimes \cdots \otimes  M_{i_k}],
\end{align}
where we are optimizing over all functions $f: [n]^k \to [n]$ and over all sets of $n$ $d$-dimensional quantum (or classical) states where in the optimization over the states we are also always choosing the optimal local measurement $M$ for the states. 

Our goal is to derive a general upper bound for $P^{FIX}(d,n,k)$  in terms of another optimization problem which we can explicitly solve when $d=2$, $n=3$ and $k=2$. However, first we need to simplify the optimization in \Cref{eq:fix-measurement}. Let us start with considering the optimal postprocessing. Indeed, we can show the following:

\begin{lem}\label{lem:fix-pp}
    The optimal postprocessing function $f: [n]^k \to [n]$ in \Cref{eq:fix-measurement} can be chosen such that
    \begin{itemize}
        \item[\textit{i)}] $f$ is permutation invariant: $f(\pi(i_1, \ldots, i_k)) = f(i_1, \ldots,i_k)$ for all permutations $\pi$, and
        \item[\textit{ii)}] $f(i,\ldots,i) = i$ for all $i \in [n]$. 
    \end{itemize}
\end{lem}
\begin{proof}
For all sets of states $\{\varrho_i\}_{i=1}^n$ and any $n$-outcome measurement $M$, let us define a (class of) postprocessing map(s) $f^*:[n]^k \to 2^{[n]}$ by 
\begin{align}\label{eq:optimal-pp}
f^*(i_1, \ldots, i_k)=\arg\max_x \tr[\rho_x^{\otimes k}   M_{i_1} \otimes \cdots \otimes  M_{i_k}]\, .
\end{align}
Strictly speaking for states $\{\varrho_i\}_{i=1}^n$ and measurement $M$ we have that $f^*(i_1, \ldots, i_k)$ returns a subset $X_{i_1, \ldots,i_k} \subseteq [n]$. Now clearly
\begin{align}
\tr[\rho_{f(i_1, \ldots, i_k)}^{\otimes k} M_{i_1} \otimes \cdots \otimes  M_{i_k}] \leq \tr[\rho_{x}^{\otimes k} M_{i_1} \otimes \cdots \otimes  M_{i_k}], 
\end{align}
for all $x \in X_{i_1, \ldots,i_k}$ for all $(i_1, \ldots, i_k) \in [n]^k$ and for any other postprocessing map $f$. Thus, for all $(i_1, \ldots, i_k) \in [n]^k$ we can fix any $x_{i_1, \ldots,i_k} \in X_{i_1, \ldots,i_k}$ and (with slight abuse of notation) set $f^*(i_1, \ldots, i_k) = x_{i_1, \ldots,i_k}$ so that then we get a class of maps by selecting different $x_{i_1, \ldots,i_k} \in X_{i_1, \ldots,i_k}$ for all $(i_1, \ldots, i_k) \in [n]^k$. Clearly by definition we see that all these maps are optimal and equally good for the fixed measurement strategy when using the states $\{\rho_i\}_{i=1}^n$ and measurement $M$.

\textit{i)} Let us denote the symmetric group of all permutations of $k$ elements by $\mathfrak{S}_k$. Since as sets $f^*(\pi(i_1, \ldots, i_k))=f^*(i_1, \ldots, i_k)$ for all $(i_1, \ldots, i_k) \in [n]^k$  and for all permutations $\pi \in \mathfrak{S}_k$ and since all the maps defined through $f^*$ are optimal, we can only consider those class of maps where the outcomes $\pi(i_1, \ldots, i_k)$ and $(i_1, \ldots, i_k)$ are in the same preimage $(f^*)^{-1}$ for all permutations $\pi \in \mathfrak{S}_k$. By the definition of $f^*$ the success probability for these maps will be still just as good.

\textit{ii)} Since we assume that the global measurement consists of each party measuring the fixed local optimal measurement $M$, the measurement $M$ satisfies the local optimality condition (when we have equal priors)
\begin{equation} \label{eq:local-optimality}
    \sum_{i=1}^n \tr[\varrho_i M_i] \geq \sum_{i=1}^n \tr[\varrho_i \tilde M_i] \, , 
\end{equation}
where $\tilde M = \{\tilde M_i\}_{i=1}^n$ is any measurement.
Now assume that there exist $x,\, y \in [n]$ such that $\tr[\varrho_x M_x] < \tr[\varrho_y M_x]$. That is to say, the state $\varrho_y$ maximizes the outcome $M_x$.
This allows us to look at a reduced discrimination problem, where we ignore the state $\varrho_x$. Mathematically we merge the two outcomes and write down a new measurement $\hat M$:
\begin{equation}
    \hat M_i = \begin{cases}
        M_i : i \notin \{x,y\} \\
        0 : i = x \\
        M_{xy} : i=y
    \end{cases}, \text{ where }  M_{xy} = M_x + M_y \, .
\end{equation}
It can easily be verified that $\hat M$ is a valid measurement. But now, according to the above, we must have the following.
\begin{align}
    \sum_{i=1}^n \tr[\varrho_i \hat M_i] &= \sum_{i \in [n]\setminus \{x,y\}} \tr[\varrho_i M_i] + \tr[\varrho_y M_x] + \tr[\varrho_y M_y] + \tr[\varrho_x \,  0] \\
    &> \sum_{i \in [n]\setminus \{x,y\}} \tr[\varrho_i M_i] + \tr[\varrho_x M_x] + \tr[\varrho_y M_y] = \sum_{i=1}^n \tr[\varrho_i M_i]
\end{align}
This clearly violates our assumption of local optimality \Cref{eq:local-optimality} and therefore $\tr[\varrho_j M_j] \geq \tr[\varrho_i M_j]$ for all $i,j \in [n]$ must hold.
Thus, it follows that for all $i \in [n]$ we have that
\begin{align}
    \max_x \tr[\varrho_x M_i]^k \leq \tr[\varrho_i M_i]^k
\end{align}
so that by \Cref{eq:optimal-pp} we can choose $f^*$ such that $f^*(i,\ldots,i) = i$ for all $i \in [n]$.
\end{proof}
In fact, it is not hard to see that indeed $f^*(i,\ldots,i) = i$ for all $i \in [n]$ if and only if \Cref{eq:local-optimality} is satisfied, where $f^*$ is the optimal postprocessing defined in \Cref{eq:optimal-pp}.
\begin{rem}
    Note that \Cref{lem:fix-pp} could also be stated for all general probabilistic theories (GPTs) by writing $M_i(s_j)$ instead of $\tr[\varrho_j M_i]$ (see \Cref{sec:GPTs} for more in depth discussion about GPTs). In quantum theory however, we could derive the condition $\tr[\varrho_i M_i] \geq \tr[\varrho_j M_i]$ for all $i,j \in [n]$ also via exploiting the known optimality conditions \cite{Barnett_2009}:
    \begin{gather}
        M_i(\varrho_i - \varrho_j)M_j = 0 \quad \forall i,j \in [n], \\
        \sum_i \varrho_i M_i - \varrho_j \geq 0 \quad \forall j \in [n]
    \end{gather}
\end{rem}

Let us now derive the desired general upper bound for $P^{FIX}(d,n,k)$. For any set of states $\{\varrho_i\}_{i=1}^n$ and any $n$-outcome measurement we denote by $C_{\varrho,M}$ the \emph{communication matrix} \cite{Heinosaari_2020} implemented by the prepare-and-measure set-up $(\varrho,M)$: the communication matrix is simply defined as the row-stochastic matrix consisting of the conditional outcome probabilities of the measurement $M$ when the system is in one of the states $\varrho_i$, i.e., $(C_{\varrho,M})_{ij} = \tr[\varrho_i M_j]$ for all $i,j \in [n]$. In terms of the communication matrix $C_{\varrho,M}$ we can write $P^{FIX}(n,k)$ as
\begin{align}\label{eq:fix-cm}
     P^{FIX}(d,n,k) &= \max_{\varrho} \max_{f \in \mathcal{F}} \frac{1}{n} \sum_{i=1}^n \sum_{(i_1, \ldots, i_k) \in f^{-1}(i)} (C_{\varrho,M})_{i i_1} \cdots (C_{\varrho,M})_{i i_k} = \max_{\varrho} \max_{f \in \mathcal{F}} P_f(C_{\varrho,M}),
\end{align}
where we have defined $P_f(C_{\varrho,M})= \frac{1}{n} \sum_{i=1}^n \sum_{(i_1, \ldots, i_k) \in f^{-1}(i)} (C_{\varrho,M})_{i i_1} \cdots (C_{\varrho,M})_{i i_k}$ and $\mathcal{F}$ as the set of postprocessing functions which satisfy the conditions of \Cref{lem:fix-pp} and again, when optimizing over the states we are always considering the local optimal discriminating measurement $M$ for the states.

A general upper bound can be achieved by relaxing the conditions of the state (and measuremeent) optimization by not optimizing over the quantum (or classical) implementation $(\varrho,M)$ of the communication matrix but rather optimizing over general row-stochastic matrices themselves with an additional constraint that is satisfied both in $d$-dimensional quantum and classical theories. 
This constraint can be stated in terms of the \emph{information storability} of a communication matrix: it is defined as a function $\mathfrak{IS}$ from the set of row-stochastic matrices (of fixed size) to the set of real numbers $\real$ as $\mathfrak{IS}(C) = \sum_j \max_i C_{ij}$. 
Let $C$ be a communication matrix $C$ that can be implemented with a $d$-dimensional quantum (or classical) theory via some prepare-and-measure set-up $(\varrho,M)$. From the basic decoding theorem \cite{Schumacher_Westmoreland_2010} it follows that $\mathfrak{IS}(C) \leq d$. We can use this as an additional constraint in the optimization and obtain the following upper bound:
\begin{equation}
\begin{aligned}\label{eq:fix-cm-ub}
     P^{FIX}(d,n,k) &= \max_{\varrho} \max_{f \in \mathcal{F}} \frac{1}{n} \sum_{i=1}^n \sum_{(i_1, \ldots, i_k) \in f^{-1}(i)}(C_{\varrho,M})_{i i_1} \cdots (C_{\varrho,M})_{i i_k} \\
     &\leq  \max_{
         C: \ \mathfrak{IS}(C) \leq d }  \max_{f \in \mathcal{F}} \frac{1}{n} \sum_{i=1}^n \sum_{(i_1, \ldots, i_k) \in f^{-1}(i)} C_{i i_1} \cdots C_{i i_k} =: P^{FIX}_{d-\mathrm{like}}(n,k)
\end{aligned}
\end{equation}

In the case when $d=2$, $n=3$ and $k=2$ we can actually list all the inequivalent postprocessings and calculate the upper bound. This allows us to show the following.
\begin{prop} \label{prop:local-fixed-bound}
Perfect discrimination of two copies from a 3-state ensemble is impossible in bit-like theories. In particular, it holds that
\begin{equation}
   P^{FIX}_{2-\mathrm{like}}(3,2) = \frac{8}{9} \, .
\end{equation} 
\end{prop}
\begin{proof}
In this case we can simply list all the functions $f: [3]^2 \to [3]$ which obey the conditions of \Cref{lem:fix-pp}. A postprocessing map $f$ of this type maps the outcomes as follows:
\begin{align}
(1,1) &\mapsto 1, \quad (2,2) \mapsto 2, \quad (3,3) \mapsto 3\\
(1,2) &\mapsto x_{12}, \quad (2,1) \mapsto x_{12}, \\
(2,3) &\mapsto x_{23},  \quad (3,2) \mapsto x_{23}, \\
(3,1) &\mapsto x_{31}, \quad (1,3) \mapsto x_{31}
\end{align}
for some $x_{12},x_{23},x_{31} \in \{1,2,3\}$. In total there are 27 of such maps. However, one last symmetry that we have not considered yet is related to the fact that we can freely bijectively relabel the states and it does not affect the average success probability. For example, the three cases when $x_{12}=x_{23}=x_{31}$ are all equivalent under relabeling. Similarly, all the three cases when $x_{12}=x_{23} \neq x_{31} $ are equivalent as well as the cases when $x_{12}\neq x_{23} = x_{31} $ and $x_{12}=x_{31} \neq x_{23} $.
We are hence left with only 4 inequivalent postprocessings.

The only cases we miss are the ones when $x_{12} \neq x_{23} \neq x_{31} \neq x_{12}$. There are in total 6 of these cases: $(x_{12},x_{23},x_{31}) \in \{(1,2,3),(1,3,2),(2,1,3),(2,3,1),(3,1,2),(3,2,1)\}$. One can see that the case $(x_{12},x_{23},x_{31}) = (2,1,3)$ is the same as $(x_{12},x_{23},x_{31}) = (1,3,2)$ by relabeling the states by $i \mapsto i+1 \ (\textrm{mod} \ 3)$. Similarly the case $(x_{12},x_{23},x_{31}) = (2,3,1)$ is the same as $(x_{12},x_{23},x_{31}) = (1,2,3)$ by swapping the labels $1 \leftrightarrow 2$ and the case $(x_{12},x_{23},x_{31}) = (3,2,1)$ is the same as $(x_{12},x_{23},x_{31}) = (3,1,2)$ by relabeling the states by $i \mapsto i+1 \ (\textrm{mod} \ 3)$. However, straightforwardly by going through all the relabelings one can confirm that the three cases when  $(x_{12},x_{23},x_{31}) \in \{(1,2,3),(1,3,2),(3,1,2)\}$ are inequivalent. Thus, in total we now have the following representations of the 7 inequivalent postprocessings. In the following we simply denote a pair of outcomes $(i,j)$ as $ij$ and we define the maps in terms of their preimages $f^{-1}(i)$:
\bea
f_1^{-1}(1)=\{ 11\}, \quad  f_1^{-1}(2)=\{ 22\}, \quad f_1^{-1}(3)=\{ 33,12,21,23,32,31,13\}\\
f_2^{-1}(1)=\{ 11\}, \quad  f_2^{-1}(2)=\{ 22, 12,21\}, \quad  f_2^{-1}(3)=\{ 33,23,32,31,13\}\\
f_3^{-1}(1)=\{ 11\}, \quad  f_3^{-1}(2)=\{ 22, 23,32\}, \quad  f_3^{-1}(3)=\{ 33,12,21,31,13\}\\
f_4^{-1}(1)=\{ 11\}, \quad  f_4^{-1}(2)=\{ 22, 31,13\}, \quad f_4^{-1}(3)=\{33,12,21, 23, 32\}\\
f_5^{-1}(1)=\{ 11, 23, 32\}, \quad  f_5^{-1}(2)=\{   22, 31, 13\}, \quad  f_5^{-1}(3)=\{ 33, 12, 21\}\\
f_6^{-1}(1)=\{ 11, 12, 21\}, \quad  f_6^{-1}(2)=\{   22, 31, 13\}, \quad  f_6^{-1}(3)=\{ 33, 23, 32\} \\
f_7^{-1}(1)=\{ 11, 12, 21\}, \quad  f_7^{-1}(2)=\{22, 23, 32\}, \quad  f_7^{-1}(3)=\{ 33, 31, 13\}
\eea
Let us denote 
\begin{align}\label{eq:P_j}
    P_j(C):=\frac{1}{3} \sum_{i=1}^3 \sum_{(a,b) \in f_j^{-1}(i)} C_{i a}C_{i b}
\end{align}
for all $j \in [7]$ and for all row-stochasic matrices $C$. Now we can separately perform the optimization $\max_{C: \ \mathfrak{IS}(C) \leq 2 } P_j(C)$ for all $j \in [7]$. We note that in the optimization we can restrict to matrices $C$ with $C_{ii} \geq C_{ij}$ for all $i,j$. This is because of the assumption about the local optimality of the measurements in the original optimization \Cref{eq:fix-cm} so that considering row-stochastic matrices with $C_{ii} \geq C_{ij}$ are enough for the upper bound in \Cref{eq:fix-cm-ub}. In this case clearly $\mathfrak{IS}(C) = \sum_i C_{ii}$.

Now we can perform the optimization $\max_C P_j(C)$ for all $j \in [7]$ using the constraints $0 \leq C_{ab} \leq 1$, $\sum_{i=1}^3C_{ai}=1$ and $\sum_{i=1}^3 C_{ii} \leq 2$  for all $a,b \in [3]$. The mathematical details of the analytic optimizations can be found in \Cref{appendix:post-process-local-fixed}. The results are the following:
\begin{table}[h]
\centering
\bgroup
\def\arraystretch{1.5}
    \begin{tabular}{|c||c|c|c|c|c|c|c|} \hline
    $j$ & $1$ & $2$ & $3$ & $4$ & $5$ & $6$ & $7$ \\ \hline 
    $\max_C P_j(C)$ & $\frac{5}{6}$  & $\frac{5}{6}$ & $\frac{5}{6}$ & $\frac{5}{6}$ & $\frac{5}{6}$ & $\frac{5}{6}$ & $\frac{8}{9}$ \\ \hline
\end{tabular}
\egroup
    \caption{Results of the maximization of \Cref{eq:P_j} over all row-stochastic matrices $C$ with $\mathfrak{IS}(C) \leq 2$.}
    \label{tab:fix-results}
\end{table}

This proves the claim.
\end{proof}
In fact we observed that the postprocessing function $f_7$ obtains the maximum value with unique communication matrix. Up to trivial relabeling of rows and colums the unique communication matrix reads
    \bea
    C=\begin{pmatrix}
    \frac{2}{3} & \frac{1}{3} & 0\\
    0 & \frac{2}{3} & \frac{1}{3}\\
    \frac{1}{3} & 0 & \frac{2}{3}
    \end{pmatrix} \, .
    \eea
In \Cref{sec:GPTs} we will give an example of a bit-like GPT where this communication can be implemented, thereby showing that the bit can be beat even with a classical measurement strategy. What is left is to see whether or not qubit can beat bit in this or not.
It turns out that for $n=3$ states and $k=2$ copies there is no advantage in using qubit over bit with the fixed measurement strategy. In particular, the following statement can be proven.
\begin{thm}
     Qubit and bit have equal optimal success probability using the fixed measurement strategy for $k=2$ copies of $n=3$ messages, i.e.,
    \begin{align}
        P^{FIX}_q(3,2) = P^{FIX}_b(3,2) = \frac{5}{6} \, .
    \end{align}
\end{thm}
\begin{proof}
    First of all we note that quantum theory can implement the strategy constructed in \Cref{cor:class-optimum-n3-k2} and hence we have $P_q^{FIX}(3,2) := P^{FIX}(2,3,2) \geq \frac{5}{6}$.
    
    In the proof of \Cref{prop:local-fixed-bound} we saw that the only postprocessing that can beat the bit success probability $\frac{5}{6}$ is the postprocessing $f_7$. Now the success probability reads as  
    \bea
     P_7(C)=\frac{1}{3}(C_{11}^2+ 2 C_{11} C_{12} + C_{22}^2 + 2 C_{22} C_{23} +  C_{33}^2 + 2 C_{33} C_{31}),
    \eea
    where $C = C_{\varrho,M}$ is qubit communication matrix implemented by using some qubit states $\varrho = \{\varrho_i\}_{i=1}^3$ and some trichotomic qubit measurement $M$. We see that the success probability constitutes of three blocks of similar terms $C_{ii}^2+2 C_{ii} C_{i(i+1)}$ for all $i \in [3]$ with constraint $C_{11}+C_{22}+C_{33}\leq 2$ across the blocks. 
    Using $C_{i(i+1)}=1-C_{ii}-C_{i(i+2)}$, we can rewrite each block as $C_{ii}^2+2 C_{ii} C_{i(i+1)}=2 C_{ii} - C_{ii}^2-2 C_{ii} C_{i(i+2)}$. Note that this is monotonically increasing for $C_{ii} \leq 1- C_{i(i+2)}$, a condition that holds with any communication matrix. Therefore, the maximum success probability is attained when $C_{11}+C_{22}+C_{33}=2$. 

    In quantum theory, $C_{11}+C_{22}+C_{33}=2$ holds if and only if the implementing measurement consists of rank-1 effects $M_i$ and the states are exactly the corresponding (pure) eigenstates. Thus, for all $i \in [3]$ we have that $\varrho_i = \ketbra{\varphi_i}$ for some unit vector $\varphi_i \in \complex^2$ and $M_i=a_i \varrho_i$, where $\sum_{i=1}^3 a_i =2$ since $\sum_{i=1}^3 M_i = \one_2$. Then, $C_{ij} = a_j |\ip{\varphi_i}{\varphi_j}|^2$ for all $i,j \in [3]$ and then one finds for such $C$ that
    \begin{align}
   P_7(C)&=\frac{1}{3}(C_{11}^2+ 2 C_{11} C_{12} + C_{22}^2 + 2 C_{22} C_{23} +  C_{33}^2 + 2 C_{33} C_{31}) \\
   &= \frac{1}{3}(a_1^2+a_2^2+a_3^2+2(a_1a_2 |\ip{\varphi_1}{\varphi_2}|^2 + a_2a_3 |\ip{\varphi_2}{\varphi_3}|^2 +a_3a_1 |\ip{\varphi_3}{\varphi_1}|^2)) \\
   &=\frac{1}{3}(\tr[M_1^2]+\tr[M_2^2]+\tr[M_3^2 ]+2 (\tr[M_1 M_2]+\tr[M_2 M_3]+\tr[M_3 M_1]))\nonumber \\
    &=\frac{1}{3}(\tr[(M_1+M_2+M_3)^2])=\frac{2}{3}
    \end{align}
    Therefore, using the processing map $f_7$ yields the success probability $\frac{2}{3}$. This proves the claim since in the proof of \Cref{prop:local-fixed-bound} (see \Cref{appendix:post-process-local-fixed}) we saw that all other postprocessing maps cannot give a success probability larger than $\frac{5}{6}$.
\end{proof}

This motivates us to look more deeply into the different measurement strategies in other operational theories whether these observations about specific strategies are just specific to qubit or can we find other (qu)bit-like operational theories where each different measurement strategy can offer advantages over bit or even qubit. This helps us understand which type of local properties of a theory contribute to optimality of different measurement strategies, and allows us to clarify in various theories where the nonclassical traits in the multi-copy state discrimination originates from.

\section{Generalized probabilistic theories}\label{sec:GPTs}

General Probabilistic Theories (GPTs) comprise a framework for operational theories developed as an attempt to include possible physical theories in one mathematical framework. In fact, the idea to encapsulate classical and quantum theory in one model is almost as old as quantum theory itself \cite{neumann36}.
Traditionally, GPTs have been used to derive quantum theory from few well-motivated axioms \cite{hardy2001quantumtheoryreasonableaxioms,Masanes_2011,Chiribella_2011}. They were further developed to answer foundational questions as 'When is a theory physical?' or 'What are genuine features of quantum theory?'.
It turned out that many phenomena of nature, which were to believed quantum, were in fact rather nonclassical.
For example, there is a whole class of nonlocal theories, and some of those, known as Popescu-Rohrlich (PR) boxes \cite{rohrlich1995nonlocalityaxiomquantumtheory}, are even more nonlocal than quantum theory.
In modern days, researchers formulate information theoretic processes in the GPT framework. This allows to compare the performance in a certain communication task based on the theory and in this way, characterize theories. In one instance, the derivation of quantum theory through state discrimination in GPTs was attempted \cite{arai2024derivationstandardquantumtheory}.

Our motivation for considering GPTs comes from the same direction. So far, we have compared quantum theory with classical theory in multi-copy state discrimination and found gaps between them. We want to examine further, the source for those gaps as well as the conditions on a theory to perform on a certain level in multi-copy discrimination.
Even though we are considering toy models, which are obviously nonphysical, we are able to get insights into the differences as well as similarities of classical and quantum theories.

Concisely, as operational theories, GPTs are built on the concepts of states (preparations), effects ('yes-no' questions), measurements, state transformations and probabilistic outcomes when measuring states. 

\subsection{Framework of GPTs}

Here we give a brief introduction of the mathematical framework of GPTs. For a comprehensive review on GPTs we refer to \cite{Pl_vala_2023}.

The set of states of a GPT is described by a compact convex subset $K \subset V$ of a real, finite dimensional vector space $V$. The convexity of the state space stems from the fact that operationally we can make probabilistic mixtures of preparations of states. The extreme points of $K$ are the pure states and all other states are mixed. For mathematical convenience we choose $V$ such that $\dim(V) = \dim(\mathrm{aff}(K))+1$, where $\mathrm{aff}(\cdot)$ denotes the affine hull. Furthermore we embed $K$ into $V$ such that $0 \notin K$. In this case we can define a proper cone\footnote{A subset $C \subset V$ of a vector space is called a \emph{cone} if for all $v \in C$ we have that $\lambda v \in C$ for all $\lambda \in \real_+$. A cone $C$ is \emph{pointed} if $C \cap -C = \{0\}$ and \emph{generating} if $\mathrm{span}(C) = V$. A convex, closed, pointed, and generating cone is called a \emph{proper cone}.} $V^+  \subset V$ by setting $V^+ := \{\lambda x \, : \, x \in K, \ \lambda \in \real_+\} $. A proper cone $V^+$ induces an order relation on the vector space and gives a notion of positivity: $v \geq 0$ if and only if $v \in V^+$. Let $V^*$ be the dual space of $V$, containing linear functionals on $V$. The dual cone $(V^*)^+ = (V^+)^* \subset V^*$ is the set $(V^*)^+ = \{f \in V^* : f(s) \geq 0 \mspace{10mu} \forall s \in V^+ \}$.

\begin{eg} \label{example:psd-cone}
    Let $PSD_d$ be the set of $d$-dimensional, positive semidefinite matrices. Since a positive semidefinite matrix is invariant under scaling by a positive scalar, $PSD_d \subset \mathcal{M}(\complex)_d^{sa}$ forms a cone in the vector space of selfadjoint matrices $\mathcal{M}(\complex)_d^{sa}$.
    Since $PSD_d$ is self-dual, we get the dual cone as $PSD_d \equiv PSD_d^*$.
\end{eg}

 In particular we have the constant function  $\one_K \in (V^*)^+$ on $K$, defined as $\one_K(s) = 1 \, \forall s \in K$. Mathematically $\one_K$ is then an \emph{order unit} of $V^*$. Now we may represent our state space as the set of nonnegative normalized vectors:
\begin{align}\label{eq:states}
    K = \{x \in V \, : \, x \geq 0, \ \one_K(x) =1\} \, .
\end{align}

An effect is the simplest type of measurement and it corresponds to a 'yes-no' question about the systems properties. We define the set of effects to be the set of linear functionals giving probabilities on states, i.e., 
\begin{align}\label{eq:effects}
     \mathcal{E}(K) &= \{f \in V^*:0 \leq f(s) \leq 1 \mspace{10mu} \forall s \in K \} \\
     &= \{f \in V^*:0 \leq f \leq \one_K \mspace{10mu}\},
\end{align}
where the partial order on $V^*$ is now induced by the dual cone $(V^+)^*$. We call $ \mathcal{E}(K)$ as the effect space of $K$. Now for an effect $e \in \mathcal{E}(K)$ we interpret $e(s)$ as the probability that answer 'yes' is obtained in the question represented by $e$ when the system is in state $s \in K$. Since $K$ is convex and the functionals $f \in V^*$ are linear, $\mathcal{E}(K)$ is convex, too. More general measurements can be now definied as collection of effects:  $M = \{ M_i \}_{i=1}^n \subset \mathcal{E}(K)$ is called a measurement with $n$ outcomes if $\sum_{i = 1}^n M_i = \one_K$. We call $M_i$ the $i$'th outcome of measurement $M$, which is obtained with probability $M_i(s)$ for some state $s \in K$. The normalization of the measurement guarantees that some outcome is always detected when a state is measured.

Thus, starting from a convex set we were able to define the state space, the effect space and measurements of a GPT. On the other hand, starting from a vector space $V$, any proper cone $V^+$ and an order unit $\one \in V^*$ one is able to derive the convex set $K$ corresponding to the state space and the effect space $\mathcal{E}(K)$ as in \Cref{eq:states} and \Cref{eq:effects}. Hence, this is all we need to define a GPT.
\begin{defn}
    A GPT is a triple $(V, V^+, \one_K)$, where $V$ is a real finite-dimensional vector space with proper cone $V^+$ and $\one$ is an order unit of $V^*$.
\end{defn}

We illustrate above definitions with the two most well-known examples.
\begin{eg}
    From \Cref{example:psd-cone} we take the sets $PSD_d$ and $\mathcal{M}(\complex)_d^{sa}$ to define quantum theory as the GPT $(\mathcal{M}(\complex)_d^{sa},PSD_d, \tr)$, where $\tr$ denotes the trace function $\tr[\cdot]$, providing the order unit. The state space is then the set of density matrices (positive semidefinite matrices with unit trace) denoted by $\mathcal{D}_d$. The effect space is then isomorpic to the set of positive semidefinite matrices which are below the identity matrix, i.e., $\mathcal{E}(\mathcal{D}_d) \cong \{ E \in \mathcal{M}(\complex)_d^{sa} \, : \, 0 \leq E \leq \one_d\}$, where $\one_d$ is now the identity matrix and corresponds to the order unit represented by the trace function.
\end{eg}
\begin{eg}
    Classical theory is described by a probability simplex, that is, a set of $d$-dimensional vectors containing a probability distribution $s = (p_1,\ldots,p_d)^\intercal$, such that $\sum_{i=1}^d p_i = 1$ and $p_i \geq 0$ for all $i \in [d]$. Classical theory is the triple $(\real^d,\real_+^d,1_d)$, where $1_d = (1,\ldots,1)$. If we would embed classical theory in quantum theory, then the states would correspond to diagonal density matrices and effects would correspond to diagonal effect operators (as in \Cref{sec:classical}). Classical theory has a unique measurement from which all other measurements can be postprocessed and it is the measurement which perfectly distinguishes all $d$ pure states.
\end{eg}
Here we omit deeper discussions about composite systems and introduce as much further details as needed for our investigations. 

In this work we will focus on classes of theories that we can straightforwardly compare to qubit and bit. One important property that those two share is that the \emph{operational dimension} $d_{op}$ of the theory, defined as the maximal number of perfectly distinguishable pure states, is 2 both in qubit and bit. Thus, we will be focusing on other theories with operational dimension $d_{op} =2$. Another related quantity is the \emph{information storability} of the theory \cite{Kimura_2018,Heinosaari_2024}, defined as the maximum amount of information that can be encoded and decoded by using the theory as
\begin{align}
    \mathfrak{IS}  = \max_{S,M} \sum_i M_i(s_i) \, .
\end{align}
In $d$-dimensional quantum and classical theory we have that $\mathfrak{IS} = d_{op} = d$ but in general we only have $\mathfrak{IS} \geq d_{op}$. Interestingly. there are theories where the inequality is strict, and in that case we say that the theory has \emph{super-information storability}. However, for a fair comparison to qubit and bit we will be mostly focusing on theories which are (qu)bit-like also in this sense, i.e., theories with $\mathfrak{IS} =2$. In \cite{Kimura_2018} it was actually shown that these theories are exactly theories which have a point-symmetric state space.

\subsection{Perfect fixed and adaptive measurement strategies in GPTs}

Earlier we saw that perfect discrimination is not possible in quantum and classical theory in general. The question arises, if this is due to the properties of the state spaces and what are the conditions to obtain success probability 1 in a local strategy for a finite number of copies.
We start with the simplest case of fixed strategies, where the choice of measurement is never changed. Contrary, we will later consider adaptive strategies, i.e.\ the measurement choice depends on the outcomes obtained of the prior measurements.

For any set of states $S=\{s_i\}_{i=1}^n$ we define the fixed measurement strategy to consist of each of the $k$ parties applying some fixed $m$-outcome measurement $M$ on their state after which they can jointly classically postprocess the measurement outcomes to make their guess. Now the overall discriminating measurement $M^{(k)}$ is of the form $M^{(k)}_{i_1, \ldots, i_k} = M_{i_1} \otimes \cdots \otimes M_{i_k}$ for all $(i_1, \ldots, i_k) \in [m]^k$. (Note that unlike in \Cref{sec:fixed-measurement-strategy} here we do not assume that the local fixed measurements are the local optimal measurements for the given set of states but they can be any local measurements). The optimal success probability in the fixed measurement strategy now reads as
    \begin{align}\label{eq:fix-def}
        P^{FIX}(n,k) = \max_{S,M,f} \frac{1}{n} \sum_{i=1}^n \sum_{(i_1, \ldots, i_k) \in f^{-1}(i)} M_{i_1}(s_i) \cdots M_{i_k}(s_i) ,
    \end{align}
where we optimize over all sets of $n$ states $S$, all $m$-outcome measurements $M$ and all postprocessing functions $f:[m]^k \to [n]$ for all $m \in \naturale$.

\begin{prop}\label{prop:FIX-prob-1}
    In any operational theory, we have $P^{FIX}(n,k) =1$ if and only if the operational dimension $d_{op}$ of the theory is at least $n$.
\end{prop}
\begin{proof}
    If we denote the $m \times n$ communication matrix implemented with states $S=\{s_i\}_{i=1}^n$ and an $m$-outcome measurement $M$ by $C_{S,M}$ so that $(C_{S,M})_{ij} = M_j(s_i)$ for all $i\in [n],j \in [m]$, we can simply rewrite the success probability as 
        \begin{align}
        P^{FIX}(n,k) &= \max_{S,M,f} \frac{1}{n} \sum_{i=1}^n \sum_{(i_1, \ldots, i_k) \in f^{-1}(i)}(C_{S,M})_{ii_1} \cdots (C_{S,M})_{ii_k} \\
        &=  \max_{S,M,f} p^{(n,k)}_f(C_{S,M}),
    \end{align}
    where we have denoted $p^{(n,k)}_f(C) = \frac{1}{n} \sum_{i=1}^n \sum_{(i_1, \ldots, i_k) \in f^{-1}(i)} C_{ii_1} \cdots C_{ii_k}$ for any postprocessing function $f: [m]^k \to [n]$ and any $n \times m$ row-stochastic matrix $C$.

    If the operational dimension of the theory is at least $n$, then there exists a set of states $S=\{s_i\}_{i=1}^n$ and an $n$-outcome measurement $M$ such that for all $i \in [n]$ we have $M_{i}(s_i) =1$. Now we can define a postprocessing $f: [n]^k \to [n]$ by setting $f(i, \ldots, i) = i$ and the rest of the outcome strings can be assigned arbitrarily. It is straightforward to confirm that with these choices we have that $p^{(n,k)}_f(C_{S,M})= 1$.
    
    Suppose now that $P^{FIX}(n,k)=1$. Thus there are some postprocessing function $f: [m]^k \to [n]$ and some $m \times n$ communication matrix $C$ which can be implemented with some states $S = \{s_i\}_{i=1}^n$ and some $m$-outcome measurement $M$ as $C = C_{S,M}$ such that $p^{(n,k)}_f(C)=1$. By convexity, this implies that 
    \begin{align}
        \sum_{(i_1, \ldots, i_k) \in f^{-1}(i)} C_{ii_1} \cdots C_{ii_k} =1, \quad \forall i \in [n] \,.
    \end{align}
    However, by the row-stochasticity of $C$ we have that
    \begin{align}\label{eq:fixed-n-k-1}
        1 = \sum_{(i_1, \ldots, i_k) \in f^{-1}(i)} C_{ii_1} \cdots C_{ii_k} \leq \sum_{i_1, \ldots, i_k =1}^m C_{ii_1} \cdots C_{ii_k} =1, \quad \forall i \in [n] \,.
    \end{align}
    In particular, now the above inequality must be an equality. This shows the following two things for all $i\in[n]$: \textit{i)} if $C_{ii_1} \cdots C_{ii_k} >0$, then $(i_1, \ldots, i_k) \in f^{-1}(i)$ and there exists at least one such string of outcomes $(i_1, \ldots,i_k)$ (so that in particular $f^{-1}(i) \neq \emptyset$), and \textit{ii)} if $(i_1, \ldots, i_k) \notin f^{-1}(i)$, then $C_{ii_1} \cdots C_{ii_k} =0$. Furthermore, we also see that if $(i_1, \ldots,i_k) \in f^{-1}(i)$ such that  $C_{ii_1} \cdots C_{ii_k} >0$, then clearly $C_{ii_l} \cdots C_{ii_l} >0$ for all $l \in [k]$ so that by \textit{i)} we have that $(i_l ,\ldots, i_l) \in f^{-1}(i)$ for all $l \in [k]$.
    
    Let us denote $J^{(i)}_{>0} = \{ j \in [m] \, : \,  C_{ij} >0\}$. Now from the properties listed above we see that $J^{(i)}_{>0} \neq \emptyset$ for all $i \in [n]$ and $J^{(i)}_{>0} \cap J^{(i')}_{>0} = \emptyset$ for all $i \neq i'$. The latter fact simply means that an outcome string can only be in the preimage of exactly on $i$. On the other hand, without loss of generality we can assume that all $m$ outcomes of the measurement $M$ are nonzero on the states $S=\{s_i\}_{i=1}^n$ (otherwise we could just merge those outcomes that are zero with some other outcome and the success probability would still remain 1 on the states $S$). This means that for all $j \in [m]$ there exists $i \in [n]$ such that $C_{ij} >0$, i.e., that $j \in J^{(i)}_{>0}$. Thus the sets $\{ J^{(i)}_{>0}\}_{i=1}^n$ form a partition of $[m]$.

    We can now define an $n$-outcome measurement $M'$ by setting $M'_i = \sum_{j \in J^{(i)}_{>0}} M_j$ for all $i \in [n]$. It follows that
    \begin{align}
        M'_i(s_i) = \sum_{j \in J^{(i)}_{>0}} M_j(s_i) = \sum_{j \in J^{(i)}_{>0}} C_{ij} = \sum_{j=1}^m C_{ij} =1
    \end{align}
     for all $i \in [n]$. Thus, $M'$ can perfectly discriminate the states $S=\{s_i\}_{i=1}^n$ so that by definition the operational dimension of the theory is at least $n$. 

\end{proof}
Hence, by the previous result we see that (qu)bit-like theories with $d_{op}=2$ can never achieve perfect discrimination for $n\geq 3$ states for any number of copies. In fact, our previous results from \Cref{sec:fixed-measurement-strategy} show that in this case qubit performs no better than bit. On the other hand we will later see that this is not a general property of (qu)bit-like theories but we find an instance of such theory which can beat bit even with the fixed measurement strategy.

Let us consider the adaptive measurement strategy next. For any set of states $S=\{s_i\}_{i=1}^n$ we define the adaptive measurement strategy to consist of each party $l \in [k]$ applying on their copy of the state some $m$-outcome measurement $M^{(l)}_{\cdot |i_1, \ldots,i_{l-1}}$ conditioned on the all previous $l-1$ parties outcomes $i_1, \ldots, i_{l-1}$ after which they can apply some joint postprocessing. Now the overall discriminating measurement $M$ is of the form $M_{i_1, \ldots, i_k} = M^{(1)}_{i_1} \otimes M^{(2)}_{i_2|i_1} \cdots \otimes M^{(k)}_{i_k|i_1, \ldots,i_{k-1}}$ for all $(i_1, \ldots, i_k) \in [m]^k$. The optimal success probability in the adaptive measurement strategy reads as
    \begin{align}\label{eq:adaptive}
        P^{AD}(n,k) = \max_{M^{(1)}, \ldots, M^{(k)}, S, f} \frac{1}{n} \sum_{j=1}^n \sum_{(i_1, \ldots,i_k) \in f^{-1}(j)} M^{(1)}_{i_1}(s_j) M^{(2)}_{i_2|i_1}(s_j) \cdots M^{(k)}_{i_k|i_1,\ldots, i_{k-1}}(s_j),
    \end{align}
where we optimize over all sets of $n$ states $S$, all adaptive $m$-outcome measurements $M^{(1)}, \ldots, M^{(k)}$ and all postprocessing functions $f: [m]^k \to [n]$.  

Clearly if $n\leq d_{op}$, where $d_{op}$ is the operational dimension of the theory, then $P^{AD}(n,k) =P^{FIX}(n,k)=1$ for all $k \in \naturale$ by \Cref{prop:FIX-prob-1}. In the case $n>d_{op}$ we can also derive a necessary and sufficient condition for $P^{AD}(n,k) =1$. This condition relies on the notion of \emph{antidiscrimination} which captures the idea of state exclusion instead of state discrimination \cite{Caves_antidist,Heinosaari_2018}. Here we need to consider an adaptive version of antidiscrimination:
\begin{defn}
    A set of states $S = \{s_j\}_{j=1}^n$ can be \emph{$k$-adaptively antidiscriminated to at most $l \leq n$ states} if there exists some $k$ adaptive measurements $M^{(1)}, \ldots, M^{(k)}$ as in \Cref{eq:adaptive} each with some $m \in \naturale$ outcomes such that for all outcomes $(i_1, \ldots, i_k) \in [m]^k$ there exists some $j_1, \ldots, j_{l} \in [n]$ such that
    \begin{align}\label{eq:k-adaptive-exclusion}
        M^{(1)}_{i_1}(s_j) M^{(2)}_{i_2|i_1}(s_j) \cdots M^{(k)}_{i_k|i_1,\ldots, i_{k-1}}(s_j) = 0
    \end{align}
    for all $j \in [n] \setminus \{j_1, \ldots, j_l\}$.
\end{defn}

Thus, $k$-adaptive antidiscrimination of $n$ states to at most $l$ states means that each string of outcomes that is obtained from all the adaptive measurements of all $k$ parties can exclude at least $n-l$ states. Using this terminology we can show the following:
\begin{prop} \label{prop:adaptively-antidist-to-dist-states}
    In the case $n>d_{op}$ in any operational theory with operational dimension $d_{op}$, we have that $P^{AD}(n,k) = 1$ if and only if there exists $n$ states $S=\{s_i\}_{i=1}^n$ that can be $(k-1)$-adaptively antidiscriminated to at most $d$ states for some $d\leq d_{op}$ such that the remaining states can always be perfectly discriminated.
\end{prop}
\begin{proof}
    Suppose that $P^{AD}(n,k)=1$. Thus, there exists some set of $n$ states $S= \{s_j\}_{j=1}^n$, some adaptive $m$-outcome measurements $M^{(l)}$ for all $l \in [k]$ and some postprocessing $f: [m]^k \to [n]$ which attain probability 1 in \Cref{eq:adaptive}. For party $l \in [k]$ the quantity $M^{(l)}_{i_l|i_1, \ldots,i_{l-1}}(s_j)$ represents the probability that outcome $i_l$ is obtained when the state is $s_j$ and all the previous $l-1$ parties obtained outcomes $i_1, \ldots, i_{l-1}$. We denote this quantity by $p_{i_l|j,i_1, \ldots,i_{l-1}}$. Now by Bayes' theorem we have that
    \begin{align}\label{eq:Bayes}
       p_{j|i_1, \ldots, i_{l-1}}  p_{i_l|j,i_1, \ldots,i_{l-1}} = p_{j|i_1, \ldots,i_l} p_{i_l|i_1, \ldots, i_{l-1}}
    \end{align}
    for all $(i_1, \ldots, i_l) \in [m]^l$, $l \in [k]$ and $j \in [n]$. Substituting these into the success probability given by \Cref{eq:adaptive} we have that
    \begin{align}
        P^{AD}(n,k) &= \frac{1}{n} \sum_{j=1}^n \sum_{(i_1, \ldots,i_k) \in f^{-1}(j)} p_{i_1|j} p_{i_2|j,i_1} \cdots p_{i_{k-1}|j,i_1, \ldots,i_{k-2}} p_{i_k|j,i_1, \ldots,i_{k-1}} \\
        &= \frac{1}{n} \sum_{j=1}^n \sum_{(i_1, \ldots,i_k) \in f^{-1}(j)} \frac{p_{j|i_1} p_{i_1}}{p_{j}} \frac{p_{j|i_1,i_2} p_{i_2|i_1}}{p_{j|i_1}} \cdots \frac{p_{j|i_1, \ldots,i_{k-1}} p_{i_{k-1}|i_1, \ldots, i_{k-2}}}{p_{j|i_1, \ldots, i_{k-2}}} p_{i_k|j,i_1, \ldots,i_{k-1}} \\
        &= \sum_{j=1}^n \sum_{(i_1, \ldots,i_k) \in f^{-1}(j)} p_{i_1} p_{i_2|i_1} \cdots p_{i_{k-1}|i_1, \ldots,i_{k-2}} p_{j|i_1, \ldots,i_{k-1}} p_{i_k|j,i_1, \ldots,i_{k-1}} \\
        &= \sum_{i_1=1}^{m} \cdots \sum_{i_{k-1}=1}^{m}  p_{i_1, \ldots,i_{k-1}} \sum_{j=1}^n p_{j|i_1, \ldots, i_{k-1}} \sum_{i_k \in f_{i_1, \ldots,i_{k-1}}^{-1}(j)} p_{i_k|j,i_1, \ldots, i_{k-1}},
    \end{align}
    where we have defined the outcome-dependent postprocessing $f_{i_1, \ldots,i_{k-1}}: [m] \to [n]$ by setting $f_{i_1, \ldots,i_{k-1}}(i_k) = f(i_1, \ldots, i_k)$ for all outcomes $(i_1, \ldots,i_k) \in [m]^k$ and we have denoted the joint probability of obtaining outcomes $(i_1, \ldots,i_{k-1})$ by $p_{i_1, \ldots,i_{k-1}} = p_{i_1} p_{i_2|i_1} \cdots  p_{i_{k-1}|i_1, \ldots, i_{k-2}}$ for all $(i_1, \ldots, i_{k-1}) \in  [m]^{k-1}$. 

    Now if $P^{AD}(n,k)=1$, it means that for all $(i_1, \ldots, i_{k-1}) \in [m]^{k-1}$ either $p_{i_1, \ldots, i_{k-1}} = 0$ or 
    \begin{equation}\label{eq:adaptive-perfect}
        \sum_{j=1}^n p_{j|i_1, \ldots, i_{k-1}} \sum_{i_k \in f_{i_1, \ldots,i_{k-1}}^{-1}(j)} p_{i_k|j,i_1, \ldots, i_{k-1}} = \sum_{j=1}^n p_{j|i_1, \ldots, i_{k-1}} \sum_{i_k \in f_{i_1, \ldots,i_{k-1}}^{-1}(j)} M^{(k)}_{i_k|i_1,\ldots, i_{k-1}}(s_j) =1 \,.
    \end{equation}
     We see that without loss of generality we can assume that $p_{i_1, \ldots, i_{k-1}} \neq 0$: if $p_{i_1, \ldots, i_{k-1}} = 0$ it means that the $k-1$ parties did not detect an outcome $(i_1, \ldots,i_{k-1})$. In this case we could simply merge these kind of outcomes in the total measurement of the $k-1$ parties with some other outcome $(i'_1, \ldots, i'_{k-1}) \in [m]^{k-1}$ for which $p_{i'_1, \ldots, i'_{k-1}} \neq 0$; this new outcome, which we can still label by $(i'_1, \ldots, i'_{k-1})$ is now described by the effect
    \begin{equation}
        M_{i'_1} \otimes M_{i'_2|i'_1} \otimes \cdots \otimes M_{i'_{k-1}|i'_1, \ldots,i'_{k-2}} + \sum_{\substack{(i_1, \ldots,i_{k-1}): \\ p_{i_1, \ldots,i_{k-1}} = 0}}  M_{i_1} \otimes M_{i_2|i_1} \otimes \cdots \otimes M_{i_{k-1}|i_1, \ldots,i_{k-2}} \, .
    \end{equation}
    For this new total measurement of the first $k-1$ parties clearly $p_{i_1, \ldots, i_{k-1}} \neq 0$ for all outcomes $(i_1, \ldots,i_{k-1})$.
    
    Let us now fix some outcome string  $(i_1, \ldots, i_{k-1}) \in [m]^{k-1}$. Thus, by \Cref{eq:adaptive-perfect} we must have that the state ensemble $\{p_{j|i_1, \ldots, i_{k-1}}, s_j\}_{j=1}^n$ is perfectly distinguishable. Since $n>d_{op}$ this means that there exists some $d\leq d_{op}$  and $j_1, \ldots, j_{d} \in [n]$ such that only the probabilities $\{p_{j_i|i_1, \ldots, i_{k-1}}\}_{i=1}^{d}$ are non-zero and the states $\{s_{j_i}\}_{i=1}^{d}$ are perfectly distinguishable and the $k$th party manages to discriminate them with his adaptive measurement $M^{(k)}$ and the postprocessing $f_{i_1, \ldots, i_{k-1}}$. Now for any $j \in [n] \setminus \{j_1,\ldots,j_{d}\}$ we thus have that $p_{j|i_1, \ldots, i_{k-1}} = 0$. By \Cref{eq:Bayes} we then have that either $p_{i_{k-1}|j,i_1, \ldots, i_{k-2}}= M^{(k-1)}_{i_{k-1}|i_1, \ldots, i_{k-2}}(s_j)=0$ or $p_{j|i_1, \ldots,i_{k-2}}=0$. This means that either the $(k-1)$th party excluded the state $j$ or it had already been excluded by the previous $k-2$ parties. If it was excluded by the previous $k-2$ parties, i.e., $p_{j|i_1, \ldots,i_{k-2}}=0$, then again by using \Cref{eq:Bayes} we see that then either $p_{i_{k-2}|j,i_1, \ldots, i_{k-3}}= M^{(k-2)}_{i_{k-1}|i_1, \ldots, i_{k-3}}(s_j)=0$ or $p_{j|i_1, \ldots,i_{k-3}}=0$. Again this means that either the $(k-2)$th party excluded the state $j$ or it had already been excluded by the previous $k-3$ parties. We can continue using \Cref{eq:Bayes} this way to see that since the priors at the beginning are nonzero, $p_j = 1/n$, \Cref{eq:k-adaptive-exclusion} must hold for the first $k-1$ parties for the $n-d$ states in $S\setminus \{s_{j_i}\}_{i=1}^{d}$. Since this holds for all outcomes $(i_1, \ldots, i_{k-1}) \in [m]^{k-1}$, the claim follows.
    
\end{proof}
Even though in general it might be hard to easily see when a set of states $S$ can be $k$-adaptively antidiscriminated, an important part of the result is the fact that in order for $S$ to achieve perfect discrimination using adaptive strategy, $S$ must contain some perfectly distinguishable subset. This already rules out perfect adaptive strategies for most sets of quantum states such as the trine states. We will present an instance of (qu)bit-like theory where perfect adaptive discrimination is possible.

\subsection{Multicopy state discrimination in polygon theories}
\label{sec:polygon-theories}

In this section we discuss polygon theories in depth as an example of a GPT other than classical or quantum theory. 
By polygon theories we mean GPTs whose state spaces are regular 2D polygons. 
Polygon theories were introduced in \cite{Janotta_2011}. 
The question arises why one should consider toy theories like polygons when studying natural phenomena and practical applications in information processing tasks. 
What can we learn about physical theories like classical and quantum theory? 
Well, first of all, the operational dimension of all polygons is 2 so that they are (qu)bit-like in this sense. 

As described in the beginning of \Cref{sec:GPTs}, features like nonlocality are nonclassical features and not genuine quantum features. Nonlocality is witnessed by violation of a Bell inequality \cite{chsh}. However, also quantum theory is bounded in Bell inequalities, that is, quantum correlations satisfy Tsirelon's bound \cite{Cirelson1980QuantumGO}. On the other hand, correlations stemming from the previously mentioned PR boxes \cite{rohrlich1995nonlocalityaxiomquantumtheory} are able to achieve the algebraic maximum of Bell inequalities.
Why is this the case?

The axiom of \textit{macroscopic locality} states that any physical theory shall lead to classical theory in the continuum limit \cite{Navascu_s_2009}, meanwhile retrieving Tsirelon's bound if satisfied.
The authors in \cite{Janotta_2011} showed that if the state space of a theory exhibits strong self-duality, then the correlations obtainable from this theory must be compatible with macroscopic locality. This motivates polygons with odd number of vertices (odd polygons) because they are strongly self-dual and hence, should allow for similar correlations as qubit.

On the other hand, the information storability for polygons with even number of vertices (even polygons) is always equal to 2 \cite{Kimura_2018}. It follows that even polygons are (qu)bit-like theories in this sense. This makes a comparison, especially in information and communication tasks, very interesting.

Moreover, polygons are highly symmetric and relatively simple, making them popular GPTs other than classical and quantum to be considered. Notably, adding infinitely many more vertices to polygons yields a circle, which, in turn, corresponds to the state space or \emph{real qubit (rebit)}, a state space corresponding to a projection of the Bloch sphere onto the $x-z$ plane.

\subsubsection{Comparison of strategies} \label{sec:global-stragegies-in-polygons}

Before starting with the actual mathematical definition of a polygon GPT, we present our main results on polygons that give some interesting insights on different strategies in different polygons. Let us start with $SEP$ and $GLOBAL$. For a regular polygon $\polyspace$ with $m$ pure states (vertices) just through iterations of linear programs we can  optimize the average success probability of multi-copy states discrimination in the case of $k=2$ copies of $n=3$ states for both $SEP$ and $GLOBAL$ strategies. 

We note that entanglement in GPTs is not as straightforward as in quantum theory as there is no unique way to define composite system for GPTs. The difference between separable and general global strategies is as follows: In a separable strategy, the measurements must be of separable form (similarly to quantum theory) as tensors. However, in general, a global measurement just needs to be a positive functional on the states of the composite system which we can choose differently. Our choice of composite system corresponds to the so-called \emph{minimal tensor product} which only takes convex combinations of product states as the states of the composite system. By duality the effect space then contains the least restrictive set of functionals which are only required to be positive on product states. In quantum these would correspond to block-positive effects.

The details of the programs can be found in \Cref{appendix:polygon-theories}. In particular we considered 2-copy state discrimination of ensembles of 3 pure states. We optimized the $SEP$ strategies for polygons up to $m=15$ vertices and $GLOBAL$ strategies for polygons up to $m=8$ vertices (the $GLOBAL$ program is computationally more demanding than the $SEP$ program). Our findings can bee seen in \Cref{fig:global-strategies-in-polygons}. 

\begin{figure}[!t]
    \centering
    \includegraphics[width=0.9\textwidth]{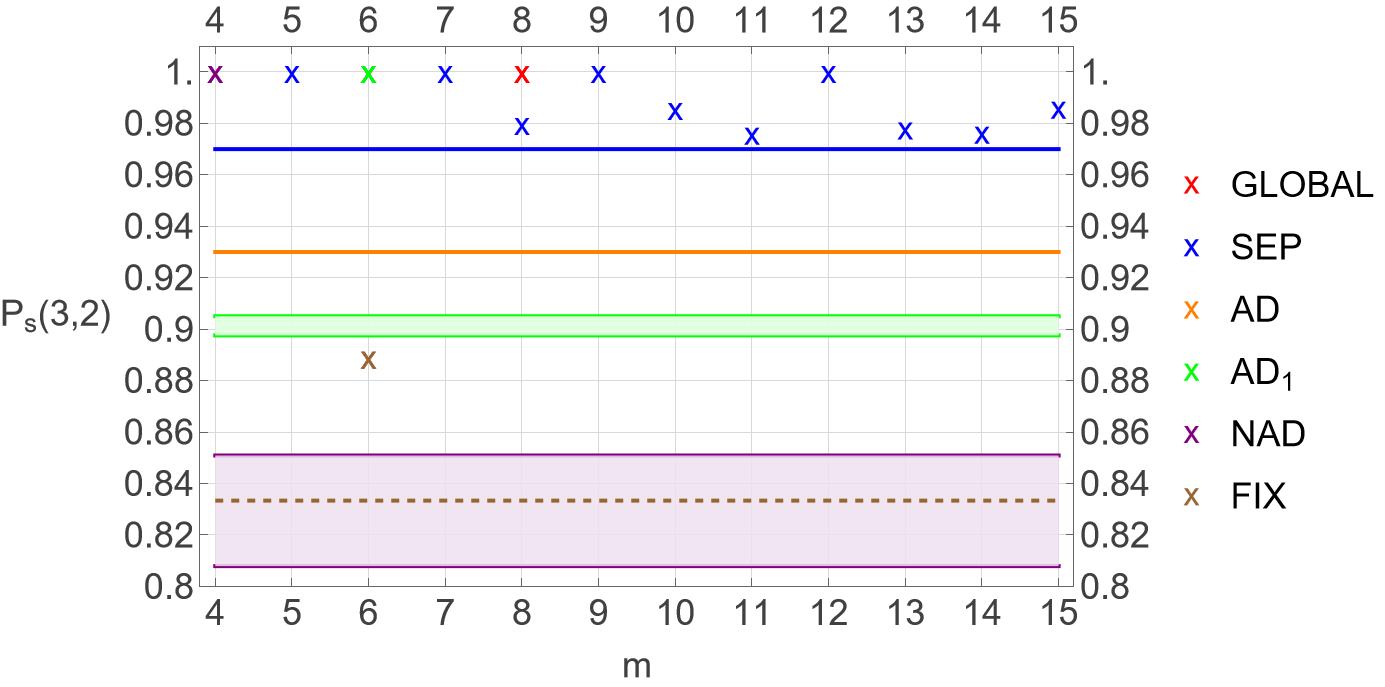}
    \caption{This plot shows the success probability for 2 copies from 3-state ensembles, given a theory and measurement strategy. Solid lines correspond to the trine ensemble and the dashed line to classical theory. For the double trine ensemble, the optimal success probability by \textit{SEP} and \textit{AD} measurements are analytically derived in \cite{eric_optimal_detec}. The bounds for the success probability by $AD_1$ and $NAD$ measurements (green and purple shaded regions respectively) are numerically calculated in \cite{Dutra_2026}. Interestingly, we numerically find that entangled measurements can outperform separable measurements for the polygon theory with $m=8$.} 
    \label{fig:global-strategies-in-polygons}
\end{figure}

Let us analyze \Cref{fig:global-strategies-in-polygons}. We first note that just from the numerical optimization we only obtain $SEP$ and $GLOBAL$ strategies so just from the numerics we obtain the points for polygons either as $SEP$ or $GLOBAL$ and the further refinement of different $LOCC$ (which are also $SEP$) strategies which are represented in the plot is explained below. 

Just from the numerics we see that up until heptagon ($m=7$), it can be seen that separable strategies suffice for perfect state discrimination. Though, the octagon is the first polygon in which a general global measurement is necessary for perfect discrimination. This clearly shows a gap between separable and global strategies. Interestingly, from there on the polygons behave very differently. It would have been intuitive that, due the symmetries and regularities of polygons, there should be patterns. Though, the absence of those patterns motivates further studies.

While the numerical optimization only showed the difference between $SEP$ and $GLOBAL$ in anticipation we have already marked down the further specification of the strategy that we find in the next section. In particular, we see that in square perfect discrimination is possible already with $NAD$ strategy and in hexagon with $AD_1$ strategy. On the other hand, we we find an instance of $FIX$ hexagon strategy that can beat bit unlike qubit. We also show that perfect $AD$ strategies are not possible in any other polygon other than square and hexagon so the perfect $SEP$ strategies cannot be $AD$. Our conjecture which we state in the end is that those perfect $SEP$ strategies cannot be general $LOCC$ either.

Thus our results show that we have found examples of various different strategies ($FIX$, $NAD$, $AD$, $SEP$ and $GLOBAL$) in different polygons which can beat both bit and qubit. This in particular shows that while different measurement strategies do place some restrictions on how well multi-copy discrimination can be achieved, the local properties of the theory play even more significant role as there are simple (qu)bit-like theories which can outperform bit and qubit even with very restricted measurement strategies.

\subsubsection{Construction} \label{sec:construction-of-polygons}
In a polygon theory, the state space $K$ is set to be represented by a polygon with $m$ vertices $\mathcal{P}(m)$ and corresponding effect space $\mathcal{E}(\mathcal{P}(m))$.
Since we often have to distinguish between even and odd polygons, we denote the set of even numbers as $2\mathds{N}$ and write for an even $k$: $k \in 2 \mathds{N}$ as well as for an odd $k$: $k \notin 2 \mathds{N}$.

\begin{defn}
    For $m \in \mathds{N}$ we define vectors
    \begin{equation}
        s_i=
\begin{pmatrix}
    r_m \cos \left(\frac{2 i \pi }{m}\right) \\
    r_m \sin \left(\frac{2 i \pi }{m}\right) \\
    1
\end{pmatrix}, \quad i \in [m], \, r_m = \sqrt{\sec{\frac{\pi}{m}}} 
    \end{equation}
    as the pure states of our theory. The state space is obtained as the convex hull of pure states $\mathcal{P}(m)= \mathrm{conv}\{s_i : i \in [m] \}$.
\end{defn}
Those pure states are also referred to as extreme points of the state space. For given $m \in \mathds{N}$ this yields a regular $m$-sided polygon embedded into $\real^3$.
Above states are normalized in the sense that we define an order unit $\polyone$, that is, a functional from the effect space $\mathcal{E}(\polyspace)$ that always outputs 1 when acting on a state: $\polyone(s)=1 \mspace{5mu} \forall s \in \polyspace$ and we clearly have $\polyone = (0,0,1)^\intercal$.
Furthermore the zero effect is defined as $0 = (0,0,0)^\intercal$. For the rest of the effect space, we have to distinguish between even and odd polygons because they manifest different symmetries. For example, even polygons always have point symmetry around the maximally mixed state $s_0 = (0,0,1)^\intercal$, which is not the case for odd polygons.

\begin{figure}[!t]
    \centering
        \begin{tikzpicture}[scale=1.9] 

          \def\n{7}
          \def\radius{1}

          \foreach \i in {1,...,7} {
            \coordinate (P\i) at ({\radius*cos(360/\n * \i + 90)}, { \radius*sin(360/\n * \i + 90)});

            \fill (P\i) circle (0.012);
          }

          \draw[thick] (P1) -- (P2) -- (P3) -- (P4) -- (P5) -- (P6) -- (P7) -- cycle;

          \foreach \i in {4,...,6} {
            \node[font=\small, right] at (P\i) {${s_\i}$};
            \fill[black] (P\i) circle (0.05);
          }
        \foreach \i in {1,...,3} {
            \node[font=\small, left] at (P\i) {${s_\i}$};
            \fill[black] (P\i) circle (0.05);
          }
        \node[font=\small, above] at (P7) {${s_7}$};
        \fill[black] (P7) circle (0.05);

        \draw[red,dashed,thick] (P1) -- (P2) node[midway,left] {$\bar f_5$};
        \draw[red,dashed,thick] (P2) -- (P3) node[midway,below left] {$\bar f_6$};
        \draw[red,dashed,thick] (P3) -- (P4) node[midway,below] {$\bar f_1$};
        \draw[red,dashed,thick] (P4) -- (P5) node[midway,below right] {$\bar f_7$};
        \draw[red,dashed,thick] (P5) -- (P6) node[midway,right] {$\bar f_2$};
        \draw[red,dashed,thick] (P6) -- (P7) node[midway,above right] {$\bar f_3$};
        \draw[red,dashed,thick] (P7) -- (P1) node[midway,above left] {$\bar f_4$};
        
    \end{tikzpicture}
    \caption{The Heptagon state space $\mathcal{P}(7)$ is represented by pure states $s_i$. The mixed states lie on the edges and inside the polygon. The extreme effects $f_i$ are dual vectors identical to the pure states $s_i$. It can be seen that the complement effects $\bar{f}_i$ correspond to picking out the edges of the polygon.}
    \label{fig:heptagon}
\end{figure}
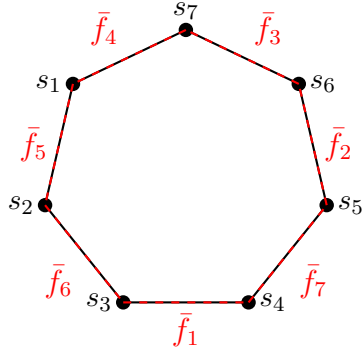

\begin{defn}
    For even polygons $m \in 2\mathds{N}$ the effect space is defined as $\mathcal{E}(\polyspace)= \mathrm{conv}\{o,\polyone,e_1,\ldots,e_m\}$, where
    \begin{equation}
        e_i= \frac{1}{2}
        \begin{pmatrix}
        r_m \cos\left(\frac{(2i-1) \pi }{m}\right) \\
        r_m \sin\left(\frac{(2i-1) \pi }{m}\right) \\
        1
        \end{pmatrix}, \quad i \in [m].
    \end{equation}
\end{defn}
Here, an effect $e_i$ can be viewed as a functional that picks a certain edge in the sense that one gets outcome 'yes' with probability 1 when a state lies on that specific edge. The complement effect $\bar e_i := \one_{\polyspace} - e_i = e_{i \plusmod \frac{m}{2}}$ (where we have denoted $a \plusmod b = a + b \ (\mathrm{mod} \ m)$) thus picks the opposing edge.
In contrast, there are no opposing edges and the state space is self dual in odd polygons. As a consequence, we have to define the complement effects separately.
\begin{defn}
    The set $\mathcal{E}(\polyone) = \mathrm{conv}\{o, \polyone,f_1, \ldots,f_m,\bar f_1, \ldots, \bar f_m \}$ is the effect space of an odd polygon, whereas
    \begin{equation}
        f_i
        = \frac{1}{1+r^2_m}
        \begin{pmatrix}
        r_m \cos\left( \frac{2i \pi }{m}\right) \\
        r_m \sin\left( \frac{2i \pi }{m}\right) \\
        1
        \end{pmatrix}, \quad \bar f_i = \polyone -f_i, \, i \in [m].
    \end{equation}
\end{defn}

One example is provided in \Cref{fig:heptagon}.
Our setting for the state discrimination task is as follows: A source prepares two copies of an unknown state, where Alice and Bob receive one each respectively.
If not stated otherwise, we are always going to consider ensembles of three pure states $S = \{s_x,s_y,s_z\}$ stemming from $\polyspace$, whereas it is assumed that all states are different $x \neq y, \, x \neq z \,, y \neq z$.
In addition, unless stated otherwise, here we will only consider strategies in $AD$, i.e., strategies that with onlt one-way communication, that is, w.l.o.g.\ we are assuming that Alice is the only party to send and Bob is the only party to receive classical messages.
Moreover, we exclude the cases $m \in \{1,2,3\}$ as the first one is trivial and the latter two correspond to probability simplices, i.e.\ classical bit- or trit-theory.

\subsubsection{Local strategies in polygons}

Let us start with the fixed measurement strategy. As all polygons have operational dimension $d_{op}=2$, we see from \Cref{prop:FIX-prob-1} that perfect discrimination of $k$ copies of $n$ states is not possible for any $n >3$  and $k\in \naturale$. However, we see that in hexagon we find a set of $n=3$ states such that $k=2$ copies of them can be discriminated with higher average success probability than what is possible in bit. Let us take $S= \{s_2,s_4,s_6\}$ and $M= \{\frac{2}{3}e_2, \frac{2}{3}e_4,\frac{2}{3} e_6\}$. Furthermore, let us consider the following postprocessing function $f: [3]^2 \to [3]$: 
\begin{align}
    f(1,1)=f(1,2) = f(2,1)=1, \\
    f(2,2)=f(2,3) = f(3,2)=2, \\
    f(3,3)=f(3,1) = f(1,3)=3 \, .
\end{align}
Using these states, measurement and postprocessing in the fixed measurement strategy in \Cref{eq:fix-def} one obtains the success probability $\frac{8}{9}$. Thus, 
\begin{align}
    P^{FIX}_{\mathcal{P}(6)}(3,2) \geq \frac{8}{9} > \frac{5}{6} =  \Pb(3,2) \, .
\end{align}
Again, our previous resulsts show that hexagon beats qubit as well in $FIX$ strategies.

Let us move on to the non-adaptive strategies $NAD$. The smallest polygon we consider is the square, which already portrays a special case. Namely, the set of pure states can be \emph{perfectly pairwisely distinguished}, that is, every pair $\{s_a,s_b\}  \subset \mathcal{P}(4)$ can be perfectly discriminated. For each pair of pure states the discriminating measurement can be chosen to be $M=\{e, \bar{e}\}$, where $e \in \{e_1,e_2\}$ (see \Cref{fig:hexagon_and_square}). It is exactly this property that makes the square state space extremely powerful in state discrimination tasks, allowing for perfect discrimination in quite restricted settings.

\begin{prop}
    All four pure states of square $\mathcal{P}(4)$ can always be perfectly discriminated with $k=2$ copies with a non-adaptive strategy with only one bit of communication. That is, $P^{NAD}_{\mathcal{P}(4)}(n,2)=1$ for $n \in \{3,4\}$.
\end{prop}
\begin{proof}
    Taken from \Cref{fig:hexagon_and_square}, w.l.o.g.\ we apply $e_2$ to check if the drawn state is on a top or bottom edge of the square and save the outcome into one bit $a$. Next, we check with $e_1$ if the drawn state is on the left or right edge of the square. Together with bit $a$, this uniquely determines every pure state in $\mathcal{P}(4)$. Thus Alice and Bob always apply non-adaptive measurements $\{e_2, \bar{e}_2\}$ and $\{e_1, \bar{e}_1\}$, respectively.
    \begin{equation}
        \xymatrix{
        & & e_1 \ar@{=>}[r]^{b=0} \ar@{=>}[rd]_{b=1} &s_1\\
        & & &s_2\\
        \Box \ar[r]^{s^{\otimes 2}} &e_2 \otimes \id \ar[rdd]^{s\in \{s_3,s_4\}}_{a=1} \ar[ruu]^{a=0}_{s\in \{s_1,s_2\}}\\
        & & &s_3\\
        & & e_1 \ar@{=>}[r]_{b=0} \ar@{=>}[ru]^{b=1} &s_4
        }
    \end{equation}
\end{proof}

\begin{figure}
    \begin{minipage}{.45\textwidth} 
        \begin{tikzpicture}[scale=2]
          \def\n{6}
          \def\radius{1}
          \foreach \i in {1,...,6} {
            \coordinate (P\i) at ({\radius*cos(360/\n * \i + 90)}, { \radius*sin(360/\n * \i + 90)});
            \fill (P\i) circle (0.012);
          }
          \draw[thick] (P1) -- (P2) -- (P3) -- (P4) -- (P5) -- (P6) -- cycle;
          \draw[red,dashed,thick] (P1) -- (P2) node[midway,left,xshift=-2pt] {$e_2$};
          \draw[red,dashed,thick] (P2) -- (P3) node[midway,below left,xshift=-2pt] {$e_3$};
          \draw[red,dashed,thick] (P3) -- (P4) node[midway,below right,xshift=2pt] {$e_4$};
          \draw[red,dashed,thick] (P4) -- (P5) node[midway,right,xshift=2pt] {$e_5$};
          \draw[red,dashed,thick] (P5) -- (P6) node[midway,above right,xshift=2pt] {$e_6$};
          \draw[red,dashed,thick] (P6) -- (P1) node[midway,above left,xshift=-2pt] {$e_1$};

          \foreach \i in {1,...,2} {
            \node[font=\small, left] at (P\i) {${s_\i}$};
            \fill[black] (P\i) circle (0.05);
          }
        \foreach \i in {4,...,5} {
            \node[font=\small, right] at (P\i) {${s_\i}$};
            \fill[black] (P\i) circle (0.05);
          }
        \node[font=\small, below] at (P3) {${s_3}$};
        \fill[black] (P3) circle (0.05);

        \node[font=\small, above] at (P6) {${s_6}$};
        \fill[black] (P6) circle (0.05);
        \node at (-1,-1.2,.5) {a)};
        \end{tikzpicture}
    \end{minipage}
    \hfill
    \begin{minipage}{.45\textwidth}
        \begin{tikzpicture}[scale=.8]
            \draw[thick, black] (0,0) rectangle (4,4);
        
            \node[below left] at (0,0) {$s_1$};
            \fill[black] (0,0) circle (0.1);
            \node[below right] at (4,0) {$s_2$};
            \fill[black] (4,0) circle (0.1);
            \node[above right] at (4,4) {$s_3$};
            \fill[black] (4,4) circle (0.1);
            \node[above left] at (0,4) {$s_4$};
            \fill[black] (0,4) circle (0.1);
        
            \draw[red,dashed,thick] (0,-1) -- (0,5) node[midway,left,xshift=-2pt] {$e_1$};
            \draw[red,dashed,thick] (-1,0) -- (5,0) node[midway,below,yshift=-2pt] {$e_2$};
            \draw[red,dashed,thick] (4,-1) -- (4,5) node[midway,right,yshift=2pt] {$e_3$};
            \draw[red,dashed,thick] (-1,4) -- (5,4) node[midway,above,yshift=+2pt] {$e_4$};
            \node at (0,-1.2,2) {b)};
        \end{tikzpicture}
    \end{minipage}
    \caption{(a) The black nodes in hexagon state space represent the pure states, whereas the pure effects (red) can be understood as projectors on an edge. Thus, effects of opposing edges, e.g.\ $e_1$ and $e_4$ are inverse to each other $e_1 = \mathds{1}_{\mathcal{P}(6)} - e_4$. (b) Square state space is of unique structure, that is, every pair of states can be distinguished by one pure effect.}
    \label{fig:hexagon_and_square}
\end{figure}

We can now show that it is exactly the property of pairwise distinguishability that is needed for perfect discrimination for any adaptive strategy in polygons.

\begin{prop}\label{prop:polygon-perfect-ad}
    Perfect adaptive discrimination of $k=2$ copies of $n=3$ states $S \subset \polyspace$ for $m>4$ is possible only if the states in $S$ are perfectly pairwisely distinguishable.
\end{prop}
\begin{proof}
    According to \Cref{prop:adaptively-antidist-to-dist-states} perfect discrimination of $k=2$ copies of $n=3$ states it is possible only when there exists a set of $3$ states $S=\{t_1, t_2, t_3\}$ that can be $1$-adaptively discriminated into a set of at most $d\leq 2$ perfectly discriminable states. This simply means that there exists a measurement $M$ which perfectly excludes one or two states from $S$ and the remaining states are perfectly distinguishable. 
    
    Suppose that $M$ excludes two states, so that w.l.o.g.\ $M_1(t_1)=M_1(t_2)=0$. It is straightforward to see that in polygons this can only happen if $t_1$ and $t_2$ lie on the same edge of the polygon and the effect $M_1$ is complement to the extreme effect that picks that edge. w.l.o.g.\ we can have $t_1$ and $t_2$ lie on the edge between the pure states $s_1$ and $s_2$ in which case $M_1 = \bar{e}_2$ (even polygons) or $M_1 = \bar{f}_2$ (odd polygons). Since $M$ is a measurement, there exists other outcomes $i,j \in \{2, \ldots,m\}$ such that $M_{i}(t_1)>0$ and $M_{j}(t_2)>0$. 

     Suppose now that we cannot choose an outcome $i$ such that $M_{i}(t_1)>0$ and $M_{i}(t_2)>0$. Thus, for all other outcomes $i \in \{2, \ldots,m\}$, we have that either $M_{i}(t_1) >0$ and $M_{i}(t_2) = 0$ or $M_{i}(t_2) >0$ and $M_{i}(t_1) = 0$. Since $t_1$ and $t_2$ lie on the same edge this is only possible if they are the two extreme points of that edge, i.e., $t_1=s_1$ and $t_2=s_2$. It can be seen that then $M_{i}(t_1) >0$ and $M_{i}(t_2) = 0$ is only possible if $M_i$ is a convex mixture of the zero effect and effects $\bar{e}_2$ and $\bar{e}_3$ (even polygons) or effects $\bar{f}_2$, $f_{2\plusmod\frac{m-1}{2}}$  and $f_{2\plusmod\frac{m+1}{2}}$ (odd polygons). However, now if $m>4$ we can never fulfill the normalization condition  $\sum_i M_i = \one_{\polyspace}$. This is a contradiction. Thus, there always exists outcome $i$ such that $M_{i}(t_1)>0$ and $M_{i}(t_2)>0$. Since $M$ must always exclude at least one state it follows that $M_i(t_3)=0$. Then clearly there is a non-zero probability of detecting outcome $i$ when measuring $M$ on a randomly chosen state $t \in S$ and in that case we can only exclude state $t_3$. Thus states $t_1$ and $t_2$ must be perfectly distinguishable. However, for polygons larger than square, perfect discrimination of two states lying on the same edge is never possible (arguments from the proof of upcoming \Cref{prop:no-pairwise-dist} can be used to see this). Thus, $M$ cannot exclude two states.
    
    It now follows that $M$ must exclude just one state so that w.l.o.g.\ $M_1(t_1)=0$. Then clearly if we detect outcome $1$ (or any other outcome $i$ such that $M_i(t_1)=0$), the randomly chosen state from $S$ must be either $t_2$ or $t_3$ and they must be perfectly distinguishable. On the other hand, if we detect an outcome $j \in [m]$ such that $M_j(t_1)>0$, it follows that either $M_j(t_2)=0$ or $M_j(t_3)=0$ but not both can be zero. Suppose that $M_j(t_2)=0$. Then the states $t_1$ and $t_3$ must be perfectly distinguishable. Similarly is $M_j(t_3)=0$ the states $t_1$ and $t_2$ must be perfectly distinguishable. Suppose that there does not exist an outcome $j$ such that $M_j(t_3)=0$ and $M_j(t_1)>0$. Then $M_j(t_2)=0$. Let us denote $e = \sum_{j: \ M_j(t_1)=0} M_j$. Then clearly $e(t_1) = 0$ and $\bar{e}(t_1)=1$. On the other hand now $\bar{e}(t_2)=0$ and $e(t_2)=1$. Then the states $t_1$ and $t_2$ are perfectly distinguishable. By similar arguments if there does not exist an outcome $j$ such that $M_j(t_2)=0$ and $M_j(t_1)>0$, then $t_1$ and $t_3$ must be perfectly distinguishable. Thus, we have exhausted all the possibilities for all the outcomes of $M$ and in all cases the states in $S$ are perfectly pairwisely distinguishable.
    \end{proof}

As we saw earlier perfect (non)adaptive discrimination for $k=2$ copies of $n \in \{3,4\}$ states in square was possible and as was just observed in \Cref{prop:polygon-perfect-ad} it was possible only because of perfect pairwise distinguishability. However, for $n=3$ pure states in polygons we can show the following no-go result about pairwise distinguishability:

\begin{prop} \label{prop:no-pairwise-dist}
    Let $m \notin\{3, 4,6\}$. There does not exist a set of three states in $\mathcal{P}(m)$ that can be perfectly pairwisely distinguished.
\end{prop}
\begin{proof}
    We note that for pairwise distinguisability it is enough to consider dichotomic measurements and since the success probability is a convex function with respect to the measurement we can restrict to extreme dichotomic measurements. Thus in even polygons it means that the distinguishing measurement is always of the form $\{e_u, \bar{e}_u\}$ and in odd polygons it is of the form $\{f_u, \bar{f}_u\}$ for some $u \in [m]$. Also, if three states $\{t_1, t_2, t_3\}$ are perfectly pairwisely distinguishable, and they can be written in convex combinations $t_i = \sum_{l=1}^{a_i} \alpha^{(i)}_l s^{(i)}_l$ for some $a_i \in [m]$, $\alpha^{(i)}_l >0$, $\sum_{l=1}^{a_i} \alpha^{(i)}_l =1$ for some set of pure states $\{s^{(i)}_l\}_l$ for all $i \in [3]$, it follows that also all sets of three pure states $\{s^{(1)}_{l_1},s^{(2)}_{l_2},s^{(3)}_{l_3}\}$ are perfectly distinguishable for all $l_i \in [a_i]$ for all $i \in [3]$. Thus, we can also only focus on pure states in the proof. We first only consider even polygons and odd polygons will follow in a similar manner afterwards.
    
    Let $s_p \in S=\{s_x,s_y,s_z\} \subset \polyspace$ be given, where $m \in 2 \naturale$. We compute the states that can be perfectly distinguished from $s_p$.
    We start by calculating the effect $e_u$ that has zero probability measuring state $s_p$ using above vector-definitions of effects and states.
    \begin{equation}
        e_u(s_p) = 0 \Leftrightarrow -\cos{\left(\frac{\pi}{m}\right)} = \cos{\left(\frac{\pi(1+2p-2u)}{m}\right)}
    \end{equation}
    We want to understand the difference of indices $p-u$ as a condition for measuring outcome $u$ given state $s_p$ with zero probability. The next step thus consists of solving for $p-u$.
    \begin{equation}
        \cos{(\theta)} = x \Rightarrow \theta = (-1)^r \arccos{(x)} + r \pi + (-1)^{r+1} \frac{\pi}{2}+ \frac{\pi}{2}
    \end{equation}
    for some $r \in \mathds{Z}$ and in our case $x=\cos{\left(\frac{\pi(n-1)}{m}\right)}$. Plugging this into previous equation leads to:
    \begin{gather}
        \theta = (-1)^r \frac{\pi(m-1)}{m}+\pi r + \frac{\pi}{2}(1+(-1)^{r+1}) = \frac{\pi}{m}(1+2p-2u) \\
        \Leftrightarrow p-u = -\frac{1}{2}\underbrace{(1+(-1)^{r+1})}_{(i)} + \frac{m}{2} \underbrace{\left(r + \frac{1}{2}+\frac{(-1)^r}{2}\right)}_{(ii)} \\
        \Rightarrow p-u = \begin{cases} \frac{m}{2} : r = 0 \\ \frac{m}{2} -1 : r = 1 \end{cases} \label{eq:alpha-beta}
    \end{gather}
    As the last equation shows, there is no unique solution for obtaining zero probability.
    However, (\textit{ii}) is just a multiple of $m$ depending on $r$. Due to the '$m$-periodicity' of an $m$-polygon, we just need to consider the first two solutions since the others are equivalent. In hexagon for example, taking 9 steps is the same as taking 3 steps since one ends up at the same place after 6 steps. By steps , we mean going from vertex to vertex in one direction following the edges of the polygon.
    This also enforces us to do addition and subtraction without going in a circle, i.e.\ only add/subtract the remainder if the addend is larger than $m$. To capture this, we write $a \oplus_mb \coloneqq a + ( a + b \mod m)$ and similarly $a \ominus_m b$.
    (\textit{i}) is a downshift by 1 depending on $r$. Therefore, $u$ and $p$ must be either $\frac{m}{2}-1 \eqqcolon \alpha$ or $\frac{m}{2} \eqqcolon \beta$ steps apart from each other in order to yield zero probability $e_u(s_p) = 0 \Leftrightarrow u \in \{p \ominus_m \alpha,p \ominus_m \beta\}$.

    If we want to perfectly distinguish $\{s_x, s_y\} \subset S$, we need to pick $e \in \mathcal{E}(\mathcal{P}(m))$ in such a way that
    \begin{equation} \label{eq:conditions-for-dist}
            \underbrace{ e(s_x)}_{(a)} = \underbrace{ \bar e(s_y)}_{(b)} = 0, \quad \underbrace{e(s_y)}_{(c)} > 0 \, ,
    \end{equation}
    where $\bar e = \mathds{1}_{\mathcal{P}(m)} - e$.
    Therefore, $e$ must be chosen as follows.
    \begin{equation}
        e \overset{(a)}{=} \begin{cases} e_{x \ominus_m \beta} \\ e_{x \ominus_m \alpha} \end{cases} \, \Rightarrow \bar e = \begin{cases} e_{(x \ominus_m \beta) \oplus_m \frac{m}{2}}=e_x \\ e_{(x \ominus_m \alpha) \oplus_m \frac{m}{2}}=e_{x \oplus_m 1} \end{cases}
    \end{equation}
    It is convenient to note that $\beta - \alpha = 1$. We find the complement effect in even polygons just by walking to the opposite side, i.e.\ half all the vertices: $\bar e_{(x)}=e_{x\oplus_m \frac{m}{2}}$.
    Using above cases to write down proper measurements yields the relation index $y$ must have to $x$.
    
    \begin{minipage}{.43\textwidth}
        \begin{align}
            &M= \{e_{x \ominus_m \beta}, e_x\} \\
            &(c) \, y \notin \{x \ominus_m 1,x\}\\
            &(b) \, y \in \{x \oplus_m \alpha, x \oplus_m \beta\}
        \end{align}
    \end{minipage}
    \hfill
    or
    \hfill
    \begin{minipage}{.43\textwidth}
        \begin{align}
            &M = \{e_{x \ominus_m \alpha}, e_{x \oplus_m 1}\} \\
            &(c) \, y \notin \{x, x \oplus_m 1\} \\
            &(b) \, y \in \{x \oplus_m \beta, x \oplus_m \beta \oplus_m 1\} \, ,
        \end{align}
    \end{minipage}
    
    where we applied the measurement on $s_y$ and then made use of \Cref{eq:conditions-for-dist,eq:alpha-beta}.
    Bringing it all together, we have:
    \begin{equation} \label{eq:relation-of-x-y-for-dist}
        y \in \{x \oplus_m \alpha, x \oplus_m \beta, x \oplus_m \beta \oplus_m 1 \} 
    \end{equation}
    Note that if $m=4$, \Cref{eq:relation-of-x-y-for-dist} contains all other states from the square, so all states can be distinguished in principle. However, once a measurement is picked, not all pairs of states can be distinguished by this specific one.
    
    Now, we want to distinguish all possible pairs from $S$. We must choose $y$ and $z$ according to \Cref{eq:relation-of-x-y-for-dist} if we want to distinguish $(s_x,s_y)$ and $(s_x,s_z)$. The closest we can put $s_y$ relative to $s_x$ is exactly $\alpha=\frac{m}{2}-1$ steps apart.
    Since we also want to be able to distinguish $(s_y,s_z)$, the least $s_y$ and $s_z$ are allowed to be apart is again $\alpha$.
    So all states from the ensemble must be at least $\alpha$ steps apart to ensure distinguishability of every pair.
    
    However, this stops working for $6 < m \in 2\mathds{N}$: $3 \alpha \leq m \Leftrightarrow m \leq 6$.
    That is, even polygons of bigger size do not provide 'enough space to fit distinguishable states' as square and hexagon do.

    We want to repeat the same argument for odd polygons $\mathcal{P}(m)$ with $m \notin 2 \mathds{N}$. However, the conditions are slightly different here as odd polygons do not have opposing faces and require different effects as defined in \Cref{sec:polygon-theories}.
    Thus, the first step again consists of setting $ f_u(s_p) = 0$ and solving for $p-u$. This yields:
    \begin{equation}
        p-u = \begin{cases}
            \frac{m-1}{2}+\frac{mr}{2} \text{ for some even } r\\
            -\frac{m-1}{2}+m(\frac{r+1}{2}) \text{ for some odd } r
        \end{cases} = \begin{cases}
            \frac{m(r+1)-1}{2} \text{ for some even } r\\
            \frac{mr+1}{2} \text{ for some odd } r
        \end{cases}
    \end{equation}
    Due to the $m$-periodicity of the polygon we again only need to consider $r$ such that $p-u \in [m]$. More precisely, it can be seen in above equation that $p-u$ has a fixed, offset-like, part and one periodic part. That means adding a multiple of $m$ to $p-u$ doesn't give a new relation (one could also just look at equivalence classes). We hence restrict ourself to the first two $r\in \{0,1\}$ and define $\alpha' \coloneqq \frac{m-1}{2}, \, \beta' \coloneqq \frac{m+1}{2}$ and conclude $p-u \in \{ \alpha',\beta'\}$.
    
    Let $S=\{s_x,s_y,s_z\} \subset \mathcal{P}(m)$. We again make use of \Cref{eq:conditions-for-dist}, which leads us to:
    
    \begin{minipage}{.45\textwidth}
         \begin{align}
            &M = \{f_{x \ominus_m \alpha'}, \mathds{1}_{\mathcal{P}(m)} - f_{x \ominus_m \alpha'} \} \\
            &(c) \, y \notin \{x, x \oplus_m 1\} \\
            & (b) \, y \in \{x \ominus_m \alpha'\}
         \end{align}
    \end{minipage}
    \hfill
    \begin{minipage}{.45\textwidth}
        \begin{align}
            &M = \{f_{x \ominus_m \beta'}, \mathds{1}_{\mathcal{P}(m)} - f_{x \ominus_m \beta'} \} \\
            &(c) \, y \notin \{x \ominus_m 1, x \} \\
            &(b) \,  y \in \{x \ominus_m \beta'\}
        \end{align}
    \end{minipage}
    
    All in all, it follows that $y \in \{x \ominus_m \alpha', x \ominus_m \beta'\}$.
    It is apparent that we are left with even less options as in the $m \in 2 \mathds{N}$ case and we can easily repeat the argument from before:
    we have $\alpha = \frac{m-1}{2}$ and it follows in the same manner that: $3 \alpha \leq m \Leftrightarrow m \leq 3$. Thus, it doesn't work for any odd polygon (except the omitted triangle). Consequently, odd polygons lack this sort of 'state orthogonality' totally.
\end{proof}

 Thus, it follows that for polygons $\polyspace$ for $m \notin\{3, 4,6\}$ perfect adaptive discrimination of $k=2$ copies of $n=3$ states is not possible. Before phrasing this result formally, let us first see what happens in the case of hexagon, which is currently the only polygon where we have not considered perfect pairwise distinguishability in. 

\begin{prop}
    Up to rotation, there exists exactly one set of three pure states  $S= \{s_2, s_4, s_6\} \subset \mathcal{P}(6)$ that can be perfectly distinguished using adaptive measurement and 1 bit of communication. In particular, $P^{AD_1}_{\mathcal{P}(6)}(3,2)=1$.
\end{prop}
Here we refer to the states and effects as depicted in \Cref{fig:hexagon_and_square}.
\begin{proof}
    Let us see when three pure states $\{s_x,s_y,s_z\}$ can be perfectly distinguished. w.l.o.g.\ we can pick $s_x=s_6$ to be our first state. It follows from the proof of \Cref{prop:no-pairwise-dist} that we must have $s_y,s_x \in \{s_2,s_3,s_4\}$ in order to distinguish from $s_6$. Since $\{s_2,s_4\}$ is distinguishable, however $\{s_2,s_3\}$ and $\{s_3,s_4\}$ are not, we have w.l.o.g.\ that $s_y=s_2$ as well as $s_z=s_4$.
    It is clear that a fixed strategy cannot be sufficient, since the possible pairs require different effects for discrimination. The perfect strategy is as follows:
    \begin{equation}
            \xymatrix{
            & & e_3 \ar@{=>}[r]^{b=0} \ar@{=>}[rd]_{b=1} &s_2\\
            & & &s_6\\
            \Box \ar[r]^{s^{\otimes 2}} &e_2 \otimes \id \ar[rdd]^{s\in \{s_4,s_6\}}_{a=1} \ar[ruu]^{a=0}_{s \in \{s_2,s_6\}}\\
            & & &s_6 \\
            & & e_4 \ar@{=>}[r]_{b=0} \ar@{=>}[ru]^{b=1} &s_4
            }
        \end{equation}
    This works especially because $ e_3(s_6) = e_4(s_6) = 0$. The outcome probability of applying $e_2$ on $s_6$ hence does not matter since we can always exclude $s_2$ or $s_4$ in the first step.
    It is clear via rotating the polygon by $\pi$, that the ensemble $\{s_3,s_1,s_5\}$ can equivalently be perfectly discriminated.
    Moreover, as soon as we pick a different ensemble $\tilde S$ there is going to be at least one pair $\{s_x,s_y\} \subset \tilde S$, which shares an edge and cannot be distinguished.
\end{proof}

It follows that square and hexagon are the only polygons where perfect discrimination of two copies of three states using adaptive measurements is possible. Thus, collecting all the results of the section we can state the main result of this section:

\begin{thm}
    There is no adaptive measurement strategy that achieves perfect discrimination of $k=2$ copies of $n=3$ states in the polygon state space $\mathcal{P}(m)$ for $m \in \mathds{N}\setminus\{3,4,6\}$.
\end{thm}

Case $m \in [3]$ is the trit where $d_{op}=3$ so that by \Cref{prop:FIX-prob-1} even perfect fixed measurement strategy works for $n=3$. On the other hand for $m \in  \{4,6\}$ we have already seen that perfect adaptive discrimination is possible.

Even though adaptive measurement strategy is only an instance of general LOCC strategies, we believe that also perfect LOCC strategy is impossible in other polygons than triangle, square and hexagon.

\begin{con}
    There is no LOCC protocol that achieves the perfect discrimination of $k=2$ copies of $n=3$ states in polygon theory $ \mathcal{P}(m)$, where $m \in \naturale \setminus \{3,4,6\}$.
\end{con}

Together with what we have seen in \Cref{fig:global-strategies-in-polygons} from \Cref{sec:global-stragegies-in-polygons}, this would show a gap between LOCC and separable strategies. This is known under the terminology nonlocality without entanglement (NLWE), firstly established in \cite{Bennet_NLWE}. In our context, NLWE means that there is an advantage over LOCC, i.e.\ local strategies, even though the operators are still of separable form.
Another instance of NLWE in polygons was found in \cite{Bhattacharya_2020}, using separable, multipartite states that could not be discriminated by LOCC strategies.

\section*{Note added}
During the final stages of this work, we became aware of a related manuscript by Kvashchuk \emph{et al.}, titled “The most discriminable quantum states in the multicopy regime” \cite{kvashchuk2026discriminablequantumstatesmulticopy}. The two works were carried out independently and address closely related questions.

\section*{Acknowledgement}
\noindent
The authors acknowledge financial support from the Business Finland project \mbox{BEQAH}.

\bibliographystyle{unsrturl}
\bibliography{sources}

@misc{Quintino_2026,
      title={The most discriminable quantum states in the multicopy regime}, 
      author={Kvashchuk, M. and Chernyshova, P.  and Porto, L. E.A. and Ohst, T.-A. and  Vieira, L. B. and Quintino, M. T. },
      year={2026},
      eprint={2604.xxxxx},
      archivePrefix={arXiv},
      primaryClass={quant-ph}
}

@misc{kvashchuk2026discriminablequantumstatesmulticopy,
      title={The most discriminable quantum states in the multicopy regime}, 
      author={Kvashchuk, M. and Chernyshova, P. and Porto, L.E.A. and Ohst, T.A. and Vieira, L.B. and Quintino, M.T.},
      year={2026},
      eprint={2604.26927},
      archivePrefix={arXiv},
      primaryClass={quant-ph}
}

@misc{Kimura_2018,
      title={Information storing yields a point-asymmetry of state space in general probabilistic theories}, 
      author={Matsumoto, K. and Kimura, G.},
      year={2018},
      eprint={1802.01162},
      archivePrefix={arXiv},
      primaryClass={quant-ph}
}

@article{heinosaari2024can,
  title={Can a qudit carry more information than a dit?},
  author={Heinosaari, T. and Hillery, M.},
  journal={Contemp. Phys.},
  volume={65},
  number={1},
  pages={2--11},
  year={2024},
  doi={10.1080/00107514.2024.2390279}
}

@article{Chiribella_2011,
  title = {Informational derivation of quantum theory},
  author = {Chiribella, G. and D'Ariano, G. M. and Perinotti, P.},
  journal = {Phys. Rev. A},
  volume = {84},
  issue = {1},
  pages = {012311},
  numpages = {39},
  year = {2011},
  publisher = {American Physical Society},
  doi = {10.1103/PhysRevA.84.012311}
}

@article{Masanes_2011,
    doi = {10.1088/1367-2630/13/6/063001},
    year = {2011},
    publisher = {},
    volume = {13},
    number = {6},
    pages = {063001},
    author = {Masanes, L. and Müller, M. P.},
    title = {A derivation of quantum theory from physical requirements},
    journal = {New J. Phys.}
}

@article{Montanaro_2007,
   title={On the Distinguishability of Random Quantum States},
   volume={273},
   ISSN={1432-0916},
   DOI={10.1007/s00220-007-0221-7},
   number={3},
   journal={Commun. Math. Phys.},
   publisher={Springer Science and Business Media LLC},
   author={Montanaro, A.},
   year={2007},
   pages={619–636}
   }

@article{Zhou_2025,
   title={On the distinguishability of geometrically uniform quantum states},
   volume={58},
   ISSN={1751-8121},
   DOI={10.1088/1751-8121/ae0a95},
   number={41},
   journal={J. Phys. A: Math. Theor.},
   publisher={IOP Publishing},
   author={Zhou, J. and Chessa, S. and Chitambar, E. and Leditzky, F.},
   year={2025},
   pages={415303} 
   }

@article{Dutra_2026,
   title={Structure of quantum measurements implementable with one round of classical communication},
   volume={89},
   ISSN={1361-6633},
   DOI={10.1088/1361-6633/ae4fef},
   number={3},
   journal={Rep. Prog. Phys.},
   publisher={IOP Publishing},
   author={Dutra, A. C. R. and Ohst, T.-A. and Nguyen, H.-C. and Gühne, O.},
   year={2026},
   pages={037601}
}

@article{eric_optimal_detec,
  title = {Revisiting the optimal detection of quantum information},
  author = {Chitambar, E. and Hsieh, M.-H.},
  journal = {Phys. Rev. A},
  volume = {88},
  issue = {2},
  pages = {020302},
  numpages = {4},
  year = {2013},
  publisher = {American Physical Society},
  doi = {10.1103/PhysRevA.88.020302},
}

@Article{JOUR,
    author={Ban, M. and Kurokawa, K. and Momose, R. and Hirota, O.},
    title={Optimum measurements for discrimination among symmetric quantum states and parameter estimation},
    journal={Int. J. Theor. Phys.},
    year={1997},
    volume={36},
    number={6},
    pages={1269-1288},
    issn={1572-9575},
    doi={10.1007/BF02435921}
}

@article{Jozsa_2000,
   title={Distinguishability of states and von {N}eumann entropy},
   volume={62},
   ISSN={1094-1622},
   pages = {012301},
   DOI={10.1103/physreva.62.012301},
   number={1},
   journal={Phys. Rev. A},
   publisher={American Physical Society (APS)},
   author={Jozsa, R. and Schlienz, J.},
   year={2000}
   }

@article{Johnston_2025,
   title={Tight bounds for antidistinguishability and circulant sets of pure quantum states},
   volume={9},
   ISSN={2521-327X},
   DOI={10.22331/q-2025-02-04-1622},
   journal={Quantum},
   publisher={Verein zur Forderung des Open Access Publizierens in den Quantenwissenschaften},
   author={Johnston, N. and Russo, V. and Sikora, J.},
   year={2025},
   pages={1622}
   }

@article{Hausladen01121994,
author = {Hausladen, P. and Wootters, W. K.},
title = {A ‘Pretty Good’ Measurement for Distinguishing Quantum States},
journal = {J. Mod. Opt.},
volume = {41},
number = {12},
pages = {2385--2390},
year = {1994},
publisher = {Taylor \& Francis},
doi = {10.1080/09500349414552221}
}

@ARTICLE{Optimal_detec_sym_states,
  author={Eldar, Y.C. and Megretski, A. and Verghese, G.C.},
  journal={IEEE Trans. Inf. Theory}, 
  title={Optimal detection of symmetric mixed quantum states}, 
  year={2004},
  volume={50},
  number={6},
  pages={1198-1207},
  doi={10.1109/TIT.2004.828070}
  }

@article{Barnum:2000cwz,
    author = "Barnum, H. and Knill, E.",
    title = "{Reversing quantum dynamics with near-optimal quantum and classical fidelity}",
    doi = "10.1063/1.1459754",
    journal = "J. Math. Phys.",
    volume = "43",
    number = "5",
    pages = "2097",
    year = "2002"
}

@article{Harrow_2012,
   title={How Many Copies are Needed for State Discrimination?},
   volume={58},
   ISSN={1557-9654},
   DOI={10.1109/tit.2011.2169544},
   number={1},
   journal={IEEE Trans. Inf. Theory},
   publisher={Institute of Electrical and Electronics Engineers (IEEE)},
   author={Harrow, A. W. and Winter, A.},
   year={2012},
   pages={1–2} 
   }

@article{Heinosaari_2020,
   title={Communication tasks in operational theories},
   volume={53},
   ISSN={1751-8121},
   DOI={10.1088/1751-8121/abb5dc},
   number={43},
   journal={J. Phys. A: Math. Theor.},
   publisher={IOP Publishing},
   author={Heinosaari, T. and Kerppo, O. and Leppäjärvi, L.},
   year={2020},
   pages={435302} 
   }

@book{Schumacher_Westmoreland_2010,
    place={Cambridge},
    title={Quantum Processes Systems, and Information},
    publisher={Cambridge University Press},
    author={Schumacher, B. and Westmoreland, M.},
    year={2010},
    doi={10.1017/CBO9780511814006}
    }

@article{neumann36,
 ISSN = {0003486X, 19398980},
 doi = {10.2307/1968621},
 author = {Birkhoff, G. and Von Neumann, J.},
 journal = {	Ann. Math.},
 number = {4},
 pages = {823--843},
 publisher = {[Annals of Mathematics, Trustees of Princeton University on Behalf of the Annals of Mathematics, Mathematics Department, Princeton University]},
 title = {The Logic of Quantum Mechanics},
 volume = {37},
 year = {1936}
}

@article{arai2024derivationstandardquantumtheory,
    doi = {10.1088/1367-2630/ad4d18},
    year = {2024},
    publisher = {IOP Publishing},
    volume = {26},
    number = {5},
    pages = {053046},
    author = {Arai, H. and Hayashi, M.},
    title = {Derivation of standard quantum theory via state discrimination},
    journal = {New J. Phys.}
}

@Article{rohrlich1995nonlocalityaxiomquantumtheory,
    author={Popescu, S. and Rohrlich, D.},
    title={Quantum nonlocality as an axiom},
    journal={Found. Phys.},
    year={1994},
    volume={24},
    number={3},
    pages={379-385},
    issn={1572-9516},
    doi={10.1007/BF02058098}
}

@article{Pl_vala_2023,
   title={General probabilistic theories: An introduction},
   volume={1033},
   ISSN={0370-1573},
   DOI={10.1016/j.physrep.2023.09.001},
   journal={	Phys. Rep.},
   publisher={Elsevier BV},
   author={Pl{\'a}vala, M.},
   year={2023},
   pages={1–64} 
   }

@misc{hardy2001quantumtheoryreasonableaxioms,
      title={Quantum Theory From Five Reasonable Axioms}, 
      author={Hardy, L.},
      year={2001},
      eprint={quant-ph/0101012},
      archivePrefix={arXiv},
      primaryClass={quant-ph}
}

@article{Navascu_s_2009,
   title={A glance beyond the quantum model},
   volume={466},
   ISSN={1471-2946},
   DOI={10.1098/rspa.2009.0453},
   number={2115},
   journal={Proc. Roy. Soc. A},
   publisher={The Royal Society},
   author={Navascués, M. and Wunderlich, H.},
   year={2009},
   pages={881–890} 
   }

@article{Janotta_2011,
doi = {10.1088/1367-2630/13/6/063024},
year = {2011},
publisher = {},
volume = {13},
number = {6},
pages = {063024},
author = {Janotta, P. and Gogolin, C. and Barrett, J. and Brunner, N.},
title = {Limits on nonlocal correlations from the structure of the local state space},
journal = {New J. Phys.}
}

@article{Cirelson1980QuantumGO,
  title={Quantum generalizations of Bell's inequality},
  author={Cirel'son, B. S.},
  journal={Lett. Math. Phys.},
  year={1980},
  volume={4},
  pages={93-100},
  doi={10.1007/BF00417500}
}

@article{chsh,
  title = {Proposed Experiment to Test Local Hidden-Variable Theories},
  author = {Clauser, J. F. and Horne, M. A. and Shimony, A. and Holt, R. A.},
  journal = {Phys. Rev. Lett.},
  volume = {23},
  issue = {15},
  pages = {880--884},
  numpages = {0},
  year = {1969},
  publisher = {American Physical Society},
  doi = {10.1103/PhysRevLett.23.880},
}

@article{Heinosaari_2024,
   title={Encoding and decoding of information in general probabilistic theories},
   volume={22},
   pages = {2440007},
   ISSN={1793-6918},
   DOI={10.1142/s0219749924400070},
   number={05},
   journal={Int. J. Quantum Inf.},
   publisher={World Scientific Pub Co Pte Ltd},
   author={Heinosaari, T. and Leppäjärvi, L. and Plávala, M.},
   year={2024}
   }

@article{Bhattacharya_2020,
   title={Nonlocality without entanglement: Quantum theory and beyond},
   volume={2},
   issue = {1},
   pages = {012068},
   ISSN={2643-1564},
   DOI={10.1103/physrevresearch.2.012068},
   number={1},
   journal={Phys. Rev. Research},
   publisher={American Physical Society (APS)},
   author={Bhattacharya, S. S. and Saha, S. and Guha, T. and Banik, M.},
   year={2020}
   }

@article{Bennet_NLWE,
  title = {Quantum nonlocality without entanglement},
  author = {Bennett, C. H. and DiVincenzo, D. P. and Fuchs, C. A. and Mor, T. and Rains, E. and Shor, P. W. and Smolin, J. A. and Wootters, W. K.},
  journal = {Phys. Rev. A},
  volume = {59},
  issue = {2},
  pages = {1070--1091},
  year = {1999},
  publisher = {American Physical Society},
  doi = {10.1103/PhysRevA.59.1070}
}

@article{Heinosaari_2018,
doi = {10.1088/1751-8121/aad1fc},
year = {2018},
publisher = {IOP Publishing},
volume = {51},
number = {36},
pages = {365303},
author = {Heinosaari, T. and Kerppo, O.},
title = {Antidistinguishability of pure quantum states},
journal = {J. Phys. A: Math. Theor.}
}

@article{Caves_antidist,
  title = {Conditions for compatibility of quantum-state assignments},
  author = {Caves, C. M. and Fuchs, C. A. and Schack, R.},
  journal = {Phys. Rev. A},
  volume = {66},
  issue = {6},
  pages = {062111},
  numpages = {11},
  year = {2002},
  publisher = {American Physical Society},
  doi = {10.1103/PhysRevA.66.062111}
}

@article{Chitambar_2014,
   title={Everything You Always Wanted to Know About LOCC (But Were Afraid to Ask)},
   volume={328},
   ISSN={1432-0916},
   DOI={10.1007/s00220-014-1953-9},
   number={1},
   journal={Commun. Math. Phys.},
   publisher={Springer Science and Business Media LLC},
   author={Chitambar, E. and Leung, D. and Mančinska, L. and Ozols, M. and Winter, A.},
   year={2014},
   pages={303–326} 
   }

@article{Chefles01112000,
author = {Chefles, A.},
title = {Quantum state discrimination},
journal = {	Contemp. Phys.},
volume = {41},
number = {6},
pages = {401--424},
year = {2000},
publisher = {Taylor \& Francis},
doi = {10.1080/00107510010002599}
}

@article{barnett2008quantumstatediscrimination,
    author = {Barnett, S. M. and Croke, S.},
    journal = {Adv. Opt. Photon.},
    number = {2},
    pages = {238--278},
    publisher = {Optica Publishing Group},
    title = {Quantum state discrimination},
    volume = {1},
    year = {2009},
    doi = {10.1364/AOP.1.000238}
}

@article{Barnett_2009,
   title={On the conditions for discrimination between quantum states with minimum error},
   volume={42},
   ISSN={1751-8121},
   url={http://dx.doi.org/10.1088/1751-8113/42/6/062001},
   DOI={10.1088/1751-8113/42/6/062001},
   number={6},
   journal={Journal of Physics A: Mathematical and Theoretical},
   publisher={IOP Publishing},
   author={Barnett, Stephen M and Croke, Sarah},
   year={2009},
   month=Jan, pages={062001} }

@article{Bae_2015,
   title={Quantum state discrimination and its applications},
   volume={48},
   ISSN={1751-8121},
   DOI={10.1088/1751-8113/48/8/083001},
   number={8},
   journal={J. Phys. A: Math. Theor.},
   publisher={IOP Publishing},
   author={Bae, J. and Kwek, L.-C.},
   year={2015},
   pages={083001}
   }

@article{Higgins_2011,
   title={Multiple-copy state discrimination: Thinking globally, acting locally},
   volume = {83},
   issue = {5},
   pages = {052314},
   ISSN={1094-1622},
   DOI={10.1103/physreva.83.052314},
   journal={Phys. Rev. A},
   publisher={American Physical Society (APS)},
   author={Higgins, B. L. and Doherty, A. C. and Bartlett, S. D. and Pryde, G. J. and Wiseman, H. M.},
   year={2011}
   }

@article{Audenaert07,
  title = {Discriminating States: The Quantum Chernoff Bound},
  author = {Audenaert, K. M. R. and Calsamiglia, J. and Mu\~noz-Tapia, R. and Bagan, E. and Masanes, Ll. and Acin, A. and Verstraete, F.},
  journal = {Phys. Rev. Lett.},
  volume = {98},
  issue = {16},
  pages = {160501},
  numpages = {4},
  year = {2007},
  publisher = {American Physical Society},
  doi = {10.1103/PhysRevLett.98.160501}
}

@article{Acin01,
  title = {Statistical Distinguishability between Unitary Operations},
  author = {Ac{\'i}n, A.},
  journal = {Phys. Rev. Lett.},
  volume = {87},
  issue = {17},
  pages = {177901},
  numpages = {4},
  year = {2001},
  publisher = {American Physical Society},
  doi = {10.1103/PhysRevLett.87.177901}
}

@misc{roy2026quantumstateexclusioncopies,
      title={Quantum state exclusion with many copies}, 
      author={Roy, D. and  Gupta, T. and  Ghosal, P. and Sen, S. and Bandyopadhyay, S.},
      year={2026},
      eprint={2601.14410},
      archivePrefix={arXiv},
      primaryClass={quant-ph},
}

@article{PhysRevA.65.052315,
  title = {Classification of nonasymptotic bipartite pure-state entanglement transformations},
  author = {Bandyopadhyay, S. and Roychowdhury, V. and Sen, U.},
  journal = {Phys. Rev. A},
  volume = {65},
  issue = {5},
  pages = {052315},
  numpages = {4},
  year = {2002},
  publisher = {American Physical Society},
  doi = {10.1103/PhysRevA.65.052315},
}

@article{Davies:1970cf,
    author = "Davies, E. B. and Lewis, J. T.",
    title = "{An operational approach to quantum probability}",
    doi = "10.1007/BF01647093",
    journal = "Commun. Math. Phys.",
    volume = "17",
    pages = "239--260",
    year = "1970"
}

@misc{cha_2026_20586963,
  author       = {Cha, H.},
  title        = {Double-trine optimality},
  month        = {06},
  year         = {2026},
  publisher    = {Zenodo},
  doi          = {10.5281/zenodo.20586963},
  url          = {https://doi.org/10.5281/zenodo.20586963},
}

\appendix

\addtocontents{toc}{\protect\setcounter{tocdepth}{1}}

\section{Multiple copies of GU-states} 
\label{appendix:GU-f(k)}

\subsection{Proof of \Cref{prop:GU-prop-to-1-n}}

For the proof of \Cref{prop:GU-prop-to-1-n} we need a few mathematical preliminaries, which are going to appear several times in the proof. We collect them in 2 lemmata. We denote the set of even numbers as $2\naturale$ as we often have to distinguish between even and odd numbers.

\begin{lem} \label{lem:prelim-(ii)}
    For $v(\cdot) \in \{\sin(\cdot), \cos(\cdot)\}$ and some $n \in \naturale \setminus \{1\}, \, n > k \in 2\naturale$ holds
    \begin{equation}
        \sum_{j=1}^n v\left(k \frac{\pi j}{n} \right)=0 \, .
    \end{equation}
\end{lem}
\begin{proof}
    Since $\cos(\cdot)$ and $\sin(\cdot)$ form the real and imaginary part of the complex exponential function, we consider first $k=2$ and $\sum_{x=1}^n \exp{[2 \complexi \frac{\pi x}{n}]} = A$.
    The trick is to use the periodicity of the complex exponential function and multiply both sides by $\exp{[2 \complexi \frac{\pi}{n}]}$.
    \begin{gather}
        \exp{\left[2\complexi \frac{\pi}{n}\right]} \cdot A = \exp{\left[2\complexi \frac{\pi}{n}\right]} \sum_{j=1}^n \exp{\left[2 \complexi \frac{\pi j}{n}\right]} \\
        = \exp{\left[2\complexi \frac{\pi}{n}\right]} \left(\exp{\left[2\complexi \frac{\pi}{n}\right]} + \exp{\left[4\complexi \frac{\pi}{n}\right]} + \ldots + \exp{[2 \pi \complexi]}\right) = A
    \end{gather}
    Therefore, either $n=1$, which we exclude, or $A=0$. If $k \in 2 \naturale$, we can repeat above trick by multiplying with $\exp{[k\complexi \frac{\pi}{n}]}$. Then we are left with two options: Either $\frac{k}{n} \equiv 0 \mod 2 \Leftrightarrow \frac{k}{n} = 2t$ for some $t \in \naturale$ or $A=0$. Though, we can exclude the first option because we demanded $n > k$. Since the whole sum evaluates to zero, real and imaginary part must be zero respectively.
\end{proof}
The next lemma already makes use of \Cref{lem:prelim-(ii)} and deals with a type of sum that is about to occur often in the proof of \Cref{prop:GU-prop-to-1-n}.
\begin{lem} \label{lem:prelim-(iii)}
    Let $s \in \naturale$ and $t \in \naturale \setminus 2 \naturale$ such that $s + k \leq 2k$ for some $n >k \in \naturale$. Then, it holds that:
    \begin{equation} \label{eq:tedious-sum}
        \sum_{j=1}^n \cos \left( \frac{\pi j}{n} \right)^s \sin \left( \frac{\pi j}{n} \right)^t = 0
    \end{equation}
\end{lem}
For convenience, from now on we may write $c_j \coloneqq \cos \left( \frac{\pi j}{n} \right)$ and similarly $s_j$.
\begin{proof}
    We divide the proof into two cases.

    (a) $s$ is even: We use the power reduction formula to rewrite \Cref{eq:tedious-sum}.
    \begin{align}
        \sum_{j=1}^n c_j^s \, s_j^t &=\sum_{j=1}^n \left(2^{-s} {s \choose \frac{s}{2}} + 2^{-s+1} \sum_{r=0}^{\frac{s}{2}-1} {s \choose r} \cos \left((s-2r)\frac{\pi j}{n}\right) \right) \times \\ 
        &\times \left(2^{-t+1} \sum_{r'=0}^{\frac{t-1}{2}} (-1)^{\frac{t-1}{2}-r'} {t \choose r'} \cos \left((t-2r')\frac{\pi j}{n}\right) \right) \\
        &= 2^{-s} {s \choose \frac{s}{2}} 2^{-t+1} \sum_{r'=0}^{\frac{t-1}{2}} (-1)^{\frac{t-1}{2}-r'} {t \choose r'} \sum_{j=1}^n \cos \left((t-2r')\frac{\pi j}{n}\right) \\
        &+ 2^{-s+1} 2^{-t+1} \sum_{r'=0}^{\frac{t-1}{2}} (-1)^{\frac{t-1}{2}-r'} {t \choose r'} \sum_{r=0}^{\frac{s}{2}-1} {s \choose r} \sum_{j=1}^n \cos \left((s-2r)\frac{\pi j}{n}\right) \cos \left((t-2r')\frac{\pi j}{n}\right)
    \end{align}
    The first term in above equation equals zero according to \Cref{lem:prelim-(ii)} as $t-2r' < n$. Using the identity
    \begin{equation} \label{eq:product-of-cosines}
        \cos(u) \cos(v) = \frac{\cos(u+v) + \cos(u-v)}{2} \, ,
    \end{equation}
    we can rewrite the sum in the second term.
    \begin{gather}
        \sum_{j=1}^n \cos \left((s-2r)\frac{\pi j}{n}\right) \cos \left((t-2r')\frac{\pi j}{n}\right) \\
        \propto \sum_{j=1}^n \cos \left((s-2r+t-2r')\frac{\pi j}{n}\right)  + \sum_{j=1}^n \cos \left((s-2r-t+2r')\frac{\pi j}{n}\right)
    \end{gather}
    But again, as $s-2r \pm (t -2r') <n$, this term is also equal to zero.
    
    (b)$s$ is odd: Instead of again using the power reduction formula for odd $s$, we use the antisymmetry of $\sin(\cdot)$.
    \begin{align}
        \int_0^n \mathrm{dx} \, c_x ^s \, s_x^t \propto \int_0^\pi \mathrm{d}\theta \cos(\theta)^s \sin(\theta)^t \propto \int_{-\frac{\pi}{2}}^\frac{\pi}{2} \mathrm{d}\vartheta \cos(\vartheta)^t \sin(\vartheta)^s = 0
    \end{align}
    due to an uneven integrand over a symmetric interval. In the steps we used substitutions $\theta = \frac{\pi x}{n}$ and $\vartheta = \theta - \frac{\pi}{2}$. Now, we can just sandwich the sum in question.
    \begin{equation}
        0 = -\int_0^n \mathrm{dx} \, c_x ^s \, s_x^t \leq \sum_{j=0}^n \, c_j^s \, s_j^t = \sum_{j=1}^n \, c_j^s \, s_j^t \leq \int_0^n \mathrm{dx} \, c_x^s \, s_x^t = 0
    \end{equation}
\end{proof}
With these identities, we can show that the success probability for multi-copy state discrimination of GU-states is independent of the number of states in the ensemble, up to a factor of $\frac{1}{n}$, which comes from averaging probabilities of every outcome. 
\begin{prop}
    The success probability of discriminating $k$ copies of a CGU-state from an ensemble of $n > k$ messages is of the form $P_{CGU}(n,k)=\frac{f(k)}{n}$, where $f(k)$ is a number only depending on $k$.
\end{prop}

\begin{proof}
    We consider the formula for the success probability from \cite{Zhou_2025}:
    \begin{equation}
        P_{CGU}(n,k) = \frac{1}{n} [\tr(\sqrt{\varrho^{(k)}})]^2 \, ,\quad \varrho^{(k)} = \frac{1}{n} \sum_{i=1}^n \varrho_i^{\otimes k}
    \end{equation}
    To proof the claim we follow a 3 step strategy:
    \begin{itemize}
        \item[\RomanNumeralCaps{1}] We identify where $\varrho_i^{\otimes k}$ has components consisting of even powers of $\sin{\frac{\pi i}{n}}$ and show that those form the only-non-zero components of $\varrho$. Those are independent of $n$ and thus, $\varrho$ is independent of $n$.
        \item[\RomanNumeralCaps{2}] If $\varrho^{(k)}$ is independent of $n$, so is $\sqrt{\varrho^{(k)}}$.
        \item[\RomanNumeralCaps{3}] We combine this insight with the expression of the success probability.
    \end{itemize}
    
    We start by introducing convenient notation. Note that the expression $(\ket 0 + \ket 1)^{\otimes k}$ gives every number $0,\ldots,2^{k-1}$ in binary representation, which can be proved via induction. The number of 'not-zeros' - in our case 1s -  in a bit string is known as the Hamming weight of that bit string. We use this concept to denote the number of 1's in a vector containing a bit string $\ket{\vec i}=\ket{i_1,\ldots,i_k}$ as $h(\vec i)$. For example, given $\ket{00101}$ we have $h(00101)=2$.
    Given two bit strings $\vec i$ \& $\vec j$, we denote $H(\vec i, \vec j) = 2k - h(\vec i) - h(\vec j)$. As a reminder, we use abbreviations $c_l \coloneqq \cos \left( \frac{\pi l}{n} \right), \, s_l \coloneqq \sin \left( \frac{\pi l}{n} \right)$.

    \RomanNumeralCaps{1} Let us take a closer look at the form of the average state $\varrho^{(k)}$.
    \begin{align}
        \varrho^{(k)} &= \frac{1}{n} \sum_{l=1}^n \varrho_l^{\otimes k} = \frac{1}{n} \sum_{l=1}^n [(c_l \ket 0 + s_l \ket 1)(c_l\bra{0}+s_l\bra 1)]^{\otimes k} \\
        &= \sum_{l=1}^n \sum_{\substack{\vec i \in \{0,1\}^k \\ \vec j \in \{0,1\}^k}} c_l^{H(\vec i, \vec j)} s_l^{h(\vec i) + h(\vec j)} \ket{\vec i} \bra {\vec j}
    \end{align}
    What happened here is that the powers of the trigonometric functions are determined by the number of 1's in the bit string. Since $s_x$ is the coefficient of $\ket 1$, we have as many $s_x$ as we have 1's in $\vec i$ \& $\vec j$. The remaining bits are zero and determine the power of $c_x$.
    It is clear that the powers of $c_x$ and $s_x$ must add up to $k$ for all $\vec i, \, \vec j$.
    Hence, we note
    \begin{equation} \label{eq:range-of-hammings}
        h(\vec i) \in [k] \cup \{0\} , \quad H(\vec i, \vec j) \in [2k] \cup \{0\} \quad \forall \vec i, \, \vec j \in \{0,1\}^k \, .
    \end{equation}
    
    We want to apply \Cref{lem:prelim-(iii)} because we already know that we do not need to consider the terms containing $s_x$ of an uneven exponent. Hence, $h(\vec i)$ \& $h(\vec j)$ must be such that their sum is uneven. We are going to separate the terms in the sums by whether $h(\vec i)$ or $h(\vec j)$ are even or odd. Since all sums are finite, we may change their order.
    \begin{align}
        \varrho &= \sum_{\substack{\vec i, h(\vec i) \notin 2\naturale \\ \vec j, h(\vec j) \in 2\naturale}} \sum_{l=1}^n c_l^{H(\vec i, \vec j)} s_l^{h(\vec i) + h(\vec j)} \ket{\vec i} \bra {\vec j} + \sum_{\substack{\vec i, h(\vec i) \in 2\naturale \\ \vec j, h(\vec j) \notin 2\naturale}} \sum_{l=1}^n c_l^{H(\vec i, \vec j)} s_l^{h(\vec i) + h(\vec j)} \ket{\vec i} \bra {\vec j} \label{eq:avergae-state-full} \\
        &+ \sum_{\substack{\vec i, h(\vec i) \notin 2\naturale \\ \vec j, h(\vec j) \notin 2\naturale}} \sum_{l=1}^n c_l^{H(\vec i, \vec j)} s_l^{h(\vec i) + h(\vec j)} \ket{\vec i} \bra {\vec j} + \sum_{\substack{\vec i, h(\vec i) \in 2\naturale \\ \vec j, h(\vec j) \in 2\naturale}} \sum_{l=1}^n c_l^{H(\vec i, \vec j)} s_l^{h(\vec i) + h(\vec j)} \ket{\vec i} \bra {\vec j} \notag \\
        &= \sum_{\substack{\vec i, h(\vec i) \notin 2\naturale \\ \vec j, h(\vec j) \notin 2\naturale}} \ket{\vec i}  \bra {\vec j} \underbrace{\sum_{l=1}^n c_l^{H(\vec i, \vec j)} s_l^{h(\vec i) + h(\vec j)}}_{\eqqcolon A}  \sum_{\substack{\vec i, h(\vec i) \in 2\naturale \\ \vec j, h(\vec j) \in 2\naturale}}\ket{\vec i}  \bra {\vec j}  \sum_{l=1}^n c_l^{H(\vec i, \vec j)} s_l^{h(\vec i) + h(\vec j)}
    \end{align}
    The first two terms in \Cref{eq:avergae-state-full} vanish according to \Cref{lem:prelim-(iii)}.
    Taking a closer look at the sum $A$, we fix some $\vec i, \, \vec j$ and just write $H$ for $H(\vec i ,\vec j)$ and $h$ for $h(\vec i) + h(\vec j)$ to increase readability. In order to simplify $A$, we first expand the sum with the help of the power reduction formula.
    \begin{align}
        A &= \sum_{l=1}^n \left(\frac{1}{2^H} {H \choose \frac{H}{2}} + \frac{2}{2^H} \sum_{s=0}^{\frac{H}{2}-1} {H \choose s} \cos{\left((H-2s)\frac{\pi l}{n}\right)}\right) \times \\
        &\times \left(\frac{1}{2^h} {h \choose \frac{h}{2}} + \frac{2}{2^h} \sum_{s=0}^{\frac{h}{2}-1} {h \choose s} (-1)^{\frac{h}{2}-s} \cos{\left((h-2s)\frac{\pi l}{n}\right)} \right) \notag \\
        &= \sum_{l=1}^n 2^{-H} 2^{-h} {H \choose \frac{H}{2}} {h \choose \frac{h}{2}} + \frac{2}{2^H 2^h}{h \choose \frac{h}{2}} \sum_{s=0}^{\frac{H}{2}-1} {H \choose s} \underbrace{\sum_{l=1}^n \cos{\left((H-2s)\frac{\pi l}{n}\right)}}_{\overset{\Cref{lem:prelim-(ii)}}{=}0} \\
        &+\frac{2}{2^H 2^h}{H \choose \frac{H}{2}} \sum_{s=0}^{\frac{h}{2}-1} (-1)^{\frac{h}{2}-s} {h \choose s} \underbrace{\sum_{l=1}^n \cos{\left((h-2s)\frac{\pi l}{n}\right)}}_{\overset{\Cref{lem:prelim-(ii)}}{=}0} \notag \\
        &+ \frac{4}{2^H 2^h} \sum_{s=0}^{\frac{h}{2}-1} {h \choose s} (-1)^{\frac{h}{2}-s} \sum_{s=0}^{\frac{H}{2}-1} {H \choose s'} \underbrace{\sum_{l=1}^n \cos{\left((H-2s')\frac{\pi l}{n}\right)} \cos{\left((h-2s)\frac{\pi l}{n}\right)}}_{\overset{\Cref{lem:prelim-(ii)}}{=}0} \notag \\
        &= \frac{n}{2^{2k}} {H \choose \frac{H}{2}} {h \choose \frac{h}{2}} \quad \text{since } H + h = 2k - h(\vec i) - h(\vec j) +h(\vec i) + h(\vec j)  = 2k
    \end{align}
    In the last term we made use of \Cref{eq:product-of-cosines} again to write
    \begin{gather}
        \cos{\left((H-2s')\frac{\pi l}{n} \right)} \cos{\left((h-2s)\frac{\pi l}{n}\right)}  \\ \propto \cos{\left((H-2s'+h-2s)\frac{\pi l}{n} \right)} + \cos{\left((H-2s'-h+2s)\frac{\pi l}{n} \right)} \, .
    \end{gather}
    To make sure \Cref{lem:prelim-(ii)} applies, we need to check that none of the angles in $\cos(\cdot)$ can become too large: $H-2s \leq H \leq 2k$ and similarly for $h$. Since $n > k$, $\frac{2k}{n}$ cannot be a multiple of $2$ and \Cref{lem:prelim-(ii)} applies.
    Moreover: $H + h = 2k$ and $H -h \leq 2k$.
    This yields the form of $\varrho^{(k)}$.
    \begin{equation} \label{eq:average-state-final-form}
        \varrho^{(k)} = \frac{1}{2^k} \sum_{\substack{\vec i, h(\vec i) \notin 2\naturale \\ \vec j, h(\vec j) \notin 2\naturale}} \ket{\vec i} \bra {\vec j} {H(\vec i, \vec j) \choose\frac{H(\vec i, \vec j)}{2}} {h(\vec i) + h(\vec j) \choose \frac{h(\vec i) + h(\vec j)}{2}}+ \frac{1}{2^k} \sum_{\substack{\vec i, h(\vec i) \in 2\naturale \\ \vec j, h(\vec j) \in 2\naturale}} \ket{\vec i} \bra {\vec j} {H(\vec i, \vec j) \choose\frac{H(\vec i, \vec j)}{2}} {h(\vec i) + h(\vec j) \choose \frac{h(\vec i) + h(\vec j)}{2}}
    \end{equation}
    It turns out that this quantity is not easy to compute for $\vec i, \, \vec j \in \{0,1\}^k$ and some fixed $k < n$. However, we can learn something about the structure of $\varrho$ from this (see also \Cref{rem:structure-of-average-state}).
    Most importantly, we observe that the average state $\varrho^{(k)}$ is completely independent of $n$. With this result, the most tedious part is done.

    \RomanNumeralCaps{2} Taking the square root of $\varrho^{(k)}$ cannot somehow create a dependence of $n$. More formally, since $\varrho^{(k)}$ positive semidefinite, it's singular value decomposition becomes the diagonalization of that matrix and it's eigenvalues are independent of $n$ because $\varrho^{(k)}$ is.
    That is, $\varrho^{(k)} = U \Sigma U^*$ for some unitary $U$ and diagonal $\Sigma$, which contains the eigenvalues of $\varrho$. The square root is thus obtained as $\sqrt{\varrho^{(k)}}=U \sqrt{\Sigma} U^*$. As the square root of the eigenvalues of $\varrho^{(k)}$ must be independent of $n$, so must $\sqrt{\varrho^{(k)}}$ be itself.
    In particular, we can define $\tr(\sqrt{\varrho^{(k)}})^2 \eqqcolon f(k)$.

    \RomanNumeralCaps{3} Now, it simply follows that $P_{CGU}(n,k)= \frac{f(k)}{n} \propto \frac{1}{n}$ with some number $f(k) \in \mathds{R}_+$, which is only dependent on the number of copies $k <n$.
\end{proof}
We note that the exact form of $f(k)$ is not clear and the function remains to be a bit vague. However, we recall that it's rough scaling was shown in the main part \Cref{sec:quantum-part}.
Furthermore, the restriction $k < n$ has been only due for purely mathematical reasons as some of the above calculations rely on that fact. The case $n \leq k$ may also be considered, but needs to be investigated separately.

\begin{rem} \label{rem:structure-of-average-state}
    Based on \Cref{eq:average-state-final-form}, we can make the following observation on the structure of $\varrho^{(k)}$. If we index the rows and columns by numbers in binary representation and use $\vec j$ for the columns, we note that a column $\vec j$ of $\varrho^{(k)}$ is identical to a column represented by $\vec j'$ if $h(\vec j) = h(\vec j')$. For instance, $\vec j$ is a permutation of $\vec j'$: $\vec j = \pi(\vec j')$.
    The same holds true for the rows indexed by strings $\vec i$.
    We illustrate this in \Cref{table:structure-of-average-state}.
    \begin{table}[!h]
        \centering
        \begin{tabular}{|c||c|c|c|c|c|c|}
            \multicolumn{7}{c}{} \\ \hline
            index & 000 & 100 & 010 & 110 & 001 & 101 \\
            \hline
            000 & $\alpha$ & 0 & 0 &$\beta$ &0 &$\beta$ \\
            100 & 0 & $\delta$ & $\delta$ &0 &$\delta$ &0 \\
            010 & 0 & $\delta$ & $\delta$ &0 &$\delta$ &0\\
            110 & $\beta$ & 0 & 0 &$\epsilon$ &0 &$\epsilon$\\
            001 & 0 & $\delta$ & $\delta$ &0 &$\delta$ &0\\
            101 & $\beta$ & 0 & 0 &$\epsilon$ &0 &$\epsilon$ \\ \hline
        \end{tabular}
        \caption{The index of a row and column is represented by a binary number. If the number of 1's in a bit string of the column plus the number of 1's in the bit string of the row is odd, the corresponding matrix component is zero. Otherwise it is strictly positive $\alpha,\beta,\delta,\epsilon >0$. Moreover, all entries, that have equal total number of 1's in the bit strings indexing their position, are equal. Consequently, we get an interesting block-structure. }
        \label{table:structure-of-average-state}
    \end{table}
    For example, we have that
    \bea
     \varrho^{(2)} =  \frac{1}{8}\begin{pmatrix}
        3 & 0 & 0& 1\\
        0 & 1 & 1& 0 \\
        0 & 1 & 1& 0 \\
        1 & 0 & 0& 3\\
    \end{pmatrix}, \quad \varrho^{(3)} = \frac{1}{16}
    \begin{pmatrix}
     5 & 0 & 0 & 1 & 0 & 1 & 1 & 0 \\
     0 & 1 & 1 & 0 & 1 & 0 & 0 & 1 \\
     0 & 1 & 1 & 0 & 1 & 0 & 0 & 1 \\
     1 & 0 & 0 & 1 & 0 & 1 & 1 & 0 \\
     0 & 1 & 1 & 0 & 1 & 0 & 0 & 1 \\
     1 & 0 & 0 & 1 & 0 & 1 & 1 & 0 \\
     1 & 0 & 0 & 1 & 0 & 1 & 1 & 0 \\
     0 & 1 & 1 & 0 & 1 & 0 & 0 & 5 \\
    \end{pmatrix}
    \eea
    We see that there is indeed more about the similarity of the actual components, but we are going to leave it with that.
\end{rem}
\begin{rem}
    As stated in \cite{Jozsa_2000}, the average state $\varrho^{(k)}$ and the Gram matrix of the multi-copy ensemble $G^{(k)}$ share the same nonzero eigenvalues. Since $G^{(k)}$ is positive semidefinite, it is clear that (Perron-Frobenius)
    \begin{equation}
        \min_{j \in [n]}\sum_{i=1}^n G^{(k)}_{ij} \leq \lambda_{max}^{(k)} \leq \max_{j \in [n]}\sum_{i=1}^n G^{(k)}_{ij}
    \end{equation}
    holds, whereas $\lambda_{max}^{(k)}$ is the largest eigenvalue of $G^{(k)}$. For even $k$, it turns out that
    \begin{equation}
        \sum_{i=1}^n G_{ij}^{(k)} = \frac{1}{2^{2k}} {k \choose \frac{k}{2}} \quad \forall j \in [n]  \, .
    \end{equation}
    We thus have $\lambda_{max}^{(k)} = \frac{1}{2^{2k}} {k \choose \frac{k}{2}}$ is the largest eigenvalue of $G^{(k)}$ and $\varrho^{(k)}$, which resembles our lower bound $h(k)$. If $k \notin 2 \naturale$, $\lambda_{max}^{(k)}$ is at least upper bounded by that value.
\end{rem}
\subsection{Lemmas for the proof of \Cref{thm:quantum-advantage-by-GU}}

\begin{lem}\label{lem:super-class-bound}
    The upper bound $g(k)$ from \Cref{prop:classical-n-k-bound} is further upper bounded by a function $l(k) = 0.771 \sqrt{\pi k} + 2$ for all $k$.
\end{lem}
\begin{proof}
    For convenience we abbreviate the bound from \Cref{prop:classical-n-k-bound} a bit:
    \begin{equation}
        \frac{1}{n} \left[ 2 + \sum_{j=1}^{k-1} {k \choose j} \left(\frac{j}{k}\right)^j \left(1-\frac{j}{k}\right)^{k-j} \right] \eqqcolon \frac{1}{n} \left[ 2+ \sum_{j=1}^{k-1} w(k,s) \right] \eqqcolon \frac{1}{n} (2 + B(k))
    \end{equation}
    Using Robbins inequality
    \begin{equation} \label{eq:robbins-ineq}
        \sqrt{2 \pi n} \left( \frac{n}{\euler} \right)^n \euler^{\frac{1}{12n+1}} < n! < \sqrt{2\pi n} \left( \frac{n}{\euler} \right)^n \euler^{\frac{1}{12n}} \, ,
    \end{equation}
    we can bound the binomial coefficients by plugging the upper bound in the numerator and the lower bound in the denominator.
    \begin{equation}
        {k \choose s} = \frac{k!}{(k-s)! s!} < \frac{1}{\sqrt{2\pi}} \sqrt{\frac{k}{(k-s)s}} \exp \left[ \frac{1}{12k} - \frac{1}{12(k-s)+1} - \frac{1}{12s+1} \right]
    \end{equation}
    and the exponential can be further upper bounded by $\euler^{\frac{1}{12k}} \eqqcolon \tau(k)$. This yields an upper bound for the argument of the sum $w(k,s)$:
    \begin{align}
        w(k,s) &< \frac{1}{\sqrt{2 \pi}} \sqrt{\frac{k}{(k-s)s}} \left(\frac{s}{k} \right)^s k^k s^{-s} (k-s)^{s-k} \left(\frac{k-s}{k} \right)^{k-s} \tau(k) \\
        &= \frac{1}{\sqrt{2 \pi}} \sqrt{\frac{k}{2\pi(k-s)s}} \tau(k) \eqqcolon w'(k,s)
    \end{align}
    \begin{align}
        B(k) &< \sum_{s=1}^{k-1} w'(k,s) \leq \sqrt{\frac{k}{2\pi}} \tau(k) \int_1^{k-1} \mathrm{ds} \frac{1}{\sqrt{(k-s)s}} \\
        & \propto \int \mathrm{du} \frac{2k \sin(u) \cos(u)}{k \sin(u) \cos(u)} = 2 \int \mathrm{du} = 2 \arcsin{\left (\sqrt{\frac{s}{k}} \right)} \biggr\rvert_{1}^{k-1} \\
        & < 2 \frac{\pi}{2} = \pi
    \end{align}
    The integral was solved using the substitution $s = k \sin^2u \Leftrightarrow u=\arcsin{\left(\sqrt{\frac{s}{k}} \right)}$.
    This leads us to the following upper bound:
    \begin{equation}
        \Pb(n,k) < \frac{1}{n} \left(2+\sqrt{\frac{\pi k}{2}} \tau(k) \right) < \frac{1}{n} (2+0.771 \sqrt{\pi k}) \eqqcolon \frac{l(k)}{n} \, ,
    \end{equation}
    where we used that $\tau(k) < \euler^{\frac{1}{12}}<1.09$ and $\frac{1.09}{\sqrt{2}}<0.771$.
\end{proof}

\begin{figure}[!t]
    \centering
    \includegraphics[width=0.9\linewidth]{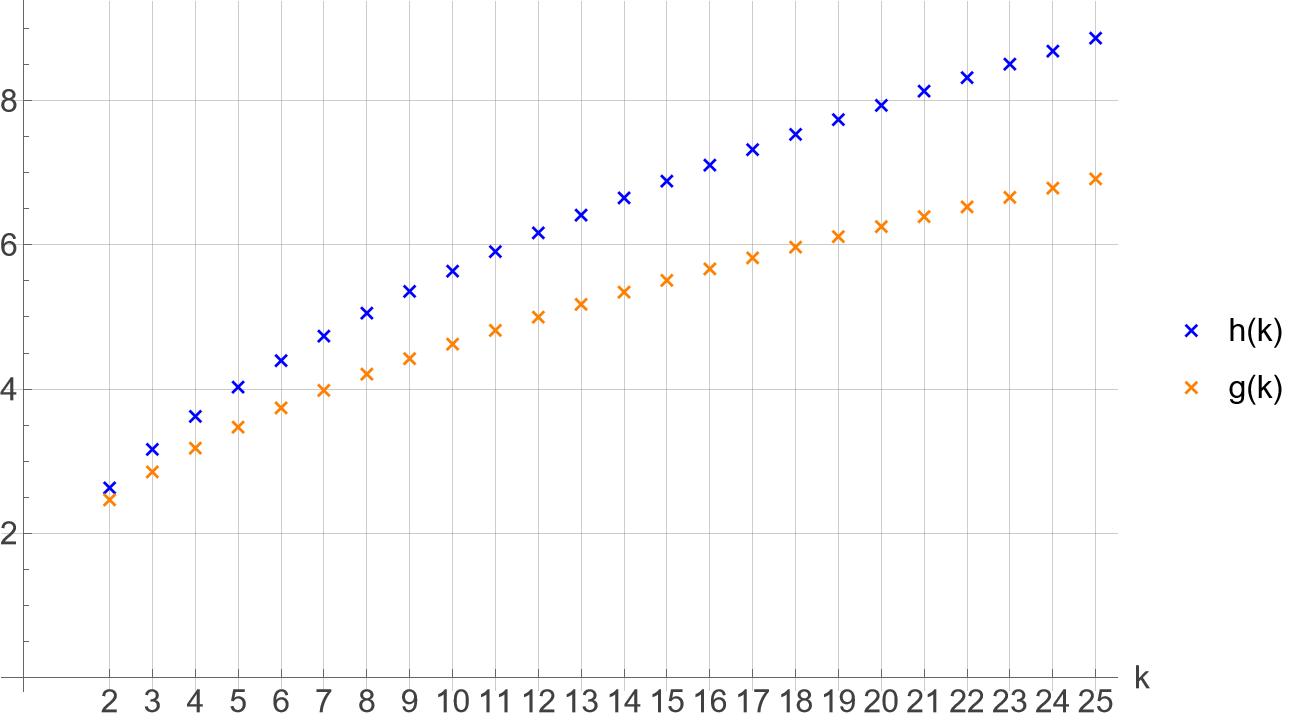}
    \caption{The lower (C)GU-bound $h(k)$ surpasses the classical upper bound $g(k)$ for all $k \in [25]$, whereas the gap increases in alignment with the success probability shown in \Cref{fig:f-vs-g}.}
    \label{fig:h_vs_g}
\end{figure}

\begin{lem}\label{lem:advantage-threshold}
    For all $k \geq 25$ it holds that $h(k) > l(k)$.
\end{lem}
\begin{proof}
    The key is to show that $h(k)$ is strictly bigger than a linear function in $\sqrt{\pi k}$. This achieved by using Robbins inequality \Cref{eq:robbins-ineq} again, but in reverse order: plug the upper bound into the denominator and the lower bound into the numerator.
    \begin{align}
        h(k) &= 2^{2k} {2k \choose k}^{-k} = 2^{2k} \frac{(k!)^2}{(2k)!} \\
        &>2 ^{2k} \left(\sqrt{2k \pi} \left(\frac{k}{\euler} \right)^k \exp{\left[\frac{1}{12k+1} \right]} \right)^2 \left(\sqrt{4k \pi} \left(\frac{2k}{\euler} \right)^{2k} \exp{\left[\frac{1}{12k} \right]} \right)^{-1} \\
        &= \sqrt{\pi k} \exp{\left[\frac{2}{12k+1} - \frac{1}{12k}\right]} > \sqrt{\pi k}
    \end{align}
    The linear function $l(k)$ has a slope of $0.771 < 1$ but also an offset of 2. To see $h(k) > l(k)$ with above estimations, we must choose $k$ sufficiently large.
    \begin{equation}
        2+0.771\sqrt{\pi k} \leq \sqrt{\pi k} \Leftrightarrow k \geq 25
    \end{equation}
\end{proof}

\section{Postprocessing maps in local-fixed strategies}
\label{appendix:post-process-local-fixed}
We show that all seven postprocessing maps cannot lead to a success probability above $\frac{8}{9}$ via analytic optimization, where this is only one map to saturate the bound.
\begin{prop}
    Let 
\begin{align}
    P_j(C):=\frac{1}{3} \sum_{i=1}^3 \sum_{(a,b) \in f_j^{-1}(i)} C_{i a}C_{i b}
\end{align}
for all $j \in [7]$ for all row-stochastic matrices $C$, where $f_j:[3]^3 \rightarrow [3]$ is a postprocessing map. Then $p_j(C) \leq \frac{8}{9}$ holds. 
\end{prop}
\begin{proof}
    As argued in \Cref{sec:fixed-measurement-strategy}, there are only 7 postprocessing maps we have to consider. We recall all 7 in terms of their preimages and separately perform the optimization $ \max_{C: \ \mathfrak{IS}(C) \leq 2 } P_j(C)$ for all $j \in [7]$ analytically.
    We have the following constraints the optimization tasks. Since $C$ is a row-stochastic matrix, we must have $(i) \, \sum_{i=1}^3 C_{ij} = 1 \, \forall j \in [3]$ and the information storability induces $(ii) \, \sum_{i=1}^3 C_{ii} \leq 2$.

    (1)  $f_1^{-1}(1)=\{ 11\}, \quad  f_2^{-1}(2)=\{ 22\}, \quad  f_2^{-1}(3)=\{ 33,12,21,23,32,31,13\}$
    \begin{align}
        P_1(C) &\propto C_{11}^2 + C_{22}^2 + C_{33}^2 + 2(C_{31}C_{32} + C_{31}C_{33} + C_{32}C_{33}) \\
        &=C_{11}^2 + C_{22}^2 + (C_{33}+C_{31}+C_{32})^2-C_{31}^2-C_{32}^2 \\
        &=1+C_{11}^2+C_{22}^2-C_{31}^2-C_{32}^2
    \end{align}
    Note that above function is monotone in $C_{11}$ and $C_{22}$ without any additional constraints. The maximum is therefore achieved when we invoke boundary conditions $C_{11}+C_{22}\leq 2-C_{33}$ and as there is symmetry between $C_{11}$ and $C_{22}$ set
    \begin{equation}
        C_{11}=C_{22}=\frac{2-C_{33}}{2}=\frac{1+C_{31}+C_{32}}{2} \, .
    \end{equation}
    Plugging in we get:
    \begin{equation}
        P_1(C) \leq \frac{3}{2} -\frac{1}{2}\left(C_{31}^2+C_{32}^2 \right) + C_{31}+C_{32}+C_{31}C_{32} \eqqcolon u
    \end{equation}
    Due to the symmetry of $C_{32}$ and $C_{31}$ w.l.o.g.\ we check the monotony.
    \begin{equation}
        \frac{\partial u}{\partial C_{31}} = -C_{31}+1+C_{32} \geq 0
    \end{equation}
    It follows that the function is monotone in both variables and using row stochasticy $C_{32}=1-C_{31}$ we write
    \begin{align}
        u&=\frac{3}{2}-\frac{1}{2}(2C_{31}^2-2C_{31}+1)+C_{31}+1-C_{31}+C_{31}-C_{31}^2 \\
        &=2+2C_{31}(1-C_{31}) \leq 2+\frac{1}{2} =\frac{5}{2} \, ,
    \end{align}
    where the last term reaches its maximum value if $C_{31}=\frac{1}{2}$. It follows that $P_1(C) \leq \frac{5}{6}$.
    For equality we must have $C_{32}=1-C_{31}=\frac{1}{2}$ as well as $C_{33}=0$ and consequently $C_{11}=C_{22}=1$.
    \begin{equation}
        C=\begin{pmatrix}
            1&0&0\\
            0&1&0\\
            \frac{1}{2}&\frac{1}{2}&0
        \end{pmatrix}
        \Rightarrow P_1(C) = \frac{1}{3}\left( 2 + 2\cdot \frac{1}{4}\right)=\frac{5}{6}
    \end{equation}
    
    (2) $f_2^{-1}(1)=\{ 11\}, \quad  f_2^{-1}(2)=\{ 22, 12,21\}, \quad  f_2^{-1}(3)=\{ 33,23,32,31,13\}$
    \begin{align}
         P_2(C) &\propto C_{11}^2 + C_{22}^2 + C_{33}^2 + 2(C_{21}C_{22} + C_{31}C_{33} + C_{32}C_{33}) \\
         &=C_{11}^2+C_{22}(C_{22}+2C_{21})+C_{33}(1+C_{31}+C_{32}) \\
         &\leq C_{11}^2+C_{22}(2-C_{22})+C_{33}(2-C_{33})\\
    \end{align}
    Above function is monotone and symmetric in the last two terms and so we set w.l.o.g.\ $C_{33}=C_{22}$. Then, the second term is for sure monotone in $C_{22}$ and we invoke $C_{11}+2C_{22}\leq 2 \Leftrightarrow C_{22}\leq \frac{2-C_{11}}{2}$. This yields:
    \begin{align}
        P_2(C) &\leq \frac{1}{3}(C_{11}^2+2C_{22}(2-C_{22})) \leq \frac{1}{3}\left(C_{11}^2+(2-C_{11})\left(1+\frac{C_{11}}{2}\right)\right) \\
        &= \frac{1}{3}\left(\frac{1}{2}C_{11}^2+2\right)\leq \frac{5}{6}
    \end{align}
    The following communication matrix saturates this bound:
    \begin{equation}
        C=\begin{pmatrix}
            1&0&0\\
            \frac{1}{2}&\frac{1}{2}&0\\
            \frac{1}{4}&\frac{1}{4}&\frac{1}{2}
        \end{pmatrix}
        \Rightarrow P_2(C)=\frac{1}{3}\left(1+2\cdot \frac{1}{4}+2\left(\frac{1}{4}+\frac{2}{8} \right) \right)=\frac{5}{6}
    \end{equation}

    (3) $f_3^{-1}(1)=\{ 11\}, \quad  f_3^{-1}(2)=\{ 22, 23,32\}, \quad  f_3^{-1}(3)=\{ 33,12,21,31,13\}$
    \begin{align}
         P_3(C) &\propto C_{11}^2 + C_{22}^2 + C_{33}^2 + 2(C_{31}C_{32} + C_{31}C_{33} + C_{22}C_{23}) \\
         &= C_{11}^2 + C_{22}(C_{22}+2C_{23})+C_{33}(C_{33}+2C_{31})+2C_{31}(1-C_{31}-C_{33}) \\
         &\leq C_{11}^2+C_{22}(2-C_{22})+C_{33}^2+\underbrace{2C_{31}(1-C_{31})}_v
    \end{align}
    We see that the first two terms are monotone in $C_{11}$ and $C_{22}$ respectively and we can set them both to 1 to achieve the maximum. This however constraints $C_{33}$, that is, this enforces $C_{33}=0$. But the rest of the function is not monotone in $C_{33}$ as $C_{33}$ is further constrained by $v$. 
    Noting that $v \leq \frac{1}{2}$, $P_3(C) \leq \frac{5}{6}$ follows. A matrix that achieves equality for $f_4$ reads
    \begin{equation}
        C=\begin{pmatrix}
            1&0&0 \\
            0&1&0 \\
            \frac{1}{2}&\frac{1}{2}&0
        \end{pmatrix}
        \Rightarrow P_3(C) =\frac{1}{3} \left(2+\frac{1}{2}\right)=\frac{5}{6} \, .
    \end{equation}

    (4) $f_4^{-1}(1)=\{ 11\}, \quad  f_4^{-1}(2)=\{ 22, 31,13\}, \quad f_4^{-1}(3)=\{33,12,21, 23, 32\}$
    \begin{align}
        P_4(C) &\propto C_{11}^2 + C_{22}^2 + C_{33}^2 + 2(C_{31}C_{32} + C_{21}C_{23} + C_{32}C_{33}) \\
        &=C_{11}^2+C_{22}^2+2C_{21}C_{23}+C_{33}(C_{33}+2C_{32})+2C_{31}C_{32} \\
        &=C_{11}^2+C_{22}^2+2C_{21}(1-C_{22}-C_{21})+C_{33}^2 + 2C_{32}(1-C_{32})\\
        &\leq 1+C_{22}^2+2C_{21}(1-C_{22}-C_{21})+C_{33}^2 + 2C_{32}(1-C_{32})
    \end{align}
    Upon realizing that the function is monotone in $C_{11}$ and thus setting $C_{11}=1$, we need to check all the boundary cases one by one.
    \begin{align}
        C_{22}=0 \Rightarrow &u=1+2C_{21}(1-C_{21})+C_{33}^2+2C_{23}(1-C_{32})\\
        &\leq 1+2C_{21}(1-C_{21})+1-C_{32}^2\leq 2+2C_{21}(1-C_{21}) \\
        &\leq \frac{5}{2} \Rightarrow P_4(C) \leq \frac{5}{6} \\
        C_{22}=1 \Rightarrow &C_{21}=C_{33}=0 \Rightarrow u=2+2C_{32}(1-C_{32}) \leq \frac{5}{2}\Rightarrow P_4(C)\leq \frac{5}{6}\\
        C_{21}=0 \Rightarrow &u=1+C_{22}^2+C_{33}^2+2C_{32}(1-C_{32})\leq 2+\frac{1}{2}=\frac{5}{2} \Rightarrow P_4(C)=\frac{5}{6}\\
        C_{21}=1 \Rightarrow &C_{22}=0 \Rightarrow u = 1+C_{33}^2+2C_{32}(1-C_{32}) \leq 2 \Rightarrow P_4(C) \leq \frac{2}{3}\\
        C_{33}=1 \Rightarrow &C_{32}=C_{22}=0 \Rightarrow u \leq 2+\frac{1}{2} \Rightarrow P_4(C) \leq \frac{5}{6}\\
        C_{33}=0 \Rightarrow &1+\underbrace{C_{22}^2+2C_{21}(1-C_{22}-C_{21})}_{\eqqcolon \tilde u}+2C_{32}(1-C_{32})
    \end{align}
    Now we need to find the maximum value of $\tilde u$.
    \begin{align}
        \frac{\partial \tilde u}{\partial C_{21}} = 0 &\Leftrightarrow C_{21}=\frac{1-C_{22}}{2}\Rightarrow \tilde u = C_{22}^2 +(1-C_{22})\frac{1-C_{22}}{2} \leq 1 \\
        &\Rightarrow u \leq 1+1+\frac{1}{2} = \frac{5}{2}\Rightarrow P_4(C) \leq \frac{5}{6}\\
        \frac{\partial \tilde u}{\partial C_{22}} = 0 &\Leftrightarrow C_{22}=C_{21} \Rightarrow u = 1+C_{22}^2+2C_{22}(1-2C_{22})+C_{33}^2+2C_{32}(1-C_{32}) \\
        &\leq 1+\frac{1}{3}+1 = \frac{7}{3} \Rightarrow P_6(C) \leq \frac{7}{9}
    \end{align}
    It remains to look at the local extrema of $u$. Note that $ \frac{\partial \tilde u}{\partial C_{21}} =  \frac{u}{\partial C_{21}}$ and $\frac{\partial \tilde u}{\partial C_{22}}=\frac{\partial u}{\partial C_{22}}$. So there are no better bounds to be found here. Furthermore examining $ \frac{\partial u}{\partial C_{33}}=0$ gives back the already considered boundary condition $C_{33}=0$.
    We are thus left with the following:
    \begin{align}
        &\frac{\partial u}{\partial C_{32}}=0 \Leftrightarrow C_{32}=\frac{1}{2} \Rightarrow C_{33} \leq \frac{1}{2} \Rightarrow C_{22} + C_{33} \leq C_{22}+\frac{1}{2}\leq 1 \Leftrightarrow C_{22} \leq \frac{1}{2} \\
        &\Rightarrow u \leq 1+ \frac{1}{4}+\frac{1}{2}+C_{22}^2+2C_{21}(1-C_{22}-C_{23})\leq \frac{7}{4}+\frac{1}{2}=\frac{9}{4} \Rightarrow P_4(C) \leq \frac{9}{12}
    \end{align}
    To show tightness of the best upper bound $\frac{5}{6}$ we write down the following communication matrix.
    \begin{equation}
        C=\begin{pmatrix}
            1&0&0\\
            \frac{1}{2}&0&\frac{1}{2}\\
            0&0&1
        \end{pmatrix}
        \Rightarrow P_4(C) = \frac{1}{3} \left(2+\frac{2}{4} \right)= \frac{5}{6}
    \end{equation}
    
    (5) $f_5^{-1}(1)=\{ 11, 23, 32\}, \quad  f_5^{-1}(2)=\{   22, 31, 13\}, \quad  f_5^{-1}(3)=\{ 33, 12, 21\}$
    \begin{align}
         P_5(C) &\propto C_{11}^2 + C_{22}^2 + C_{33}^2 + 2(C_{12}C_{13} + C_{21}C_{23} + C_{32}C_{31}) \\
         &=C_{11}^2+2C_{21}(1-C_{12}-C_{11})+C_{22}^2+2C_{21}(1-C_{22}-C_{21})+C_{33}^2\\
         &+2C_{31}(1-C_{31}-C_{33}) \notag\\
         &=C_{11}(C_{11}-2C_{12})+C_{12}(1-C_{12})+C_{22}(C_{22}-2C_{21})+2C_{21}(1-C_{21})\\
         &+C_{33}(C_{33}-2C_{31})+2C_{31}(1-C_{31}) \notag\\
         &\leq (1-C_{21})(1-3C_{12})+2C_{12}(1-C_{12})+(1-C_{21})(1-3C_{21}) \\
         &+2C_{21}(1-C_{21})(1-3C_{21})+2C_{31}(1-C_{31}) \notag \\
         &=(1-C_{12})^2 + (1-C_{21})^2+(1-C_{31})^2
    \end{align}
    We notice the symmetry in $C_{12}, \, C_{21}$ and $C_{31}$, so w.l.o.g.\ we set $C_{12}=0$ for the first boundary case. Then we have:
    \begin{equation}
        P_5(C) \leq \frac{1}{3}(1+(1-C_{21})^2+(1-C_{31})^2) \eqqcolon  u
    \end{equation}
    By setting $C_{21}=C_{31}=0$ we could achieve $u=1$. However, this can not yield the optimum since then
    \begin{equation}
        P_5(C) \leq \frac{1}{3}(C_{11}^2 + C_{22}^2 + C_{33}^2) \leq \frac{1}{3}(C_{11} + C_{22} + C_{33}) \leq \frac{2}{3}\, .
    \end{equation}
    How to choose $C_{31}$ optimally given $C_{21}=C_{12}=0$ ? With above boundaries we have
    \begin{equation}
        P_5(C)\leq \frac{1}{3}(2+2C_{31}C_{32}) \leq \frac{1}{3}\left(2+\frac{1}{2}\right) = \frac{5}{6} \, .
    \end{equation}
    It is clear that w.l.o.g.\ setting $C_{12}=1$ cannot yield the optimum value in above quantity and would lead to $P_5(C) \leq \frac{2}{3}$. Finally we check the local extrema w.l.o.g.\ with respect to $C_{12}$.
    \begin{align}
        \frac{\partial u }{\partial C_{12}} = -2(1-C_{12})=0 \Leftrightarrow C_{12} = 1 \Rightarrow C_{21}=0 \text{ and } P_5(C) \leq \frac{2}{3}
    \end{align}
    The communication matrix to achieve the maximum is:
    \begin{equation}
        C=\begin{pmatrix}
            1&0&0\\
            0&1&0\\
            \frac{1}{2}&\frac{1}{2}&0
        \end{pmatrix}
        \Rightarrow P_7(C) = \frac{1}{3}\left(2+2\cdot\frac{1}{4}\right) = \frac{5}{6}
    \end{equation}

    (6) $f_6^{-1}(1)=\{ 11, 12, 21\}, \quad  f_6^{-1}(2)=\{   22, 31, 13\}, \quad  f_6^{-1}(3)=\{ 33, 23, 32\}$
    \begin{align}
         P_6(C) &\propto C_{11}^2 + C_{22}^2 + C_{33}^2 + 2(C_{11}C_{12} + C_{21}C_{23} + C_{32}C_{33}) \\
         &= C_{11}(C_{11}+2C_{12}) + C_{33}(C_{33}+2C_{32}) + C_{22}^2 + 2C_{21}C_{23} \\
         &\leq C_{11}(2-C_{11})+C_{33}(2-C_{33})+C_{22}^2+2C_{21}(1-C_{21}-C_{22})
       \end{align}
    The first two terms are symmetric and monotone in $C_{11}$ and $C_{33}$ respectively. We thus set $C_{11}=C_{33}=\frac{1-C_{22}}{2}$.
    \begin{align}
        P_6(C) &\leq \frac{1}{3}\left((2-C_{22})\left(1+\frac{C_{22}}{2} \right) +C_{22}(C_{22}-2C_{23})+2C_{23}(1-C_{23})\right) \\
        &\propto2+\underbrace{C_{22}\left(\frac{C_{22}}{2}-2C_{23}\right)+2C_{23}(1-C_{23})}_{\eqqcolon u}
    \end{align}
    Let us look at the maximum values of $u$:
    \begin{align}
        &\frac{\partial u}{\partial C_{22}}=0 \Leftrightarrow C_{22}=C_{23} \Rightarrow u \leq -3C_{23}^2+2C_{23} \leq \frac{1}{3} \Rightarrow P_6(C) \leq \frac{7}{9} \\
        &\frac{\partial u}{\partial C_{23}}=0 \Leftrightarrow C_{23}=\frac{1-C_{22}}{2} \Rightarrow u \leq C_{22}^2-C_{22} + \frac{1}{2} \leq \frac{1}{2} \Rightarrow P_6(C) \leq \frac{5}{6}
    \end{align}
    To achieve the maximum value in the later case the communication matrix reads:
    \begin{equation}
        C=\begin{pmatrix}
            1&0&0\\
            \frac{1}{2}&0&\frac{1}{2}\\
            0&0&1
        \end{pmatrix}
        \Rightarrow P_6(C) = \frac{1}{3}\left(2+2\cdot\frac{1}{4}\right) = \frac{5}{6}
    \end{equation}
    
    (7) $f_7^{-1}(1)=\{ 11, 12, 21\}, \quad  f_7^{-1}(2)=\{22, 23, 32\}, \quad  f_7^{-1}(3)=\{ 33, 31, 13\}$
    \begin{align}
        P_7(C) &\propto C_{11}^2 + C_{22}^2 + C_{33}^2 + 2(C_{11}C_{12} + C_{22}C_{23} + C_{33}C_{31}) \\
        &= C_{11}(C_{11}+2C_{12})+C_{22}(C_{22}+2C_{23})+C_{33}(C_{33}+2C_{31})\\
        &\leq C_{11}(2-C_{11})+C_{22}(2-C_{22})+C_{33}(2-C_{33}) \\
        &\leq C_{11}(2-C_{11})+C_{22}(2-C_{22}) + (2-C_{11}-C_{22})(C_{11}+C_{22}) \\
        &=2(C_{11}(2-C_{11})+C_{22}(2-C_{22})-C_{22}C_{11}) \coloneqq u
    \end{align}
    Notice that $u$ is a function symmetric in $C_{11}$ and $C_{22}$. Since there is no bias between $C_{11}$ and $C_{22}$ w.l.o.g.\ we set $C_{22} = C_{11}$ and write
    \begin{equation}
        u \propto 2C_{11}(2-C_{11})-C_{11}^2 = -3C_{11}^2+4C_{11} \, ,
    \end{equation}
    which is maximized if $C_{11}=\frac{2}{3}$ such that $P_7(C) \leq \left(\frac{2}{3}\right)^2=\frac{8}{9}$. Consequently, $C_{22}=C_{33}=\frac{2}{3}$ and to maximize the mixed terms we invoke $C_{12}=1-C_{11}=\frac{1}{3}$ and similarly $C_{23}=C_{31}=\frac{1}{3}$. In conclusion we then have:
    \begin{equation}
        C= \frac{1}{3} \begin{pmatrix}
            2&1&0\\ 0&2&1 \\ 1&0&2
        \end{pmatrix}
        \Rightarrow P_7(C) = \frac{1}{3}\left(3\left(\frac{2}{3}\right)^2+2\cdot\frac{1}{3}\cdot\frac{2}{3}\cdot3 \right) = \frac{1}{3}\cdot\frac{2\cdot4}{3} =\frac{8}{9}
    \end{equation}
\end{proof}

\section{$SEP$ and $GLOBAL$ optimization programs in polygons}
\label{appendix:polygon-theories}

For a regular polygon $\polyspace$ with $m$ pure states (vertices) $\{s_i\}_{i=1}^m$, the following is an optimization program for multi-copy states discrimination in the case of $k=2$ copies of $n$ states for both separable and global strategies. 

In polygon state space $\polyspace$, for $m\geq n$, let us fix some $n$ pure states $S=\{s_{i_x}\}_{x=1}^n \subset \{s_i\}_{i=1}^m$ for some $i_x \in [m]$ for all $x \in [n]$. The effects $M_x$ of any $n$-outcome measurement $M$ on the compound system $\polyspace \otimes \polyspace$ can expressed as 
\begin{equation}\label{eq:poly-meas}
    M_x = \sum_{i,j=1}^m \alpha^{(x)}_{ij} g_i \otimes g_j,
\end{equation}
where $g_i =e_i$ if $m$ is even and $g_i =f_i$ if $m$ odd for some $\alpha^{(x)}_{ij} \in \real$ such that $\sum_{x=1}^n M_x = \one_{\polyspace} \otimes \one_{\polyspace} = \one_{\polyspace \otimes \polyspace}$. For separable measurements the positivity condition is simply $\alpha^{(x)}_{ij} \geq 0$ while for most general global measurements we require that $M_x(s_i \otimes s_j) \geq 0$\footnote{This condition corresponds to the choice of \emph{minimal tensor product} for the joint state space $\polyspace \otimes_{\min} \polyspace$ which contains only convex combinations of product states. By duality then the effect space corresponds to taking all effects that are positive on all product states resulting in the least restrictive effect space.} for all  $x \in [n]$ and $ i,j \in [m]$. Thus, for the pure states $S=\{s_{i_x}\}_{x=1}^n$ when we optimize over all $SEP$ and $GLOBAL$ measurement strategies we have the following programs:\\
\begin{minipage}{.49\textwidth}
    \begin{align}
        SEP: \\
        \quad \max_{M} \quad & \frac{1}{n} \sum_{x=1}^n M_x(s_{i_x}) \\
        \mathrm{s.t.} \quad & M_x = \sum_{i,j=1}^m \alpha^{(x)}_{ij} g_i \otimes g_j, \\
        & 0 \leq \alpha^{(x)}_{ij} \leq 1  \\ &\forall x \in [n], \forall i,j \in [m],\\
        & \sum_{x=1}^n M_x = \one_{\polyspace \otimes \polyspace} 
    \end{align}
\end{minipage}
\hfill
\begin{minipage}{.49\textwidth}
    \begin{align}
        GLOBAL:\\
        \quad \max_{M} \quad & \frac{1}{n} \sum_{x=1}^n M_x(s_{i_x}) \\
        \mathrm{s.t.} \quad & M_x = \sum_{i,j=1}^m \alpha^{(x)}_{ij} g_i \otimes g_j, \\
        & M_x(s_i \otimes s_j) \geq 0  \\ &\forall x \in [n], \forall i,j \in [m],\\
        & \sum_{x=1}^n M_x = \one_{\polyspace \otimes \polyspace} 
    \end{align}
\end{minipage}\\

Thus, the optimization over all $GLOBAL$ and $SEP$ measurements for all pure states is the following: By running the above programs for all the different $n$-element subsets of all the $m$ pure states we can find out the optimal separable and global measurement strategies when using pure states. In the main text we considered 2-copy state discrimination of ensembles of 3 pure states. Our findings can bee seen in \Cref{fig:global-strategies-in-polygons}.

\end{document}